\documentclass[11pt]{article}

\usepackage{iftex}
\ifPDFTeX
  \usepackage[utf8]{inputenc}
  \usepackage[T1]{fontenc}
\fi

\usepackage[margin=1in]{geometry}
\usepackage{microtype}

\usepackage{amsmath,amssymb,amsfonts,amsthm}
\usepackage{mathtools}
\usepackage{braket}
\usepackage{stmaryrd}
\usepackage{siunitx}
\usepackage{booktabs}
\usepackage{xcolor} % Allows for more color options
\usepackage{aliascnt}
\usepackage[
    colorlinks=true, 
    citecolor=blue, 
    linkcolor=violet
]{hyperref}

\usepackage{cleveref}

\usepackage{etoolbox} % for \ifblank
\usepackage{xparse}   % for \NewDocumentCommand
\usepackage{graphicx}
\usepackage{authblk}
\usepackage{orcidlink}

\newcommand{\todo}[1]{}

\newtheorem{theorem}{Theorem}[section]

\newaliascnt{lemma}{theorem}
\newtheorem{lemma}[lemma]{Lemma}
\aliascntresetthe{lemma}

\newaliascnt{corollary}{theorem}
\newtheorem{corollary}[corollary]{Corollary}
\aliascntresetthe{corollary}

\newaliascnt{claim}{theorem}

\aliascntresetthe{claim}

\newaliascnt{construction}{theorem}
\newtheorem{construction}[construction]{Construction}
\aliascntresetthe{construction}

\theoremstyle{definition}
\newaliascnt{definition}{theorem}
\newtheorem{definition}[definition]{Definition}
\aliascntresetthe{definition}

\crefname{theorem}{theorem}{theorems}
\Crefname{theorem}{Theorem}{Theorems}

\crefname{lemma}{lemma}{lemmas}
\Crefname{lemma}{Lemma}{Lemmas}

\crefname{corollary}{corollary}{corollaries}
\Crefname{corollary}{Corollary}{Corollaries}

\crefname{claim}{claim}{claims}
\Crefname{claim}{Claim}{Claims}

\crefname{construction}{construction}{constructions}
\Crefname{construction}{Construction}{Constructions}

\crefname{definition}{definition}{definitions}
\Crefname{definition}{Definition}{Definitions}

\crefname{appendix}{appendix}{appendices}
\Crefname{appendix}{Appendix}{Appendices}

\providecommand{\doi}[1]{\href{https://doi.org/#1}{\nolinkurl{https://doi.org/#1}}}

\newcommand{\Z}{\mathbb{Z}}
\newcommand{\C}{\mathbb{C}}

\newcommand{\bit}{\{0,1\}}
\newcommand{\Mod}[1]{\ (\mathrm{mod}\ #1)}
\DeclareMathOperator{\poly}{poly}
\DeclareMathOperator{\negl}{negl}

\renewcommand{\vec}[1]{\mathbf{#1}}
\renewcommand{\Mod}[1]{\ (\mathrm{mod}\ #1)}

\newcommand{\br}[1]{\left( #1 \right)}

\renewcommand{\exp}[1]{\text{exp}\br{#1}} % Exponentiation

\newcommand{\bounds}[2]{\bigg\rvert_{#1}^{#2}}
\newcommand{\boudns}[2]{\bounds} % To correct my silly spelling errors

\renewcommand{\a}{\alpha}

\renewcommand{\d}{\delta}

\newcommand{\m}{\mu}

\newcommand{\Th}{\Theta}
\newcommand{\w}{\omega}
\newcommand{\W}{\Omega}

\newcommand{\CA}{\mathcal{A}}
\newcommand{\CB}{\mathcal{B}}
\newcommand{\CC}{\mathcal{C}}
\newcommand{\CD}{\mathcal{D}}

\newcommand{\CF}{\mathcal{F}}

\newcommand{\CI}{\mathcal{I}}

\newcommand{\CO}{\mathcal{O}}
\newcommand{\CP}{\mathcal{P}}

\newcommand{\CR}{\mathcal{R}}
\newcommand{\CS}{\mathcal{S}}

\newcommand{\CW}{\mathcal{W}}
\newcommand{\Ber}{\mathsf{Ber}}
\newcommand{\Bin}{\mathsf{Bin}}

\NewDocumentCommand{\lsn}{ O{} O{} O{} }{%
  \ensuremath{%
    \mathsf{LSN}%
    \ifblank{#2}{}{^{#2}}% 
    \ifblank{#3}{}{_{#3}}  
    \ifblank{#1}{}{\!\left(#1\right)}%   
  }%
}

\NewDocumentCommand{\slsn}{ O{} O{} O{} }{%
  \ensuremath{%
    \mathsf{stateLSN}%
    \ifblank{#2}{}{^{#2}}% 
    \ifblank{#3}{}{_{#3}}  
    \ifblank{#1}{}{\!\left(#1\right)}%   
  }%
}

\NewDocumentCommand{\lpn}{ O{} }{%
  \ensuremath{%
    \mathsf{LPN}%
    \ifblank{#1}{}{\!\left(#1\right)}%   
  }%
}

\NewDocumentCommand{\symplpn}{ O{} }{%
  \ensuremath{%
    \mathsf{sympLPN}%
    \ifblank{#1}{}{\!\left(#1\right)}%   
  }%
}

\NewDocumentCommand{\owff}{ O{} }{%
  \ensuremath{%
    \mathsf{OWFF}%
    \ifblank{#1}{}{\!\left(#1\right)}%   
  }%
}

\newcommand{\sk}{\mathsf{sk}}
\newcommand{\pk}{\mathsf{pk}}
\newcommand{\gen}{\mathsf{Gen}}
\newcommand{\sample}{\mathsf{Sample}}
\newcommand{\eval}{\mathsf{Eval}}
\newcommand{\enc}{\mathsf{Enc}}
\newcommand{\dec}{\mathsf{Dec}}

\DeclareMathOperator{\Symp}{Symp}

\makeatletter
\renewcommand{\paragraph}{%
  \@startsection{paragraph}{4}%
  {\z@}{2.25ex \@plus 1ex \@minus .2ex}{-1em}%
  {\normalfont\normalsize\bfseries}%
}
\makeatother
\title{Post-Quantum Cryptography from Quantum Stabilizer Decoding}

\author[1]{Jonathan Z. Lu\footnote{\href{mailto:lujz@mit.edu}{\texttt{lujz@mit.edu}}}}
\author[2]{Alexander Poremba\footnote{\href{mailto:poremba@bu.edu}{\texttt{poremba@bu.edu}}}}
\author[3]{Yihui Quek\footnote{\href{mailto:yihui.quek@epfl.ch}{\texttt{yihui.quek@epfl.ch}}}}
\author[4]{Akshar Ramkumar\footnote{\href{mailto:aramkuma@caltech.edu}{\texttt{aramkuma@caltech.edu}}}}

\affil[1]{Massachusetts Institute of Technology, Cambridge, MA}
\affil[2]{Boston University, Boston, MA} 
\affil[3]{École Polytechnique Fédérale de Lausanne, Lausanne, Switzerland}
\affil[4]{California Institute of Technology, Pasadena, CA}

\date{\today}

\begin{document}

\maketitle

\begin{abstract}
Post-quantum cryptography currently rests on a small number of hardness assumptions, posing significant risks should any one of them be compromised.
This vulnerability motivates the search for new and cryptographically versatile assumptions that make a convincing case for quantum hardness.

\qquad In this work, we argue that decoding random quantum stabilizer codes---a quantum analog of the well-studied $\lpn$ problem---is an excellent candidate.
This task occupies a unique middle ground: it is inherently native to \emph{quantum} computation, yet admits an equivalent formulation with purely \emph{classical} input and output, as recently shown by Khesin \textit{et al.} (STOC~'26).
We prove that the average-case hardness of quantum stabilizer decoding implies the core primitives of classical \textsf{Cryptomania}, including public-key encryption (\textsf{PKE}) and oblivious transfer (\textsf{OT}), as well as one-way functions.
Our constructions are moreover practically efficient: our \textsf{PKE} scheme achieves essentially the same efficiency as state-of-the-art $\lpn$-based \textsf{PKE}.
We also provide substantial evidence that stabilizer decoding does not reduce to $\lpn$, suggesting that the former problem constitutes a genuinely new post-quantum assumption.

\qquad Our primary technical contributions are twofold.
First, we give a reduction from random quantum stabilizer decoding to an average-case problem closely resembling $\lpn$, but which is equipped with additional symplectic algebraic structure.
While this structure is essential to the quantum nature of the problem, it raises significant barriers to cryptographic security reductions.
Second, we develop a new suite of scrambling techniques for such structured linear spaces, and use them to produce rigorous security proofs for all of our constructions.
\end{abstract}

\section{Introduction}

The remarkable success of modern cryptography rests on a surprisingly small number of computational hardness assumptions. 
Over the past four decades, these assumptions—ranging from the difficulty of factoring~\cite{10.1145/359340.359342} to the hardness of discrete logarithms~\cite{1055638,10.1145/359460.359473} and various lattice problems~\cite{Hoffstein1998NTRU,repec:spr:sprchp:978-3-540-88702-7_5}—have underpinned the construction of essentially all known cryptographic primitives.
The advent of quantum computing, however, has forced a fundamental re-evaluation of this foundation. 
Shor's algorithm~\cite{365700} and its variants~\cite{regev2024efficientquantumfactoringalgorithm} show that many of the assumptions that underlie classical public-key cryptography are vulnerable to efficient quantum attacks, rendering large swathes of existing cryptographic infrastructure insecure in a quantum world.

In response, the cryptographic community has turned to an even smaller suite of candidate \emph{post-quantum} assumptions~\cite{bernstein2017postquantum,214441,1237791}, primarily in the world of lattices~\cite{repec:spr:sprchp:978-3-540-88702-7_5,10.1145/1568318.1568324}, codes~\cite{pietrzak2012lpn,BCL19}, multivariate polynomial systems~\cite{921316} and isogenies of elliptic curves~\cite{cryptoeprint:2018/383}. 
Among these assumptions, the \emph{Learning with Errors} ($\mathsf{LWE}$) problem~\cite{10.1145/1568318.1568324} and its binary predecessor, \emph{Learning Parity with Noise} ($\lpn$)~\cite{BFKL93,pietrzak2012lpn}, have emerged as a central basis of hardness. 
These problems enjoy worst-case to average-case reductions~\cite{10.1145/1568318.1568324,brakerski2019worst}, a rich algebraic structure, and a remarkable versatility in supporting a wide array of cryptographic constructions, including public-key encryption~\cite{10.1145/1568318.1568324,Alekhnovich03}, digital signatures~\cite{pietrzak2012lpn}, oblivious transfer~\cite{cryptoeprint:2019/448}, general secure multi-party computation~\cite{cryptoeprint:2019/448}, collision-resistant hashing~\cite{applebaum_et_al:LIPIcs.ITCS.2017.7} and, in some cases, even homomorphic encryption~\cite{10.5555/1834954,doi:10.1137/120868669,cryptoeprint:2024/1760}.

Yet, despite their prominence, our increasing reliance on $\mathsf{LWE}$ and $\mathsf{LPN}$ as a foundation for post-quantum security is a reason for concern. The history of cryptography has taught us that even long-standing hardness assumptions can be susceptible to unexpected algorithmic breakthroughs, as in the case of factoring and discrete logarithms~\cite{365700}. 
In fact, even ``quantum-safe'' assumptions have recently experienced devastating 
classical attacks, as in the case of isogenies~\cite{cryptoeprint:2022/1038} and multivariate quadratics~\cite{10.1007/978-3-031-15979-4_16} that were initially believed to be secure.
While $\mathsf{LWE}$ and $\mathsf{LPN}$ have so far resisted quantum attacks, there is reason to be cautious about
their long-term viability: lattice problems are intimately connected to the dihedral hidden subgroup problem~\cite{regev2003quantumcomputationlatticeproblems}, which is known to admit subexponential-time quantum algorithms~\cite{kuperberg2004subexponentialtimequantumalgorithmdihedral}; moreover, recent years have also seen renewed efforts targeting lattices~\cite{eldar2016efficientquantumalgorithmvariant,cryptoeprint:2024/555,eldar2022efficientquantumalgorithmlattice,cryptography4010010,PhysRevA.99.032314} and codes~\cite{eldar2023efficientquantumdecoderprimepower} via quantum attacks that, while not yet fully successful, suggest that the landscape is far from settled. 
This concern is further amplified by the fact that most $\mathsf{LWE}$/$\mathsf{LPN}$-based constructions operate under special parameter regimes that are much less understood~\cite{damgaard2012practical,applebaum_et_al:LIPIcs.ITCS.2017.7}, and where worst-to-average-case reductions often
do not apply~\cite{brakerski2019worst,cryptoeprint:2020/870}. 
These vulnerabilities have recently motivated the search for alternative
$\mathsf{LWE}$/$\mathsf{LPN}$-like noisy linear-algebraic assumptions which are less susceptible to attacks than existing ones, and yet still suffice for public-key encryption~\cite{10.1007/978-3-031-91124-8_3}.

More broadly, many standard post-quantum assumptions are rooted in manifestly \emph{classical} problems in mathematics---lattices, codes, or algebraic structures---whose relation to quantum computation is far from direct. 
Because of this disconnect, it is unlikely that even a major breakthrough undermining all of today’s leading post-quantum assumptions would have a substantial impact on fundamental tasks in quantum information science. 
Indeed, as recent work~\cite{bostanci_et_al:LIPIcs.ITCS.2026.24,https://doi.org/10.4230/lipics.tqc.2021.2} suggests, central quantum tasks such as compressing quantum information, decoding noisy quantum channels and other local entanglement transformations lie outside of classical cryptography altogether, and may be hard even if $\mathsf{P}=\mathsf{NP}$. 
This disparity has led to the design of  ``fully quantum'' cryptography, collectively known as \emph{MicroCrypt}~\cite{microcryptzoo,ananth2022cryptographypseudorandomquantumstates,morimae_et_al:LIPIcs.TQC.2024.4,metger2024simpleconstructionslineardepthtdesigns,bostanci_et_al:LIPIcs.TQC.2025.9},
which is inherently non-classical---typically relying on quantum communication between multiple quantum parties---and which may exist even if one-way functions do not. 
Even if we are willing to believe that $\mathsf{P} \neq \mathsf{NP}$, this begs the question of whether the foundations of post-quantum cryptography, likewise, may benefit from a fundamental re-evaluation:
\begin{quote}
\emph{Can we also base \underline{classical cryptography} on hardness assumptions which are native to quantum information processing?}
\end{quote}

By the usual win-win premise of provable cryptography, any algorithmic advances on such an assumption would likely also have  far-reaching implications for the foundations of quantum information science.

\subsection{Our approach}

In this work, we propose an affirmative answer to the aforementioned question by exploring the cryptographic potential of a natural quantum computational assumption: the hardness of decoding \emph{random} quantum stabilizer codes.

Quantum stabilizer codes are among the most central objects in all of quantum information science.
Not only do they form the backbone of all of quantum error correction and quantum fault-tolerance~\cite{gottesman2024qecc,ErrorCorrectionZoo}, but they are also fundamental in the theory of quantum communication~\cite{Graeme_thesis,Wilde_2013}, entanglement distillation~\cite{Bennett_1996,Devetak_2005,Wilde_2010}, quantum authentication~\cite{1181969,dulek2018quantumciphertextauthenticationkey,10.1007/978-3-030-45727-3_25}, the interactive verification of quantum computations~\cite{aharonov2017interactiveproofsquantumcomputations}, and even in quantum gravity and black hole physics~\cite{Hayden_Preskill_2007,yoshida2017efficientdecodinghaydenpreskillprotocol,Harlow_2013,poremba_et_al:LIPIcs.ITCS.2026.109}.

From a complexity-theoretic perspective, stabilizer decoding is a natural quantum analog of classical decoding problems, such as the \emph{nearest codeword problem}~\cite{brakerski2019worst} or its dual \emph{syndrome decoding problem}~\cite{BMT78} for linear codes. In the \emph{worst} case, the quantum decoding problem appears strictly harder than its classical counterpart: optimal quantum stabilizer decoding is $\#\mathsf{P}$-complete~\cite{iyer2015hardness}, whereas the corresponding classical decoding problem is merely $\mathsf{NP}$-complete~\cite{BMT78}.
This disparity is due to quantum degeneracy: different Pauli errors can act identically on the codespace, so optimal recovery maximizes the total probability of an error class.

At first glance, therefore, it may seem counterintuitive to base classical cryptography on a quantum decoding problem. 
After all, stabilizer decoding appears to involve inherently quantum objects—quantum codewords and quantum noise. 
However, recent work~\cite{khesin2025average} has revealed a surprising equivalence: average-case quantum stabilizer decoding is equivalent (under polynomial-time \emph{quantum} reductions) to a purely \emph{classical} average-case problem, i.e. one that involves only classical inputs and classical outputs, yet retains the essential quantum-mechanical structure and difficulty of the original problem.

This equivalence suddenly opens the door to a tantalizing possibility: \emph{Can we base post-quantum cryptography on the hardness of decoding quantum codes?}

\paragraph{Cryptomania Meets Quantum Error Correction.}

Driven by this possibility, we show that the \emph{average-case} hardness of quantum stabilizer decoding indeed suffices to construct core primitives of \emph{classical} \textsf{Cryptomania}~\cite{514853}, including:
\begin{itemize}
    \item public-key encryption ($\mathsf{PKE}$),
    \item oblivious transfer ($\mathsf{OT}$), as well as
    \item one-way functions ($\mathsf{OWF}$).
\end{itemize}
Since these primitives are known to imply symmetric encryption, asymmetric encryption, and general secure multi-party computation, this establishes stabilizer decoding as a versatile foundation for classical cryptography.

Importantly, our constructions are \emph{efficient and near-optimal}, achieving essentially the same performance as state-of-the-art $\mathsf{LPN}$-based schemes, up to small constant factors.

\begin{figure}[t]
\centering

\renewcommand{\arraystretch}
{1.2} % moderate row spacing
\setlength{\tabcolsep}{6pt}       % base padding

{\begin{tabular}{@{} l @{\hspace{8pt}} l @{\hspace{20pt}} l @{}}
\toprule
 & Learning Stabilizers with Noise & Learning Parity with Noise \\
\midrule
\textbf{Input:} &
\(\big(\vec C \sim \mathcal{C}_n,\; \vec E\,\vec C \ket{0^{n-k},\psi}\big)\) &
\(\big(\vec A \sim \mathbb{Z}_2^{n\times k},\; \vec A\vec x + \vec e \Mod{2}\big)\) \\
\textbf{Noise:} &
Depolarizing noise \(\vec E \sim \mathcal{D}_p^{\otimes n}\) &
Bernoulli noise \(\vec e \sim \Ber_p^{\otimes n}\) \\
\textbf{Task:} &
Recover $k$-qubit Haar state $\ket{\psi}$ &
Recover the string \(\vec x \sim \mathbb{Z}_2^k\) \\
\bottomrule
\end{tabular}}
\caption{{Comparison between (the \emph{state} variant of) $\lsn$~\cite{khesin2025average,poremba2024learning} and $\lpn$~\cite{BFKL93}. In both cases, the input features a classical description of a \emph{random} code; in the case of $\lsn$, it is a random \emph{Clifford} encoding $\vec C \sim \mathcal{C}_n$ of an $[[n,k]]$ quantum stabilizer code, whereas in the case of $\lpn$ it is given by a random \emph{generator matrix} $\vec A \sim \mathbb{Z}_2^{n\times k}$ of a classical $[n,k]$ linear code. The input further consists of a noisy codeword; in the case of $\lsn$, it is in the form of a random quantum codeword, whereas for $\lpn$ it consists of a random classical codeword. 
The parameters are characterized by the logical (qu)bits $k$, physical (qu)bits $n$, and the error probability per physical (qu)bit $p$.}}
\label{tab:lpn-lsn-comparison}
\end{figure}

\paragraph{Learning Stabilizers with Noise.}

The starting point of this work is a recent
characterization of average-case stabilizer decoding by the name of \emph{Learning Stabilizers with Noise} $(\lsn)$~\cite{poremba2024learning,khesin2025average}---the natural quantum analog of $\mathsf{LPN}$~\cite{BFKL93}, which is illustrated in \Cref{tab:lpn-lsn-comparison}.
The $\lsn$ problem was first introduced by Poremba, Quek and Shor~\cite{poremba2024learning} who gave an initial assessment of the problem in terms of algorithms and complexity. 
In subsequent work, Khesin, Lu, Poremba, Ramkumar and Vaikuntanathan~\cite{khesin2025average} showed that $\lsn$ (in most parameter regimes) is \emph{at least as hard} as $\lpn$, providing a much more solid foundation for the average-case hardness of quantum stabilizer decoding.

Here we make the case that $\lsn$ is a compelling post-quantum assumption.
First, $\lsn$ is \emph{genuinely quantum in origin}. 
Unlike leading post-quantum assumptions, such as $\mathsf{LWE}$ or $\mathsf{LPN}$, which are classical in nature but believed to resist quantum attacks, $\lsn$ arises naturally within fundamental quantum information processing itself.
Second, $\lsn$ appears to be \emph{incomparable} to $\mathsf{LPN}$ in cryptographically relevant parameter regimes (as we explain in~\Cref{sec:technical_overview}), thus making it likely to constitute a genuinely new and distinct hardness assumption.

Basing cryptography off of $\lsn$ thus creates a ``win-win-win'' scenario:
\begin{itemize}
\item[$\bullet$] If $\lsn$ is secure, we obtain a new basis for post-quantum cryptography, grounded in a central problem in quantum information processing.
\item[$\bullet$] If $\lsn$ is broken, this would represent a major breakthrough in our current understanding of quantum stabilizer codes, with profound downstream implications for many areas of quantum information science.
\item[$\bullet$] If $\lsn$ turns out to be equivalent to $\mathsf{LPN}$, this would reveal deep and unexpected connections between classical and quantum error correction.
\end{itemize}

\subsection{Our results}

We now briefly summarize our main contributions. 
At a high level, we show that the average-case hardness of decoding random quantum stabilizer codes---captured by (the state variant of) $\lsn$---is sufficient to instantiate all of the core primitives of classical \textsf{Cryptomania}~\cite{514853}. 

Below, we use $n \in \mathbb{N}$ to denote the security parameter and block length of the stabilizer code, and we let $p \in (0,1)$ denote the noise rate, as in \Cref{tab:lpn-lsn-comparison}.
For an in-depth discussion of our assumption and its parameters, we refer to the technical overview in \Cref{sec:technical_overview}.

\paragraph{Public-key encryption.}

In \Cref{sec:public-key}, we construct a $\mathsf{PKE}$ scheme from the hardness of $\lsn$ in the low-noise regime with $p=\Theta(1/\sqrt{n})$. 
The basic scheme matches the parameter regime and efficiency of state-of-the-art Alekhnovich-style $\lpn$-based encryption~\cite{Alekhnovich03,damgaard2012practical}, with $O(n^2)$ encryption time and $O(n)$ decryption time.
One important distinction is that our scheme is based on $\lsn$ with $k = \log(n)$, logarithmically many logical qubits. 
Meanwhile, Alekhnovich-style $\lpn$-based encryption relies on the hardness of $\lpn$ with $k = O(n)$ logical bits. 
However, while the complexity of $\lpn$ depends heavily on the number of logical bits, this dependence is not known to hold for $\lsn$, which is not known to rely in any way on the number of logical qubits. In fact, the known $\lpn$-based hardness bounds at constant noise rates are completely uniform in $k$, holding even with a single logical qubit \cite{khesin2025average}. We prove the following theorem.

\begin{quote}
\textbf{Theorem (informal).}
Assuming (the state variant of) $\lsn$ is hard in the low-noise regime, there exists an efficient post-quantum $\mathsf{PKE}$ scheme.
\end{quote}
Our $\mathsf{PKE}$ scheme is formally described in \Cref{const-pke}, and we prove its security in \Cref{thm:pke_security}. We note that our security reduction is highly non-trivial and requires many new technical insights into $\lsn$ and its related variants; these form
the main technical contributions of the paper.

\paragraph{Oblivious transfer.}

In \Cref{sec:SU-PKE}, we strengthen our scheme to obtain a strongly uniform public-key encryption scheme, which once again rests on the low-noise regime with $p=\Theta(1/\sqrt{n})$.
Using known black-box transformations, this implies malicious-secure oblivious transfer:

\begin{quote}
\textbf{Theorem (informal).}
Assuming (the state variant of) $\lsn$ is hard in the low-noise regime, there exist efficient post-quantum malicious-secure $\mathsf{OT}$ protocols with constant rounds and $\varepsilon$-simulation, or polynomial rounds and full simulation.
\end{quote}

Since oblivious transfer implies general secure multi-party computation, this establishes that $\lsn$ suffices for the full power of Cryptomania.

\paragraph{One-way functions.}

In \Cref{app:owff_tight_equivalence}, we construct a one-way function ($\mathsf{OWF}$) from $\slsn$ in the constant-noise regime $p=\W(1)$. 
This gives symmetric cryptography directly from the hardest possible parameter setting of the problem.

\begin{quote}
\textbf{Theorem (informal).}
Assuming (the state variant of) $\lsn$ is hard in the constant noise regime, there exists a post-quantum $\mathsf{OWF}$.
\end{quote}
Importantly, this variant of $\lsn$ is known to be at least as hard as $\lpn$ in most regimes~\cite{khesin2025average}.

Taken together, these results show that the hardness of decoding random quantum stabilizer codes suffices to realize symmetric encryption, public-key encryption, and malicious-secure multi-party computation. Our constructions match the efficiency of the best known $\lpn$-based schemes~\cite{pietrzak2012lpn} while resting on an assumption that is quantum-native and, as we argue in \Cref{app:reduction_barrier}, plausibly incomparable to existing post-quantum assumptions.

\subsection{Related work}

The hardness of decoding random classical linear codes and the closely related \(\mathsf{LPN}\) (Learning Parity with Noise) problem has been studied extensively in coding theory and cryptography; see e.g. \cite{10.1007/978-3-642-03356-8_35,gilbert2008encrypt,10.1007/11535218_18,10.1109/FOCS.2006.51,arora2011new,Alekhnovich03,pietrzak2012lpn,10.5555/647097.717000,10.1007/11538462_32,yu2019collision,BFKL93}. 
A common formulation of \(\mathsf{LPN}\) presents the adversary with a sequence of $n$ noisy linear samples $\{(\vec a_i,\; \langle \vec a_i , \vec x \rangle + e_i)\}_{i=1}^n$ (where $\vec x \in \Z_2^k$), and hence a significant body of work has been focused on the fine-grained regimes where the number of samples \(n\) (equivalently, the block length of the code) is significantly larger than the message length. 
From a coding-theoretic perspective, this corresponds to decoding a random linear code with vanishing rate \(R = k/n\).
For example, BKW-style techniques yield \(2^{O(k / \log k)}\)-time algorithms when \(n\) is superpolynomial in \(k\) \cite{BKW03}, and refinements give \(2^{O(k / \log \log k)}\)-time algorithms when \(n\) is slightly superlinear in \(k\) \cite{10.1007/11538462_32}. 
Despite decades of work, however, no subexponential-time algorithm is known in the constant-rate, constant-noise regime \(n = \Theta(k)\), which is the regime underlying many important \(\mathsf{LPN}\)-based cryptographic constructions. In the constant-rate, constant-noise regime, the best known algorithms for $\lpn$ tend to involve variants of \emph{information set decoding} (ISD) attacks~\cite{Prange1962,stern1988} which run in exponential time~\cite{EC:CDMT24}.

The study of quantum stabilizer decoding has historically followed a different direction. 
Most prior work has focused on \emph{worst-case} formulations of the problem, often in the form of maximum-likelihood decoding, and established classical hardness results and complexity-theoretic barriers for these variants \cite{HsiehLeGall11,iyer2015hardness,kuo2020hardnesses}. 
Only recently the attention shifted to \emph{average-case} formulations that are more suitable for cryptographic applications. 
In particular, Poremba, Quek, and Shor introduced the \emph{Learning Stabilizers with Noise} (\(\mathsf{LSN}\)) problem and provided an initial analysis of its algorithmic and complexity-theoretic properties \cite{poremba2024learning}. They also constructed a quantum bit commitment scheme from the hardness of the $\lsn$ problem.
Subsequent work of Khesin, Lu, Poremba,
Ramkumar and Vaikuntanathan further clarified the landscape by relating average-case stabilizer decoding to classical noisy linear problems like $\lpn$ in many parameter regimes \cite{khesin2025average}. 
This work shows that the problem admits a purely classical description while retaining the symplectic algebraic structure intrinsic to stabilizer codes, thereby providing initial evidence that decoding random stabilizer codes may serve as a meaningful cryptographic hardness assumption.

Outside of quantum error correction, there have been a number of recent attempts at building \emph{quantum} rather than \emph{classical} cryptography from other natural average-case hard computational tasks~\cite{bostanci_et_al:LIPIcs.TQC.2025.9,10.1145/3717823.3718145,fefferman2025hardnesslearningquantumcircuits,hiroka2025hardnessquantumdistributionlearning,Hiroka_2025,cavalar2024metacomplexitycharacterizationquantumcryptography,morimae2024quantumgroupactions}. However, unlike in our work, these constructions are inherently quantum (typically requiring quantum communication between multiple quantum parties) and do not give rise to classical cryptographic primitives as in our work.

\section{Technical overview}
\label{sec:technical_overview}

The three primitives on which this paper focuses are one-way functions (\textsf{OWF}), public-key encryption (\textsf{PKE}), and oblivious transfer (\textsf{OT}).
All have been constructed---with a high degree of practical efficiency---from the $\lpn$ assumption~\cite{Alekhnovich03,damgaard2012practical,dottling2020two,david2014universally,friolo2019black}.

Our main goal is to construct each of these primitives in such a way that \begin{enumerate}
    \item[(a) ] breaking each primitive is at least as hard as decoding a random quantum stabilizer code, and 
    \item[(b) ] each primitive's efficiency is as good as the state-of-the-art corresponding $\lpn$-based construction, up to small constant factors.
\end{enumerate}
This overview sketches the key technical ingredients for our constructions.
Our general recipe proceeds as follows.
First, we reduce the task of decoding a random quantum stabilizer code to the task of decoding a random \emph{classical} linear code drawn from an ensemble of codes satisfying a certain algebraic structure. This problem is known as \emph{symplectic} $\lpn$ ($\symplpn$), and was introduced recently as a technical tool to reduce $\lpn$ to average-case quantum stabilizer decoding in the high-noise regime~\cite{khesin2025average}.
Because $\symplpn$ and $\lpn$ appear similar conceptually, one could hope that $\lpn$-based cryptographic constructions would be readily amenable to adaptation into a $\symplpn$-based protocol, thereby achieving constructions secured by the hardness of quantum error correction.
In reality, this intuition only holds reasonably well for \textsf{OWFs}, wherein inverting the \textsf{OWF} is essentially always equivalent to solving the underlying computational problem anyway.
Asymmetric cryptographic protocols based on $\lpn$, on the other hand, rely critically on the fact that the encoding matrix of the $\lpn$ instance features uniformly random, independent entries as part of the security proofs, and thus break down when this assumption is relaxed.
\emph{This breakdown poses a serious technical barrier when proving the security of $\symplpn$-based schemes.}

The second component of our recipe involves overcoming these barriers by introducing an entirely new suite of techniques for scrambling and unscrambling information in linear subspaces, with the algebraic conditions specified by $\symplpn$.
This approach enables us to rigorously prove the security of our constructions under the hardness of $\slsn$, at the cost of significantly increasing the technicality of the proofs relative to those of $\lpn$-based schemes.

\subsection{Reduction to an $\lpn$-like Classical Problem}

Informally, $\slsn[k, n, p]$ is the task of decoding a random quantum stabilizer code with $k$ logical qubits, $n$ physical qubits, and noise rate $p \in (0, 1)$.
Specifically, we draw a Haar-random (i.e. uniformly random) logical state $\ket{\psi}$ as well as a random $n$-qubit Clifford operator $\vec C$.
We also draw a random error operator $\vec E$ from the \emph{depolarizing} distribution $\CD_p^{\otimes n}$---the natural quantum analog of Bernoulli error (i.e. binary symmetric channel) in classical error correction.

The task of $\slsn$ is to approximately recover the logical state $\ket{\psi}$, given a classical description of the code $\vec C$ and the noisy code state $\vec E \vec C (\ket{0^{n-k}} \otimes \ket{\psi})$.
Being a problem with manifestly quantum inputs and outputs, $\slsn$ as stated cannot be used directly for classical cryptography.
However, if we can produce a \emph{quantum reduction} from $\slsn$ to a problem with manifestly classical inputs and outputs, then we can at once use this classical problem directly to construct cryptography and be sure that such protocols are secured directly by the hardness of $\slsn$.
As a first step, consider a variant of $\slsn$ wherein the logical state $\ket{\psi}$ is a uniformly random \emph{bitstring} $\ket{\vec x}$ for $\vec x \sim \Z_2^k$, rather than a Haar-random state; thus the output of $\lsn$ is classical though the input remains quantum.
This variant is known as $\lsn$~\cite{poremba2024learning}, and when $k = O(\log n)$ there is a known quantum reduction from $\slsn$ to $\lsn$~\cite{khesin2025average}.
Importantly, while the hardness of $\lpn$ is characterized by the logical dimension $k$---there is a brute-force guess-and-check $\poly(2^k,n)$-time algorithm to solve it---the hardness of $\lsn$ does not appear to depend on $k$. At first glance, this may seem counterintuitive. The reason is that, unlike for $\lpn$, there is no known efficient method of verifying proposed $\lsn$ solutions.
Indeed, at constant noise rates, $\slsn$ is known to be as hard as constant-rate $\lpn$ \emph{even for a single logical qubit}~\cite{khesin2025average}.
Hence, the restriction to $k = O(\log n)$ does not necessarily come at the cost of hardness.

Moreover, $\lsn$ is equivalent (via efficient quantum reductions) to a completely classical problem, which we henceforth also refer to as $\lsn$.
Defining this classical problem requires a symplectic form on the even-dimensional ambient space $\Z_2^{2n}$.
Given bitstrings $(\vec a, \vec b), (\vec a', \vec b') \in \Z_2^{2n}$, the \emph{symplectic inner product} is given by \begin{align}
    (\vec a, \vec b) \odot (\vec a', \vec b') := \vec a \cdot \vec b' + \vec a' \cdot \vec b \pmod{2} .
\end{align}
Then, the classical equivalent of $\lsn[k, n, p]$ is the task of recovering a bitstring $\vec y \sim \Z_2^{k}$, given $([\vec A \,|\, \vec B], \vec{Ax} + \vec{By} + \vec e)$, where $\vec x \sim \Z_2^{n}, \vec A \in \Z_2^{2n \times n}, \vec B \in \Z_2^{2n \times k}, \vec e \in \Z_2^{2n}$; $\vec A$ and $\vec B$ are sampled jointly and uniformly subject to each having symplectically orthogonal columns and $[\vec A \,|\, \vec B]$ having full column rank, while $\vec e$ is drawn from a symplectic representation of the depolarizing distribution $\CD_p^{\otimes n}$.
Concretely, each pair $(e_j, e_{n+j})$ is i.i.d., being $(0,0)$ with probability $1-p$ and $(0,1), (1,0), (1,1)$ each with probability $p/3$.

While completely classical, $\lsn[k, n, p]$ still differs structurally from $\lpn[k, n, p]$ in that only a small part of the effective logical state $[\vec x \,|\, \vec y]$ needs to be recovered.\footnote{This is due to a uniquely quantum phenomenon called \emph{stabilizer degeneracy}~\cite{iyer2015hardness,gottesman2024qecc}; roughly speaking, two distinct errors can have the same effect on a codeword.}
In $\lpn$, the entire logical state must be recovered.

Surprisingly, however, we are able to reduce $\lsn$ to a much more $\lpn$-like problem, known as $\symplpn[k, n, p]$.
This is the same task as $\lpn[k, 2n, p]$, except the encoding matrix $\vec A \in \Z_2^{2n \times k}$ is uniformly random \emph{subject to} having symplectically orthogonal columns, and the error is drawn from the depolarizing distribution's symplectic representation.
Thus, $\symplpn[k, n, p]$ and $\lpn[k, 2n, p]$ describe essentially the same objects, but draw them from different distributions.
Note that formally, the $\symplpn$ and $\lpn$ variants we consider are the \emph{decision} variants, wherein one is given $(\vec A, \vec z)$ and must decide if $\vec z = \vec{Ax} + \vec e$ (``\texttt{structured}'') or $\vec z$ is uniformly random (``\texttt{unstructured}'').
The \emph{decision} variant of $\lpn$ is the standard cryptographic variant, and is essentially equivalent to the \emph{search} variant~\cite{fischer1996efficient,katz2010parallel} (up to small factors).
Importantly, we only consider the \emph{search} variant of $\lsn$, as their \emph{search} and \emph{decision} variants appear inequivalent in general~\cite{khesin2025average} and the \emph{search} variant is natural for quantum information processing.
\emph{Search} and \emph{decision} variants of $\symplpn$ are not known to be equivalent.

To achieve the first component in our recipe, we show in \Cref{thm:lsn_to_symplpn} that, perhaps surprisingly, there is a simple classical reduction from $\lsn[k, n, p]$ to $\symplpn[n, n, p]$ for any $k = O(\log n)$.
Intuitively, $\symplpn$ appears to be ``embedded'' in $\lsn$ as the $\vec A$-part of the $\lsn$ matrix: $\symplpn$ is of the form $(\vec A, \vec{Ax} + \vec e)$ while $\lsn$ takes the form $([\vec A \,|\, \vec B], \vec{Ax} + \vec{By} + \vec e)$.
Our reduction relies on the simple observation that the noisy codeword part of $\lsn$ is \emph{identical} to that of $\symplpn$ when $\vec y = \vec 0$.
If that is the case, then dropping the $\vec B$ part of the $\lsn$ matrix yields $(\vec A, \vec{Ax} + \vec e)$, precisely the form of Decision $\symplpn$ with a \texttt{structured} input.
If $\vec y \neq \vec 0$, then again dropping the $\vec B$ part of the $\lsn$ matrix yields $(\vec A, \vec{Ax} + \vec{By} + \vec e)$.
In fact, we show that $\vec{By}$ is marginally negligibly close (in total variation distance) to a uniformly random vector, and thus this instance is negligibly close to an \texttt{unstructured} $\symplpn$ task.
Therefore, a $\symplpn$ solver can decide whether or not $\vec y = \vec 0$ in the $\lsn$ problem.
If $\vec y = \vec 0$ we are done; if $\vec y \neq \vec 0$, we guess a random $\hat{\vec y} \in \Z_2^k \setminus\set{\vec 0}$.
The $\symplpn$ solver hence gives us a minuscule advantage in solving $\lsn$ by providing a larger signal only when $\vec y = \vec 0$.
However, since $k = O(\log n)$, even such a small signal gives a solver with a non-negligible overall advantage over completely random guessing for $\lsn$, which turns out to be sufficient to complete the reduction.
Overall, our chain of reductions proceeds as, for $k = O(\log n)$, \begin{align}
    \slsn[k, n, p] \leq_Q \lsn[k, n, p] \leq_C \symplpn[n, n, p] ,
\end{align}
where $\leq_Q, \leq_C$ respectively denote quantum and classical reductions.
Crucially, this reduction works \emph{only} for $\symplpn$ with exactly $n$ logical bits (i.e. rate $1/2$).
With even \emph{one} less logical bit, it is no longer clear as to how much easier the $\symplpn$ instance becomes.

\subsection{Cryptographic Constructions with $\symplpn$}

Having established that (\emph{decision}) $\symplpn$ is at least as hard as $\slsn$, we next outline our cryptographic constructions and the technical barriers to a security proof relative to their $\lpn$-based counterparts.
Our \textsf{OWF} construction is the most straightforward, and unlike the public-key protocols is based directly on $\slsn[k, n, p]$, and in the hardest regime wherein $p = \Omega(1)$ and $k = \Omega(n)$.
For simplicity, we here discuss it starting from $\lsn$.
The function is indexed by $[\vec{A} \,|\, \vec{B}]$, and maps on input $(\vec x, \vec y, \vec e)$ to $\vec{Ax} + \vec{By} + \vec{e}$.
Inverting this function is equivalent to finding $\vec{x}, \vec{y}, \vec{e}$ from $([\vec A \,|\, \vec{B}], \vec{Ax} + \vec{By} + \vec{e})$, modulo stabilizer equivalence.
This is at least as hard as solving $\lsn[k, n, p]$, which only requires finding $\vec y$.
Since the break also recovers $\vec e$ (up to stabilizer equivalence), we can also use it to solve even harder decoding problems, namely $\slsn$.
Moreover, in the regime of $p = \W(1)$, $\lpn[pn, 2n, p]$ is known to reduce to $\slsn[k = O(\log n), n, \Theta(p)]$~\cite{khesin2025average}, and thus our \textsf{OWF} is provably at least as secure as a conventional $\lpn$-based \textsf{OWF}.
We defer the full construction and security proof to \Cref{app:owff_tight_equivalence}.

On the other hand, the \textsf{PKE} and \textsf{OT} constructions in this work are built upon $\symplpn$.
Our starting point is the $\lpn$-based \textsf{PKE} scheme in~\cite{damgaard2012practical} which is rooted in the seminal scheme of Alekhnovich~\cite{Alekhnovich03}, and which remains the state-of-the-art $\lpn$-based construction in terms of security and efficiency.
This particular scheme uses $p = \Th(1/\sqrt{n})$, because the correctness of the scheme relies heavily on the inner product of two independently sampled noise vectors being $0$ with high probability; this occurs precisely when $p = O(1/\sqrt{n})$.
For constant code rate $k/n = R \in (0,1)$, information-set decoding gives a classical attack with running time $\poly(n)2^{O(pn)}$ in this low-noise regime~\cite{Prange1962}.
For $p = \Theta(1/\sqrt{n})$, this gives an attack running in time $2^{O(\sqrt{n})}$.\footnote{The constant-noise $\lpn$ construction of~\cite{yu2016cryptography} admits a $2^{O(\log^2 n)}$ exhaustive-search attack, compared with the $2^{O(\sqrt n)}$ information-set-decoding bound for Alekhnovich-type schemes.}
Much of this construction is readily adaptable to $\symplpn$ in place of $\lpn$, with one very significant barrier.
More precisely, in the process of proving the security of our $\symplpn$-based scheme (formally described in \Cref{const-pke}), we find in \Cref{sec:public-key} that the protocol's security reduces to the hardness of both $\symplpn[n, n, p]$ \emph{and} $\symplpn[n-1, n, p]$.
If this were $\lpn$, this subtlety would be essentially irrelevant---the hardness of $\lpn$ provably does not depend on small changes in $k$.
In particular, there is a simple self-reduction from $\lpn[k, n, p]$ to $\lpn[k-1, n, p]$.
For $\symplpn$, however, this subtlety becomes a significant technical barrier.
In our case, the security of our adapted \textsf{PKE} scheme would rely on both $\symplpn[n, n, p]$ and $\symplpn[n-1, n, p]$.
But, as discussed above, a reduction from $\slsn$ is only known for $\symplpn[k, n, p]$ with $k = n$ logical bits, exactly half the number of physical bits $2n$.
In addition, the na\"ive self-reduction in which we drop one logical bit, which succeeds on $\lpn$, fails completely on $\symplpn$ because the entries are jointly nearly maximally far from uniformly random.
Despite this failure, we find an alternative and substantially more involved approach which successfully reduces $\symplpn[n, n, p]$ to $\symplpn[n-1, n, p]$ by introducing new techniques for scrambling isotropic subspaces in a carefully controlled manner.
We outline the key idea for these techniques below, but with this reduction complete we prove the (\textsf{IND-CPA}) security of our \textsf{PKE}.
Notably, our basic scheme has the same runtime up to small constant factors as state-of-the-art Alekhnovich-type $\lpn$-based \textsf{PKE}, which are quite efficient with $O(n^2)$ time encryption, $O(n)$ time decryption.

Finally, we construct malicious-secure (i.e. even if one party deviates arbitrarily from protocol, the other party's security is still guaranteed) \textsf{OT} from the hardness of $\symplpn[n, p]$.

We use post-quantum transformations that give constant-round $\varepsilon$-simulation or polynomial-round full simulation.
These transformations apply to a particularly structured \textsf{PKE} scheme known as a \emph{strongly uniform} \textsf{PKE} scheme (\textsf{SU-PKE}), via its dense-cryptosystem property~\cite[Theorems~12, 13, and~20]{chatterjee2025round}.
Informally, \textsf{SU-PKE} is the same as \textsf{PKE} with the added property that the public key in the scheme is computationally indistinguishable from a uniformly random bitstring.
Since Alekhnovich-type \textsf{PKE} based on $\lpn[k, n, p]$ has public key $(\vec A \sim \Z_2^{n \times k}, \vec{Ax} + \vec{e})$ which is by assumption computationally indistinguishable from the uniformly random bitstring $(\vec{A} \sim \Z_2^{n \times k}, \vec{u} \sim \Z_2^n)$, they yield these \textsf{OT} guarantees under post-quantum hardness.
However, our $\symplpn$-based \textsf{PKE} scheme has a public key which is easily distinguishable from uniformly random because the key takes the form $(\vec A, \vec{Ax} + \vec e)$, where $\vec A \in \Z_2^{2n \times k}$ is a uniformly random full-rank matrix subject to having symplectically orthogonal columns.
As a result, the key technical barrier to the construction of malicious-secure \textsf{OT} via these transformations from $\symplpn$ is the modification of the \textsf{PKE} scheme to have a public key computationally indistinguishable from random.
We achieve this modification in \Cref{sec:SU-PKE}.
Our modification begins with a well-known idea that a randomized algorithm sampling from some not-necessarily-uniform distribution can be replaced with a deterministic algorithm which accepts an additional input---a sufficiently long ``seed'' string $\vec s$ which is genuinely uniformly random.
In our case, the generation of $\vec A$ can be executed by a deterministic algorithm $\CA$ which iteratively builds a basis of the subspace symplectically orthogonal from the columns already sampled, and then picking a random linear combination of these basis vectors to be the next column of $\vec A$.
This linear combination can be chosen with access to $O(n)$ random bits.

Thus, we would hope to replace the public key $(\vec A, \vec{Ax} + \vec e)$ with $(\vec s, \vec{Ax} + \vec e)$, where $\vec A$ is computed from $\vec s$.
This replacement has no impact on correctness, as the encrypting party simply first computes $\vec A$ from $\vec s$ using $\CA$.
However, while this form appears very much hard to distinguish from uniformly random bits $(\vec s, \vec u \sim \Z_2^{2n})$, we formally know only that $(\vec A, \vec{Ax} + \vec e)$ is indistinguishable from $(\vec A, \vec u \sim \Z_2^{2n})$.
Thus, to complete the proof, we must show that given $\vec A$ we can distributionally ``invert'' the $\vec s \mapsto \vec A$ algorithm to produce a marginally uniformly random $\hat{\vec s}$ such that $\CA(\hat{\vec s}) = \vec A$.
We give such an inversion algorithm and prove its correctness in \Cref{sec:SU-PKE}.
Thus, $\vec{Ax} + \vec e$ is indistinguishable from random given either $\vec A$ or $\vec s$, so $(\vec s, \vec{Ax} + \vec e)$ is indeed computationally indistinguishable from uniformly random bits.
We therefore achieve malicious-secure \textsf{OT} and general secure multi-party computation from the hardness of $\symplpn$, with constant-round $\varepsilon$-simulation or polynomial-round full simulation.

\subsection{Removing one logical bit in $\symplpn$ via symplectic scrambling}
We here outline our solution to the primary technical barrier to the security proof of $\symplpn$-based \textsf{PKE}, namely reducing $\symplpn[n, n, p]$ to $\symplpn[n-1, n, p]$.
If we wished to reduce one logical qubit for $\lpn[k, n, p]$, it would be straightforward: given $(\vec A \sim \Z_2^{n \times k}, \vec{Ax} + \vec e)$, let $\vec A' \in \Z_2^{n \times (k-1)}$ be $\vec A$ with the last column removed, and run a $\lpn[k-1, n, p]$ decider on $(\vec A', \vec{Ax} + \vec e)$.
If $x_k = 0$, then $\vec{Ax} + \vec e = \vec A' \vec x' + \vec e$, where $\vec x = [\vec x' \,|\, x_k]$.
Otherwise, $\vec{Ax} + \vec e = \vec A' \vec x' + \vec e + \vec a$, where $\vec a \sim \Z_2^n$ is the last column of $\vec A$---so $\vec{Ax} + \vec e$ is uniformly random given $\vec{A}'$.
Hence, a $\lpn[k-1, n, p]$ decider will output \texttt{structured} only when the input is structured with $x_k = 0$, which occurs with large enough probability to obtain a solver for $\lpn[k, n, p]$.

In the case of $\symplpn[n, n, p]$, however, an analogous technique fails completely.
Given $(\vec A \in \Z_2^{2n \times n}, \vec{Ax} + \vec e)$, we might again try to remove the last column of $\vec A$ to obtain $\vec A'$.
If $x_n = 0$, then indeed $\vec{Ax} + \vec e = \vec{A}' \vec x' + \vec e$ is a $\symplpn[n-1, n, p]$ structured instance.
However, if $x_n = 1$, then $\vec{Ax} + \vec e = \vec{A}' \vec x' + \vec e + \vec a$, but this time the last column $\vec a \in \Z_2^{2n}$ of $\vec A$ has a complicated distribution which depends sensitively on the rest of $\vec A$.
In other words, $\vec a$ is nearly maximally far from uniformly random given $\vec A'$, so the instance is not at all close to unstructured.
Consequently, we have no guarantees on the accuracy of a $\symplpn[n-1, n, p]$ decider on this input.

To overcome this barrier, we introduce four techniques which approach the reduction in a completely different way and thereby bypass the above obstruction.
Our first technique arises from the observation that the maximal dimension of a subspace of $\Z_2^{2n}$ in which all vectors are symplectically orthogonal is $n$.
Thus, a useful geometric interpretation of a $(n-1)$-dimensional symplectically orthogonal subspace---the image of the matrix $\vec A' \in \Z_2^{2n \times (n-1)}$ in a $\symplpn[n-1, n, p]$ instance---is as a hyperplane of codimension $1$ within some maximal $n$-dimensional symplectically orthogonal subspace, symplectically orthogonal to some random vector $\vec v \in \Z_2^{2n}$.
(Technically, there are many possible maximal subspaces which contain this hyperplane, but we ignore this subtlety in this overview.)
Given a $2n \times n$ matrix $\vec A$ in a $\symplpn[n,n, p]$ instance, our first step is to reduce the dimensionality by 1 by forcing all vectors in the code to be symplectically orthogonal to some fixed $\vec v$---this step effectively removes a logical bit.
This is easy to do if the first entry of the error $\vec e$ is 0: we can then easily transform the code space to always have first entry $0$.
Equivalently, the code is symplectically orthogonal to the vector $(0^n, 1, 0^{n-1})$.
However, a $\symplpn[n-1, n, p]$ instance has a code orthogonal to a \emph{random} vector $\vec v$, not a fixed one.
Thus, our next step is to design a very sparse random operator $\vec C \in \Z_2^{2n \times 2n}$ which ``rotates'' this orthogonal vector to a random vector, taking its normal hyperplane along with it.
We multiply to produce $(\vec{CA}, \vec{CAx} + \vec{Ce})$.
Since $\vec{C}$ is very sparse, the new error $\vec{Ce}$ is not irrecoverably blown up, but does have a complicated distribution far from depolarizing.
Hence, our third technique is to apply a ``noise symmetrization'' operation which, using a combination of noise flooding, starting on a fixed set of pairs, and permutations, maps $\vec{Ce}$ to a depolarizing error with only slightly larger noise rate $p'$ than $p$, while not affecting the distribution of the other objects.
This technique turns out, however, to only work when the $(n+1)$th component of $\vec e$ is 0.

In summary, our approach only works when entries $1$ and $n+1$ of $\vec e$ are $0$, which does not occur with overwhelming probability.
So, to complete the reduction, our final technique is an ``interpolation trick''.
This trick arises from the observation that we are free to add extra noise prior to the start of the reduction if we wish, which will change the probabilities of the reduction's success.
A simple analysis of these probabilities reveals that the two reductions---one with and one without the preliminary noise flooding---cannot both fail, and thus there must exist a successful reduction from $\symplpn[n, n, p]$ to $\symplpn[n-1, n, p']$ for $p'$ only slightly larger than $p$.

This reduction completes the security proof of our $\symplpn$-based \textsf{PKE} scheme. 
Given the many technical obstructions which arise in performing even the simplest $\symplpn$ self-reduction, it is somewhat surprising that a rigorous security proof of $\symplpn$-based \textsf{PKE} scheme is possible at all.

\subsection{Comparative Hardness of $\lpn$ and $\symplpn$}

The fact that $\symplpn$ readily implies all of the core primitives of \textsf{Cryptomania} with $\lpn$-level efficiency gives a compelling reason to consider $\symplpn$ as an independent post-quantum assumption, provided that \emph{$\symplpn$ is not already equivalent or reducible to $\lpn$}.

Prior work has shown a reduction from $\lpn[pn, 2n, p]$ to $\symplpn[n, n, 6p]$~\cite{khesin2025average}.
With $p = O(1/\sqrt{n})$, however, this reduction becomes vacuous, as $\lpn[\sqrt{n}, 2n, 1/\sqrt{n}]$ can be solved in polynomial time by restricting to the first $\sqrt{n} \times \sqrt{n}$ minor of $\vec A$ (where an error rate of $1/\sqrt{n}$ implies that there typically only $O(1)$ errors on this block) and brute-force enumerating constant-weight errors.
It is therefore not known if $\lpn$ reduces to $\symplpn$ in this low-noise regime, or if in fact $\symplpn$ reduces to $\lpn$.
The former reduction would establish $\symplpn$ as being at least as good as $\lpn$ for cryptography in light of this work, while the latter reduction would prove that $\symplpn$ adds no value for post-quantum security.
While we leave the former open, we prove in \Cref{app:reduction_barrier} a strong barrier against the existence of a $\symplpn$-to-$\lpn$ reduction, suggesting that the assumptions are \emph{inequivalent}.

The most natural approach to a reduction would be to, given $(\vec A \in \Z_2^{2n \times n}, \vec{Ax} + \vec e)$, simply remove around $n/2$ rows of $\vec A$.
After all, $\vec A$ has $2n^2$ entries but only has entropy $\sim \frac{3}{2} n^2$ since there are $\binom{n}{2} \sim \frac{1}{2} n^2$ constraints.
% This strategy turns out to only work if we remove $n$ rows---precisely the number of rows removed for the problem to become statistically unsolvable.
More generally, we may wish to multiply $\vec A$ and $\vec{Ax} + \vec e$ by some $\vec B \in \Z_2^{m \times 2n}$, producing $(\vec{BA}, (\vec{BA})\vec x + \vec{Be})$, and hope that this instance is close to a $\lpn[n, m, p']$ instance for some $p'$.
It is not at all clear as to why such a strategy would not work, especially because intuition from the leftover hash lemma~\cite{impagliazzo1989pseudo} suggests that it in fact may succeed.
Moreover, such a ``linear reduction'' was used, for example, to successfully produce the first non-trivial random self-reduction for $\lpn$ from polylog-weight errors in the worst case to error rate $1/2 - 1/\poly(n)$ in the average case~\cite{brakerski2019worst}.
We prove, however, that linear reductions \emph{cannot} reduce $\symplpn$ to $\lpn$.
While the proof is quite technical, the central idea is simple.
We show that for any \emph{fixed} $\vec B \in \Z_2^{m \times 2n}$, the random matrix $\vec{BA} \in \Z_2^{m \times n}$ is severely deficient in entropy.
That is, while $\vec{BA}$ must have entropy about $mn$ to be uniformly random (and thus a $\lpn$ matrix), it actually only has entropy at most $(1-d)mn$ for some constant $d$.
To compensate this deficiency, we must choose $\vec B$ from a sufficiently high-entropy distribution.
Using this result, we then prove that any distribution over $\vec B$ which adds enough entropy so as to randomize $\vec{BA}$ also irrecoverably blows up the error.
More precisely, in the codeword $(\vec{BA})\vec x + \vec{Be}$ part of the input, the error $\vec{Be}$ has expected weight at least $\frac{(1-r-\delta)}{2}m$ for any $\delta>0$ with overwhelming probability over $\vec B$, where $r$ is the rate $r = \frac{n}{m}$.
At constant rate, adding fresh noise to make the errors independent leaves their rate above Shannon's decoding threshold.

Our barrier does not imply that no reduction exists, but it shows that the most clearly motivated approach faces a noise barrier.
A reduction following this approach would have to overcome this barrier, and it is possible that if such a reduction exists, it could also improve the random self-reductions achievable for $\lpn$.
(The current linear reduction---from worst case with polylog-weight to average case with noise rate $1/2 - 1/\poly(n)$---has very weak parameters, with the primary barrier to improvement being the same issue that we rigorously derive here.)

\subsection{Cryptanalysis} 
Given the close relationship between $\symplpn$ and $\lpn$, a natural question is how the security of our \symplpn-based schemes  compares to that of standard $\lpn$-based schemes in practice. 
We focus on our $\mathsf{PKE}$ scheme in \Cref{const-pke}, as our \textsf{OWF} construction may assume the hardest variant of \lsn, which is already known to be at least as hard as $\lpn$ in most relevant parameter regimes~\cite{khesin2025average}.

As pointed out in~\cite{EC:CDMT24},
the dominant class of attacks against $\lpn$ are information-set decoding (ISD) algorithms, originating in the work of Prange and Stern~\cite{Prange1962,stern1988}. 
These attacks are combinatorial: they exploit the sparsity of the error vector and are largely agnostic to the algebraic structure of the underlying code. 
Their running time scales exponentially in the block length, with the exponent governed by the particular noise rate. 
Because of this structural insensitivity, we expect ISD-type attacks to serve as the primary \emph{generic} benchmark for \symplpn-based \textsf{PKE}.

To enable a meaningful comparison, it is important to place the two schemes on equal footing; for example, standard practice~\cite{damgaard2012practical} suggests matching parameters by fixing the decryption success probability to be $0.75$. 
Since $\symplpn$ employs depolarizing noise whereas $\lpn$ uses Bernoulli noise, we can compare our decryption analysis with that of $\lpn$-based schemes~\cite{damgaard2012practical} and equate the corresponding decryption success probabilities: 
if $p$ and $q$ denote the noise rates in \symplpn[n,p] and $\lpn(n,2n,q)$, respectively, this matching yields
\begin{equation}\label{eq:matched}
    1-\frac{4}{3}p^2 = (1-2q^2)^2,
    \qquad\text{and hence}\qquad
    p = \sqrt{3(q^2-q^4)} = \sqrt{3}\,q - O(q^2).
\end{equation}
In particular, for the same decryption probability, say $0.75$, the $\symplpn$ distribution produces errors of slightly larger expected weight (by a constant factor) than the corresponding $\lpn$ instance.

From the perspective of ISD, this has two competing effects. 
On the one hand, since ISD algorithms are driven primarily by the sparsity of the error vector, the increased error weight suggests that \emph{na\"{\i}ve} ISD attacks should perform somewhat worse against $\symplpn$ than against $\lpn$. 
On the other hand, the depolarizing noise in $\symplpn$ exhibits a mild pairwise correlation structure, which induces a small entropy loss relative to fully independent noise. 
Naturally, this structure can be exploited by a simple \emph{pair-aware} variant of ISD---e.g., modifications of Prange's original algorithm~\cite{Prange1962}---in which coordinates are processed in pairs $(j,n+j)$ rather than independently.
Taken together, these considerations suggest that while off-the-shelf ISD attacks may underestimate the vulnerability of \symplpn, tailored ISD variants should provide an accurate point of comparison, and are expected to achieve performance comparable to that observed for $\lpn$-based \textsf{PKE}. 
This heuristic is consistent with our preliminary numerical investigations.

Our results so far suggest a qualitative message: $\symplpn$-based \textsf{PKE} schemes are likely to be similarly susceptible to brute-force ISD-attacks as comparable $\lpn$-based schemes---even if the security of $\symplpn$ may ultimately hinge on the symplectic structure of the underlying code and its interaction with the noise model. 
This is because ISD's brute-force approach has little dependence on the actual code distribution itself, and therefore does not distinguish between the random matrices in $\lpn$ and the random symplectically orthogonal matrices in $\symplpn$.
Rather, as discussed above, ISD depends more sensitively on the noise distribution.
We thus believe that, for practical purposes, alternative noise distributions (possibly deviating from our quantum-inspired depolarizing model) could further amplify the security of the scheme. 
Exploring such directions, as well as refining the analysis of tailored attacks, remains an interesting avenue for future work.

\subsection{Outlook}

This paper shows that private-key encryption, public-key encryption, and secure multi-party computation can be constructed entirely from the average-case hardness of decoding quantum stabilizer codes, $\slsn[k, n, p]$.
In the private-key regime, we construct one-way functions from high-rate ($\Omega(1)$), high-noise ($\Omega(1)$) $\slsn$.
In the public-key regime, we construct \textsf{PKE} and malicious-secure $\textsf{OT}$ (which then implies malicious-secure multi-party computation) from low-rate ($O(\frac{\log n}{n})$), low-noise ($\Theta(1/\sqrt{n})$) $\slsn$, for which known generic decoding algorithms run in time $2^{\widetilde{O}(\sqrt{n})}$.
The security of $\slsn$ does not appear to depend on the rate of the code at all~\cite{khesin2025average}, and the known generic decoding bounds in the low-noise regime have the same $2^{\widetilde{O}(\sqrt{n})}$ form as those for low-noise $\lpn$-based \textsf{PKE} schemes.

From the construction point of view, a pertinent question is whether public-key encryption can be built directly out of high-rate, low-noise $\slsn$.
Our approach has been to first quantumly reduce $\slsn[k, n, p]$ with $k = O(\log n)$ logical qubits to a classical problem, $\symplpn[n, n, p]$, and then construct public-key cryptography from $\symplpn[n, n, p]$.
Whether or not either an improved reduction with larger $k$, or a more direct construction from $\slsn$, remains open.
A further open constructive avenue is security in different models.
For example, \textsf{OT} secure in the \emph{universal composability} model can be constructed from the hardness of $\lpn$~\cite{david2014universally}; can an analogous construction be made from the hardness of $\symplpn$?

A second question concerns improving our understanding of the comparative security between $\symplpn[n, n, p]$ and $\lpn[k, n', p']$, when $p = O(1/\sqrt{n})$, and thereby understanding how $\symplpn$ stands as a post-quantum security assumption relative to $\lpn$.
Is there a reduction from $\lpn$ to $\symplpn$ in this regime, and is there \emph{no} reduction in the opposite direction?
In this work, we have shown significant barriers towards the possibility of a converse reduction---from $\symplpn$ to $\lpn$ (which would nullify the motivation to build cryptography from $\symplpn$)---but a more rigorous proof that no reduction exists would substantially strengthen our hope that $\symplpn$ may remain secure even in a world where $\lpn$ is broken.
More generally, understanding precisely how symplectic structure affects the complexity of decoding remains a key open question.

\section{Preliminaries}
A comprehensive introduction to quantum computation can be found in~\cite{nielsen2010quantum}. For a detailed introduction to quantum error correction, we refer to~\cite{gottesman2024qecc}.
\paragraph{Computational model.} Unless stated otherwise, hardness and computational security are defined against non-uniform quantum polynomial-time (QPT) adversaries. Such an adversary is a family of quantum circuits of size $\poly(n)$ supplied with an advice state on $\poly(n)$ qubits; the circuit and advice may depend on $n$ but not on the freshly sampled problem instance. Computational indistinguishability is understood against the same class of distinguishers.
\paragraph{Pauli matrices.} The four $2\times 2$ \emph{Pauli matrices} are denoted as
    \begin{align}
        \vec I = \begin{pmatrix} 1 & 0 \\0 & 1 \end{pmatrix}, \quad\quad\,\,\,\,\vec X = \begin{pmatrix} 0 & 1 \\ 1 & 0 \end{pmatrix} \quad\quad\,\,\,\,\vec Y = \begin{pmatrix} 0 & -i \\ i & 0 \end{pmatrix}, \quad\quad \vec Z = \begin{pmatrix} 1 & 0 \\ 0 & -1 \end{pmatrix}.
    \end{align}
The phase-free Pauli group $\CP_n$ is the group of $n$-qubit Pauli operators $\set{\vec I, \vec X, \vec Y, \vec Z}^{\otimes n}$ which multiply modulo phase.
Any such $n$-qubit Pauli can be represented in so-called \emph{symplectic form}, via a map \begin{align}
    \operatorname{Symp} \,:\, \CP_n \to \Z_2^{2n} .
\end{align}
Here, bits $j$ and $j+n$ in $\operatorname{Symp}(\vec P)$ indicate which Pauli is on the $j$th qubit of $\vec P \in \CP_n$.
The two bits are $(0, 1)$ if $\vec X$, $(1, 0)$ if $\vec Z$, $(1, 1)$ if $\vec Y$, and $(0, 0)$ if $\vec I$.
Since $\vec{XY} = i\vec Z$, $\vec{YZ} = i\vec X$, and $\vec{ZX} = i \vec Y$, this representation shows that \begin{align}
    \operatorname{Symp}(\vec P_1 \vec P_2) = \operatorname{Symp}(\vec P_1) + \operatorname{Symp}(\vec P_2) 
\end{align}
and thus that $\Symp$ is an isomorphism (with arithmetic mod 2).
The symplectic representation is equipped with a natural inner product given by \begin{align}
    (\vec a_1, \vec b_1) \odot (\vec a_2, \vec b_2) := \vec a_1 \cdot \vec b_2 + \vec a_2 \cdot \vec b_1 \Mod{2},
\end{align}
where $\cdot$ denotes the standard dot product.
Two phase-free Paulis commute if and only if their symplectic representations have zero symplectic inner product.
We say that such vectors are symplectically orthogonal.
The standard symplectic basis is given by $\vec e_1, \dots, \vec e_n, \vec f_1, \dots, \vec f_n$, where $\vec e_i$ is $1$ on the $i$th entry and 0 elsewhere, and $\vec f_i$ is $1$ on the $(i+n)$th entry and zero elsewhere.
Note that \begin{align}
    \vec e_i \odot \vec e_j = \vec f_i \odot \vec f_j = 0 ,\quad \vec e_i \odot \vec f_j = \d_{ij} .
\end{align}

\paragraph{Quantum stabilizer codes.} Letting $\mathrm{U}(2^n)$ denote the unitary group on $n$ qubits, we define the Clifford group $\CC_n$ as the set of unitaries which leave the Pauli group invariant under conjugation: \begin{align}
    \CC_n := \set{\vec C \in \mathrm{U}(2^n) \,|\, \vec C \vec P \vec C^\dag \in \CP_n, \,\forall \, \vec P \in \CP_n} .
\end{align}
An $\llbracket n, k \rrbracket$ quantum stabilizer code is the simultaneous $+1$ eigenspace of $n-k$ independent commuting $n$-qubit Pauli operators.
Such a code can be succinctly specified by a Clifford operator $\vec C$ (e.g., in the form of a quantum circuit) as well as a choice of logical dimension $k$.
In this case, the code space is given by $\vec{C} (\ket{0^{n-k}} \otimes \ket{\psi})$, where $\ket{\psi} \in (\C^2)^{\otimes k}$ is any $k$-qubit logical state.
The stabilizers of the subspace of states $\ket{0^{n-k}} \otimes \ket{\psi}$ are $\vec{Z}_1, \dots, \vec{Z}_{n-k}$---hence, $\vec{CZ}_1\vec{C}^\dag, \dots \vec{CZ}_{n-k}\vec{C}^\dag$ are stabilizers for code-vectors $\vec{C} (\ket{0^{n-k}} \otimes \ket{\psi})$. A random quantum stabilizer code can be generated by using a uniformly random Clifford operator $\vec C \sim \mathcal{C}_n$ as the encoding map~\cite{poremba2024learning,berg2021simplemethodsamplingrandom}.

If we ignore phases, a Clifford is defined entirely on how it acts (via conjugation) on the $2n$ Paulis $\vec X_1, \dots, \vec X_n, \vec Z_1, \dots, \vec Z_n$.
Thus, a Clifford modulo phases on Paulis can be represented as a matrix in $\Z_2^{2n \times 2n}$.
Since conjugation by a Clifford does not affect the commutation relation of two Pauli operators, the symplectic representation of Cliffords preserves the symplectic inner product.
That is, $(\vec{C}\vec{v}) \odot (\vec{Cw}) = \vec{v} \odot \vec w$ for all $\vec v, \vec w \in \Z_2^{2n}$ if and only if $\vec C \in \Z_2^{2n \times 2n}$ represents a Clifford operation.
These matrices are referred to as symplectic matrices.
For a Clifford $\vec{C}$, we often abuse notation by referring to its matrix representation as $\vec{C}$ as well, instead of $\Symp(\vec{C})$. 

\paragraph{Symplectic linear algebra.}

Given a subspace $S \subseteq \Z_2^{2n}$, we denote its \emph{symplectic dual} or \emph{symplectic orthogonal complement} by \begin{align}
    S^\perp := \set{\vec v \in \Z_2^{2n} \,|\, \vec v \odot \vec w = 0,\,\,\,\forall \vec w \in S} .
\end{align}
Note that $\text{dim }S + \text{dim }S^\perp = 2n$ even though $S^\perp \cap S$ is not necessarily trivial. 

\begin{definition}[Isotropic subspaces and matrices] \label{def:isotropic_subspace}
    We say that a subspace $S$ of $\Z_2^{2n}$ is \emph{isotropic} if every pair of points in $S$ is symplectically orthogonal, i.e., $S \subseteq S^\perp$.
    Note that the maximum dimension of an isotropic subspace is $n$.
    A matrix is isotropic if its image is an isotropic subspace. Equivalently, a matrix is isotropic if its columns are pairwise symplectically orthogonal. 
\end{definition}
We remark that an equivalent characterization of an isotropic matrix $\vec M = \begin{bmatrix} \vec M_1 \\ \vec M_2 \end{bmatrix}$ is that $\vec M_2^\intercal\vec M_1 + \vec M_1^\intercal\vec M_2 = 0$. 
In other words, $\vec M$ is isotropic if and only if $\vec M_1^\intercal\vec M_2$ is symmetric. 

\paragraph{Quantum and classical noise models.}
The $n$-qubit depolarizing noise distribution with error parameter $p \in (0,1)$ is denoted as $\CD_p^{\otimes n}$.
This distribution independently draws a Pauli error on each qubit according to a distribution $\CD_p$; namely, with probability $1-p$, no error occurs, otherwise, with probability $p$, a uniformly random Pauli $\set{\vec X, \vec Y, \vec Z}$ is drawn and applied.
Note that $\CD_{3/4}$ coincides with a uniformly random pair of bits. In slight abuse of notation, we frequently use the same notation in the symplectic representation: we use $\vec e \sim \CD_p^{\otimes n}$ as shorthand notation for first sampling a Pauli operator $\vec E \sim \mathcal{D}_p^{\otimes n}$ and letting $\vec{e} = \Symp(\vec E)$. The resulting error vector $\vec e$
has length $2n$, with each pair $(e_i, e_{i+n})$ drawn independently;
Each such pair is $(0, 0)$ with probability $1-p$, and the other 3 possibilities each with probability $p/3$.
By contrast, a vector $\vec e \sim \Ber_p^{\otimes n}$ has length $n$, with each bit being $1$ independently with probability $p$, and $0$ otherwise.
For a bitstring $\vec v \in \Z_2^d$, let $w_H(\vec v) := |\{i \in [d] : v_i \neq 0\}|$ be its Hamming weight. For $\vec e \in \Z_2^{2n}$, let $w_P(\vec e) := |\{i \in [n] : (e_i,e_{i+n}) \neq (0,0)\}|$ be its Pauli weight. Then $w_P(\vec e) \leq w_H(\vec e) \leq 2w_P(\vec e)$. We write $|\vec v| = w_H(\vec v)$ for bitstrings.
\paragraph{Haar measure over states.}
To generate a random logical state $\vec x \in \Z_2^k$, say of a classical linear code, it suffices to choose $\vec x$ uniformly at random.
Quantumly, one way to analogously generate a random logical state is to choose a uniformly random $k$-qubit state.
The notion of a uniformly random state is given by the Haar measure (see~\cite{mele2024introduction}), formally defined for the unitary group.

\begin{definition}[Haar measure] \label{def:haar_measure}
    The Haar measure $\mu$ over the group of $d \times d$ unitary matrices $\text{U}(d)$ is the unique probability measure which is invariant under translation. That is, for all $\vec{V} \in \text{U}(d)$ and integrable functions $f$,
    \begin{align}
        \int f(\vec{V}\vec{U}) \, d\m(\vec{U}) = \int f(\vec{U}) \, d\m(\vec{U}) = \int f(\vec{U}\vec{V}) \, d\m(\vec{U}) .
    \end{align}
\end{definition}
We refer to a \emph{Haar-random state} $\ket{\psi} \sim \mu_k$ to be a state constructed by sampling $\vec{U} \sim \mu$ a unitary drawn from the Haar measure over $\text{U}(2^k)$ and then letting $\ket{\psi} = \vec{U} \ket{0^k}$.

\subsection{Learning Stabilizers with Noise}

The problem of decoding a random quantum error-correcting code, is known as \emph{Learning Stabilizers with Noise}~\cite{poremba2024learning}---the quantum analog of \emph{Learning Parity with Noise} ($\lpn$) which we stated in its \emph{decision} variant: decide whether its input is a true noisy codeword or simply a uniformly random bitstring, as opposed to decoding a given noisy codeword (for a formal definition, see \Cref{app:more_on_lpn}).

For \emph{Learning Stabilizers with Noise}, there are two variants which we will discuss, termed $\slsn$ and $\lsn$. 
Both computational problems characterize the hardness of decoding an encoded quantum state, where the stabilizer code is sampled uniformly at random, the starting logical state is sampled from a specified distribution, and a random depolarizing error is applied to the state. 
In $\slsn$, the logical state is a Haar-random pure state: the problem serves as a natural model for average-case quantum stabilizer decoding. 
Meanwhile, for $\lsn$, the logical state is chosen as a uniformly random computational basis state. 
Because $\lsn$ restricts itself to encoding simpler states, it is easier to characterize.
In particular, it is shown in \cite{khesin2025average} that $\lsn$ is (quantumly) equivalent to a classical computational task. 
Below, we define $\slsn$, and then $\lsn$ in its equivalent classical form. 
We note that in \cite{khesin2025average}, various forms of these two problems are discussed, including a Decision version and a variant with more than one sample. 
The definitions provided here coincide with Search $\slsn$ and Search $\lsn$ with just one sample.

\begin{definition}[Learning Stabilizers with Noise, state variant $\slsn$] \label{def:search-sLSN}
The $\mathsf{state}$ variant of the Learning Stabilizers with Noise problem, denoted by $\slsn(k,n,p)$, is characterized by integers $k,n \in \mathbb{N}$ with $1 \leq k \leq n$ and $p \in (0,1)$. Both $p$ and $k$ can vary with $n$. Given as input a sample of the form 
    \begin{align}
        (\vec C, \; \vec E \vec C\ket{0^{n-k},\psi}),
    \end{align}
where $\vec C \sim \CC_n$ is a uniformly random $n$-qubit Clifford operator (admitting a classical description via, e.g., a circuit), $\vec E \sim \mathcal{D}_p^{\otimes n}$ is an $n$-qubit Pauli sampled from a depolarizing distribution with parameter $p$, and $\ket{\psi} \sim \mathfrak{\mu}_k$ is a Haar random $k$-qubit state, the task is to output a quantum state $\rho_{\psi}$  within average fidelity at least $\frac{1}{2^k}+\frac{1}{\poly(n)}$ of $\ket{\psi}$ over the choice of $\ket{\psi} \sim \mu_k$; that is,
\begin{align}
    \underset{\ket{\psi} \sim \mu_k}{\mathbb{E}}[\braket{\psi|\rho_{\psi}|\psi}] \geq \frac{1}{2^k} + \frac{1}{\poly(n)}.
\end{align} Here $\rho_{\psi}$ denotes the decoder's output state, averaged over $\vec C$, $\vec E$, and the decoder's internal randomness and measurement outcomes for fixed $\ket{\psi}$.
\end{definition}

Next, we formally define the classical formulation of the $\lsn$ problem. 
\begin{definition}[Learning Stabilizers with Noise, classical representation]\label{def:classical-search-LSN}
$\lsn(k, n, p)$, is characterized by integers $k,n \in \mathbb{N}$ with $1 \leq k \leq n$ and $p \in (0,1)$. 
Here, the input is
    \begin{align}
        \Big(\begin{bmatrix}
      \vec A \, | \, \vec B  
    \end{bmatrix}, \;
    \begin{bmatrix}
      \vec A \, | \, \vec B  
    \end{bmatrix} \cdot \begin{bmatrix}
        \vec r \\
        \vec y
    \end{bmatrix} + \vec e \Big)
    \end{align}
    where $\vec A \in \Z_2^{2n \times n}$ and $\vec B \in \Z_2^{2n \times k}$ are sampled jointly and uniformly subject to the constraint that $\vec A$ and $\vec B$ are isotropic and that $\begin{bmatrix}
      \vec A \, | \, \vec B  
    \end{bmatrix}$ is full-rank; where $\vec r \sim \Z_2^n$ and $\vec y \sim \Z_2^k$ are uniformly random, and where $\mathbf{e}\in \Z_2^{2n}$ is a depolarizing error with parameter $p$, i.e. $\Symp^{-1}(\vec e) \sim \CD_p^{\otimes n}$ (for brevity, we sometimes write $\vec e \sim \CD_p^{\otimes n}$ instead).
    The task is to output $\vec y$ with probability at least $1/2^k + 1/\poly(n)$.
\end{definition}
$\slsn[k, n, p]$ is always at least as hard as $\lsn[k, n, p]$, and in fact the two problems are equivalent when $k = O(\log n )$.
These reductions are necessarily quantum, as the input and output of $\slsn$ include quantum states.
\begin{theorem}[$\slsn$ versus $\lsn$, \cite{khesin2025average}] \label{thm:slsn_lsn_equivalence}
Let $k, n, \in \mathbb{N}$ and $p \in (0, 1)$ such that $k = O(\log n)$. Then there exist efficient quantum reductions from $\lsn[k, n, p]$ to $\slsn[k, n, p]$ and from $\slsn[k, n, p]$ to $\lsn[k, n, p]$. 
\end{theorem}

The problem of $\slsn[k, n, p]$ (or, equivalently, $\lsn[k, n, p]$) with $k = O(\log n)$ will be the basis for the hardness of the encryption and oblivious transfer schemes that we present. 
However, it is challenging to work directly with these problems because of their complicated setup.
As a result, we will instead define a simpler problem which bears a closer resemblance to $\lpn$.
In the next section, we will prove that in fact $\lsn$ with $k = O(\log n)$ logical qubits reduces to this simpler problem, so that in combination with \Cref{thm:slsn_lsn_equivalence}, our schemes are secured solely by the hardness of $\slsn[k = O(\log n), n, p]$.
For binary distinguishing tasks, advantage means $|\Pr[\hat b=b]-1/2|$, where $b$ is a uniform bit selecting the input distribution and $\hat b$ is the distinguisher's guess.
\begin{definition}[Decision $\symplpn$] \label{def:sympLPN}
Let $n, k \in \mathbb{N}$ with $1 \leq k \leq n$ and $p \in (0,1)$ be parameters. Then, $\symplpn[k, n, p]$ is the task of distinguishing with advantage at least $1/\poly(n)$ between
    \begin{align}
        (\vec{A} , \;\vec{A}\mathbf{x} + \mathbf{e} ) \quad \text{ or } \quad (\vec{A} , \;\mathbf{u} \sim \Z_2^{2n})
    \end{align} 
    where $\vec A\in \Z_2^{2n \times k}$ is a uniformly random full-rank isotropic matrix, and where $\vec e$ is a depolarizing error with probability parameter $p$, i.e. $\Symp^{-1}(\vec e) \sim \CD_p^{\otimes n}$ (for brevity, we often write $\vec e \sim \CD_p^{\otimes n}$). Here $\vec x \sim \Z_2^k$ is uniformly random.
    The former case is called \emph{structured} and the latter called \emph{unstructured}.
    For convenience, we denote $\symplpn[n, n, p]$, a parameter regime of particular relevance, as $\symplpn[n, p]$.
\end{definition}

If $\vec A$ was instead chosen to be a uniformly random matrix, then the above problem would resemble Decision $\lpn$.
A random $2n \times k$ matrix is full-rank with at least constant probability, so the isotropy condition is the fundamental difference between $\lpn$ and $\symplpn$. 
$\symplpn$ was first defined in \cite{khesin2025average} as a tool for reducing $\lpn$ to $\lsn$---\cite{khesin2025average} demonstrates that for any regime of $k$, $\symplpn[n, p]$ reduces to $\lsn[k, n, p]$, and moreover that there is a reduction from $\lpn$ to $\symplpn$ for certain choices of $p$, namely $p = \omega(n^{-1/2})$.

\section{Hardness of $\symplpn$}
\label{sec:hardness_of_symplpn}

We will construct our public-key cryptographic protocols directly from $\symplpn(n, p)$ for certain $p$.
As a consequence, our first result is a reduction from $\slsn$ with very few logical qubits to $\symplpn$.
Therefore, cryptography built on the hardness of $\symplpn$ source their security from the hardness of decoding random quantum stabilizer codes with Haar-random logical states with logarithmically many qubits.
Classically, there is a polynomial time algorithm which solves $\lpn(k, n, p)$ when $k = O(\log n)$ by brute force.
However, quantumly, there is substantial evidence that $\lsn(k, n, p)$ is exponentially hard (in $n$) at constant noise rates \emph{for any} $k \geq 1$~\cite{khesin2025average}, intuitively because there is no known efficient way to verify proposed solutions for any $k$ so a natural brute-force approach is to enumerate possible errors rather than logical states.
Therefore, the fact that $k = O(\log n)$ in our codes does not imply that security is lost. 

Our public-key encryption will, however, rely on the hardness of both $\symplpn[n, n, p] = \symplpn[n, p]$ and $\symplpn[n-1, n, p]$.
A reduction from $\lpn[k, n, p]$ to $\lpn[k-1, n, p]$ is trivial for $\lpn$, but is surprisingly challenging for $\symplpn$.
Nonetheless, we will then show how to give such a reduction by introducing some new analytical techniques for isotropic subspaces.
As a consequence, the public-key encryption scheme will only rely on the hardness of $\symplpn[n, p]$, as desired.

\subsection{Reducing $\lsn$ for very few logical qubits to $\symplpn$}

Since $\slsn$ and $\lsn$ are equivalent for $k = O(\log n)$ (see \Cref{thm:slsn_lsn_equivalence}), we here reduce $\lsn[k, n, p]$ to $\symplpn[n, p]$.

\begin{theorem}[$\lsn$ reduces to $\symplpn$] \label{thm:lsn_to_symplpn}
    Let $k, n \in \mathbb{N}$ such that $k = O(\log n)$ and let $p \in (0, 1)$.
    Suppose $\CO$ is an oracle which solves $\symplpn[n, p]$. 
    Then there exists a polynomial time algorithm which solves $\lsn(k, n, p)$, using a single call to $\CO$.
\end{theorem}

\begin{proof}
    Let $(\vec A, \vec B, \vec z)$ be an instance of $\lsn$, where $\vec z = \vec {Ar} + \vec{B y} + \vec e$. 
    The algorithm queries the oracle $\CO$ with input $(\vec A, \vec z)$.
    If the oracle outputs \texttt{structured}, then return $\vec{y} = \vec 0$, while if the oracle outputs \texttt{unstructured}, return a uniformly random value of $\vec{y} \in \{0, 1\}^k \backslash \{\vec 0\}$.  
    
    We analyze this algorithm in two cases: either $\vec y = \vec 0$ or $\vec y \neq \vec 0$. 
    If $\vec y = \vec 0$, then $(\vec A, \vec z) = (\vec A, \vec A \vec r + \vec e)$ is precisely a structured instance of $\symplpn$, since marginally $\vec A$ is a uniformly random full-rank isotropic matrix. 
    Now, say that $\vec y \neq 0$. We can assume without loss of generality that $\vec{y} = (1, 0, \dots, 0)^\intercal$. 
    Indeed, for any given $W = \operatorname{im}\vec{B}$, each nonzero vector in $W$ is equally likely to be the first column of $\vec{B}$, $\vec{By}$. 
    The distribution of $\vec{By}$ for $\vec{y} = (1, 0, \dots, 0)^\intercal$ and $\vec{By}$ for uniformly random nonzero $\vec{y}$ are both just a uniformly random nonzero vector in $W$. 
    
    Now, observe that the first column of $\vec{B}$ is equally likely to be any vector in $\mathbb{Z}_2^{2n}\backslash V$, where $V = \operatorname{im}\vec{A}$.
    For example, $\vec{B}$ may be sampled column by column by sampling a new vector linearly independent from $\vec{A}$ and symplectically orthogonal to the previous columns of $\vec{B}$, giving the first column this simple form. 
    Thus in the $\vec y \neq \vec 0$ case, $\vec z \sim (\vec{Ar} + \vec e) + \vec b_1$, where $\vec b_1$ is the first column of $\vec B$, uniformly random over $\Z_2^{2n} \setminus V$, and $\frac{|V|}{2^{2n}} = 2^{-n}$.
    Consequently, $\vec b_1$ has total variation distance $2^{-n}$ from that of uniformly random, using the fact that the TV distance between uniform distributions over sets $S$ and $T \subseteq S$ is $1 - \frac{|T|}{|S|}$.
    Hence, the total variation distance between $(\vec{A}, \vec{z})$ and $(\vec{A}, \vec{u})$, where $\vec{u}$ is uniformly random, is also $O(2^{-n}) = \negl(n)$. 
    In summary, therefore, if $\vec y = \vec 0$, then the oracle $\CO$ is given a sample from the \texttt{structured} distribution of $\symplpn[n, p]$; if $\vec y \neq \vec 0$, then $\CO$ is given a sample drawn from a distribution within $\negl(n)$ TV distance from the \texttt{unstructured} distribution of $\symplpn[n, p]$.
    
    To complete the reduction, we calculate the success probability of our algorithm for $\lsn[k, n, p]$. 
    Let $a := \Pr[\mathcal O(S)=\texttt{structured}]$ and $b := \Pr[\mathcal O(U)=\texttt{structured}]$, where $S$ and $U$ denote the structured and unstructured challenge distributions. By complementing the oracle output if necessary, assume $a-b=2\epsilon$ for some $\epsilon \geq 1/\poly(n)$.
    If $\vec y=\vec 0$, the algorithm succeeds with probability $a$. If $\vec y\neq\vec 0$, it succeeds with probability at least $\frac{1-b-\negl(n)}{2^k-1}$.
    The probability of correctness over random $\vec y$ is therefore at least
    \begin{align}
    \frac{a}{2^k} + \frac{2^k-1}{2^k}\frac{1-b-\negl(n)}{2^k-1} = \frac{1}{2^k}+\frac{2\epsilon}{2^k}-\negl(n),
    \end{align}
    %     \begin{align}
    % &\frac{1}{2^k}\left(\frac{1}{2} + \epsilon\right) + \frac{2^k-1}{2^k} \frac{1}{2^k - 1}\left(\frac{1}{2} + \epsilon - \negl(n)\right)\\&= \frac{1}{2^k}\left(1 + 2\epsilon - \negl(n)\right)\\
    % &= \frac{1}{2^k} +  \frac{1}{2^{k} \poly(n)} \, = \, \frac{1}{2^k} + \frac{1}{\poly(n)},
    % \end{align}
    which exceeds $2^{-k}$ by an inverse-polynomial amount because $k=O(\log n)$ and $\epsilon \geq 1/\poly(n)$.

\end{proof}

\subsection{Reducing logical bits in $\symplpn$ without reducing hardness}
\label{sec:reducing_one_logical_sympLPN}

In this section, we tighten the hardness of $\symplpn$ by reducing $\symplpn(n, p) = \symplpn(n, n, p)$ to $\symplpn(n-1, n, p')$ for slightly larger $p'$.
In particular, this reduction barely changes the parameters, simply reducing the number of logical bits by 1. 
However, it is surprisingly technical and essential for our cryptographic constructions.
To give a baseline intuition for this reduction, we show in \Cref{app:more_on_lpn} a simple analogous reduction for $\lpn$, which proceeds by simply removing the last column of the code matrix.
Unfortunately, as discussed in \Cref{sec:technical_overview}, this proof technique does not extend to $\symplpn$ because the last column of the $\symplpn$ code matrix depends heavily on the rest of the matrix.
Thus, our reduction technique proceeds via a completely different approach.

Our proof will proceed in three steps.
We begin with a sample of $\symplpn[n, p]$, $(\vec{A}, \vec{u})$, where $\vec{u}$ is either structured (i.e. of the form $\vec{Ax}+\vec{e}$) or uniformly random. 
Letting $V = \operatorname{im}(\vec{A})$, we consider the subspace $W := \text{span}(\vec f_{1})^{\perp} \cap V$, i.e. the hyperplane in $V$ of vectors symplectically orthogonal to $\vec f_{1}$, the $(n+1)$st standard basis vector in $\Z_2^{2n}$. 
$W$ has dimension $n-1$ with overwhelming probability (so we will assume it does), and it contains precisely the vectors in $V$ whose first coordinate is zero. 
Suppose $\vec{u}$ is structured.
If the first physical bit has no error, i.e. $e_1 = 0$, then $\vec{Ax} \in W$ if and only if the first bit of $\vec{Ax} + \vec e$ is 0.
We wish to transform $\vec u$ to some $\vec u'$ such that $\vec u'$ is of the form $\vec{Ax}' + \vec e$ but $\vec{Ax}' \in W$.
If $\vec{Ax} \in W$, then no further action is required, and we set $\vec u' = \vec u$. 
If $\vec{Ax}$ is \textit{not} in $W$, by adding some $\vec{v} \in V \backslash W$ we can obtain $\vec{u'} := (\vec{Ax} +\vec{v}) + \vec{e} = \vec{A x}' + \vec e$ such that $\vec{Ax}' \in W$.
Note that if instead $\vec u$ were unstructured, this transformation would make $\vec u'$ uniformly random on $\vec f_1^\perp$, the hyperplane of vectors with first coordinate zero.

After this step, under the assumption that $e_1 = 0$, we have almost prepared a sample of $\symplpn$ with $n-1$ logical qubits. 
Indeed, letting $\vec{B}$ be a code with $W$ as its codespace, we can prepare $\vec{u'}$ so that:
\begin{enumerate}
\item If $\vec{u}$ is uniformly random, then $\vec{u'}$ is uniformly random on $\vec f_1^\perp$.
\item If $\vec{u}=\vec{Ax}+\vec{e}$, then $\vec{u'} = \vec{By} + \vec{e}$ for uniformly random $\vec{y}$. 
\end{enumerate}
However, $W$ is \textit{not} a uniformly random $(n-1)$-dimensional isotropic vector space, because it is always symplectically orthogonal to a known vector, namely $\vec{f}_{1}$.
Therefore, $\vec{B}$ is not a uniformly random matrix with symplectically orthogonal columns. 
In the second step of the reduction, we design a sparse Clifford operator $\vec{C}$ which randomly rotates the normal vector $\vec f_{1}$ (the hyperplane then follows along).
This operation fully scrambles the code, so that $\vec{CB} \in \Z_2^{2n \times (n-1)}$ is statistically indistinguishable from a uniformly random matrix with symplectically orthogonal columns. 

This second step converts a structured sample $\vec{u'}$ to $\vec{C}\vec{u'}=\vec{CBy} + \vec{Ce}$, however, which damages the distribution of the noise $\vec{e}$. 
However, since $\vec{C}$ is sparse, as long as $e_{n+1} = 0$, i.e. no error occurs on the $(n+1)$st physical bit, $\vec{Ce}$ is low-weight.
In this case, we are able to re-randomize to transform $\vec{Ce}$ back into depolarizing noise, at the cost of a slightly larger noise parameter $p'$.
Now, we may query the $\symplpn[n-1, n, p']$ oracle, which will output \texttt{structured} or \texttt{unstructured}, and the reduction can output the same for the sample $(\vec{A}, \vec{u})$. 

This completes the reduction, but only under the assumption that a good event occurs, namely that $e_1 = e_{n+1} = 0$. 
However, interestingly, this assumption can be completely circumvented. 
In the original sample, $(e_1, e_{n+1})$ follows a 1-qubit depolarizing distribution with parameter $p$.
We are, however, free to increase this noise parameter to any $q > p$ by adding additional error ourselves.
For example, we could create an alternate reduction by first adding completely random bits to increase the noise on $(e_1, e_{n+1})$ to that of $\CD_{3/4}$, before executing the reduction outlined above.
We show that at least one of these two reductions must succeed, or else the reduction could not possibly succeed when the good event \emph{does} occur.
This is a contradiction, and thus either the original reduction or the alternate one with extra noise added must reduce $\symplpn[n, p]$ to $\symplpn[n-1, n, p']$.

To establish the above argument rigorously, we begin by constructing the sparse, randomizing distribution of Cliffords.
Recall that $\vec e_1, \dots, \vec e_n, \vec f_1, \dots,  \vec f_{n}$ denotes the standard symplectic basis on $\Z_2^{2n}$.

\begin{definition}[Symplectic hyperplane rotation] \label{def:symp_hyperplane_rotation}
For $n \geq 2$, a \emph{random symplectic hyperplane rotation} is an efficiently sampleable ensemble of matrices $\mathcal{R}_n$ in $\Z_2^{2n \times 2n}$ constructed as follows.
Sample a uniformly random $\vec r \sim \Z_2^{2n}$. 
Let $k$ be the first (positive) index for which $r_{n+k} = 1$ (if there is no such index, output $\vec{C} := \vec I$). 
Let $\vec r'$ be the same as $\vec r$, except entries $1$ and $k$ are swapped, and entries $n+1$ and $n+k$ are swapped.
Let the map $\vec C_0$ execute the following map on standard basis vectors, where $\vec e_i, \vec f_i$ are respectively the $i$th and $(n+i)$th standard basis vectors.
\begin{align}
    \vec e_1 & \longmapsto \vec e_1, \\
    \vec f_1 & \longmapsto \vec r', \\
    \vec e_j & \longmapsto \vec e_j + (\vec r' \odot \vec e_j)\vec e_1\quad \text{for $j$ from $2$ to $n$} , \\
    \vec f_j & \longmapsto \vec f_j + (\vec r' \odot \vec f_j)\vec e_1\quad \text{for $j$ from $2$ to $n$}
\end{align}
Define a matrix $\vec{\Pi}$ which swaps $(\vec e_1, \vec f_1)$ and $(\vec e_k, \vec f_k)$, and otherwise for each $j \neq 1, k$ acts as identity on $\vec e_j, \vec f_j$. 
Then, output $\vec{C} := \vec{\Pi}\vec{C_0}$. 
Note that $\vec C\vec f_1=\vec r$ whenever such an index $k$ exists.

\end{definition}
This rotation, by randomizing $\vec{f}_1$ up to negligible total variation distance, also randomizes the hyperplane that is symplectically orthogonal to it.
The remaining relations serve only to ensure that the map is a sparse symplectic matrix (i.e., that it preserves the symplectic inner product).
We claim that $\vec C$ is a valid symplectic matrix which maps the hyperplane orthogonal to $\vec f_1$ to a hyperplane whose distribution is within negligible total variation distance of uniform.

\begin{lemma}[Randomizing symplectic hyperplanes] \label{lemma:randomizing_symp_hyperplanes}
Denote by $\CR_n$ the ensemble of random symplectic hyperplane rotations from \Cref{def:symp_hyperplane_rotation}.
Sample $\vec C \sim \CR_n$.
Then the following properties hold.
\begin{enumerate}
    \item[(1) ] $\vec C$ is a symplectic matrix.

    \item[(2) ] For $\vec v_1, \dots, \vec v_{n-1} \in \Z_2^{2n}$ a uniformly random basis of an $(n-1)$-dimensional isotropic subspace $V$ (see \Cref{def:isotropic_subspace}) which is uniformly random subject to $\vec f_1\in V^\perp\setminus V$, let $\vec w_i := \vec C \vec v_i$.
    Then the distribution of $(\vec w_1, \dots, \vec w_{n-1})$ is within total variation distance $\negl(n)$ to a uniformly random basis of a random $(n-1)$-dimensional isotropic subspace.
\end{enumerate}
\end{lemma}

The proof is deferred to \Cref{app:additional_proofs}.
Before we proceed to the reduction, we also require a lemma that total depolarizing noise, applied to only certain qubits, can be scrambled into completely symmetric depolarizing noise on all qubits.
A similar technique was used in \cite{khesin2025average} to reduce from $\lpn$ to $\symplpn$.

\begin{lemma}[Noise symmetrization] \label{lemma:noise_symmetrization}
Let $\vec e \in \Z_2^{2n}$, and let $\pi \sim S^n$ be a random permutation on $n$ elements.
We define $\pi$ to act on elements of $\Z_2^{2n}$ by executing the same permutation on the first and last set of $n$ indices.
Suppose that $(e_j, e_{n+j}) = (0, 0)$ for all $j$ except in a known set $M$ of size $m = \omega(\log n)$ with $m \leq 3n/4$; for $j \in M$, $(e_j, e_{n+j}) \sim \CD_{3/4}$ independently (i.e., random bits).
Then, for any $q\in[\frac{m}{n},\frac34]$, there exists a distribution $\mu_n$ over $\Z_2^{2n}$, sampleable in time $\poly(n)$, such that $\pi(\vec e + \vec e')$ for $\vec e' \sim \mu_n$ is within total variation distance $\negl(n)$ from $\CD_{q}^{\otimes n}$.
\end{lemma}

The proof is deferred as well to \Cref{app:additional_proofs}.
We also record a relevant lemma about the product of depolarizing distributions. 
\begin{lemma}[Depolarizing convolution, \cite{khesin2025average}] \label{lemma:depolarizing_convolution}
    Let $\vec P_1, \vec P_2 \in \CP_1$ be single-qubit phase-free Paulis such that $\vec P_1 \sim \CD_p$ and, independently, $\vec P_2 \sim \CD_u$ where $u = \frac{q-p}{1 - \frac{4}{3} p}$ for any $p \in [0, \frac34)$ and $q \in [p, \frac34]$.
    Then $\vec P_1 \vec P_2 \sim \CD_{q}$.
\end{lemma}

\begin{theorem}[Reducing by one logical bit in $\symplpn$] \label{thm:sympLPN_reducing_logicals}
    Let $p \in (0, 1)$ be such that $\frac{3}{4} - p = \W(1)$, and let $n \in \mathbb{N}$ be sufficiently large.
    Suppose $\CO$ is an oracle which solves Decision $\symplpn(n-1, n, p + \log^2 n / n)$.
    Then there exists an algorithm running in time $\poly(n)$, which solves Decision $\symplpn(n, p)$, using a single call to $\CO$.
\end{theorem}

\begin{proof}
Suppose first that the input is structured.
That is, we are given $(\vec A, \vec u := \vec{Ax} + \vec e)$, where $\vec A \in \Z_2^{2n \times n}$ is a uniformly random basis of a uniformly random isotropic subspace, $\vec x \sim \Z_2^{n}$, and $\vec e \sim \CD_p^{\otimes n}$.
Define $V := \operatorname{im}(\vec A)$.
Since $V=V^\perp$ is uniformly random, $\Pr[\vec f_1\in V]=(2^n-1)/(2^{2n}-1)=1/(2^n+1)$.
We condition on $\vec f_1\notin V$, equivalently on at least one column of $\vec A$ having first entry $1$.
Define $V_0 \subseteq \operatorname{im}(\vec A)$ to be the subspace of $\operatorname{im}(\vec A)$ with vectors whose first entry is 0.
Having first entry 0 is equivalent to being symplectically orthogonal to $\vec f_{1}$.
Conditioning on this event, this is one extra nontrivial constraint, so $\dim(V_0)=n-1$ and $\vec f_1\in V_0^\perp\setminus V_0$. Moreover, $V_0$ is uniformly random subject to these conditions: each such $V_0$ has three Lagrangian extensions, corresponding to the three lines in the two-dimensional symplectic space $V_0^\perp/V_0$, exactly two of which do not contain $\vec f_1$ and intersect $\vec f_1^\perp$ in $V_0$.
Let $\vec A_0 \in \Z_2^{2n \times (n-1)}$ be a random basis matrix for $V_0$.
Suppose that a good event occurs, namely that $(e_1, e_{n+1}) = (0, 0)$.
Conditioned on this event, either $u_1 = 0$ or $u_1 = 1$.
If $u_1 = 1$, subtract from $\vec u$ the first column of $\vec A$ whose first entry is $1$.
If $u_1 = 0$, do nothing at this step.
Denote by $\vec u' \in \Z_2^{2n}$ the outcome of this step.
Note that conditioned on $\vec A_0$, $\vec u'$ is equidistributed as $\vec A_0 \vec x' + \vec e'$, where $\vec x' \sim \Z_2^{n-1}$ and $\vec e'$ is independent depolarizing noise on pairs $(2, n+2), \dots, (n, 2n)$ but $(e'_1, e'_{n+1}) = (0, 0)$.
Thus, we can equivalently express $\vec u' = \vec A_0 \vec x' + \vec e'$.
Set $q:=\frac{\log^2 n}{n(1-\frac43p)}$, $m:=\lfloor nq\rfloor$, and $M:=\{1,\dots,m\}$.

Next, we sample and apply a random symplectic hyperplane rotation Clifford $\vec C$ from the ensemble in \Cref{def:symp_hyperplane_rotation}.
We condition on the pivot $k$ of this rotation lying in $M$, an event of probability $1-2^{-m}$, where $\vec r=\vec C\vec f_1$.
We apply $\vec C$ to $\vec A_0$ and $\vec u'$.
By \Cref{lemma:randomizing_symp_hyperplanes}, $\vec B := \vec{C A}_0$ is statistically indistinguishable from a uniformly random basis of a uniformly random isotropic subspace of dimension $n-1$. This additional conditioning changes the distribution by at most $2^{-m}=\negl(n)$.
Meanwhile, $\vec C \vec u' = \vec{Bx}' + \vec C \vec e'$.
The distribution of $\vec C \vec e'$ can be derived from the structural form of $\vec C$.
Since $\vec e'$ is independent depolarizing noise $\mathcal{D}_p$ on the last $n-1$ qubits and $e'_1, e'_{n+1} = 0$, we may write $
    \vec e' = \sum_{j=2}^{n} (a_j \vec e_j + b_{j} \vec f_{j})$,
where $(a_j,b_j)\sim\CD_p$ independently for $2\leq j\leq n$, and we set $a_1=b_1=0$. Let $\boldsymbol{\Pi}$ be the permutation matrix from \Cref{def:symp_hyperplane_rotation}. Then, we can write
\begin{align}
    \vec C \vec e' & = \vec \Pi\left(c \vec{e}_1+ \sum_{j = 2}^n (a_j \vec e_j + b_j \vec f_j)\right)\\
    &= c \vec e_k + (a_k \vec{e}_1+b_k\vec{f}_1) + \sum_{j\neq k, 1} (a_j \mathbf{e}_j + b_j\vec{f}_j)
\end{align}
where the coefficient $c$ of $\vec{e}_k$ has a potentially complicated distribution correlated with the $(a_j, b_j)$.
That is, $\vec C \vec e'$ takes the form of independent $\CD_p$ noise on all but the $k$th pair, where the distribution is complicated.
To adjust for this, we will add uniformly random bits to the $k$th pair (i.e. sample from $\CD_{3/4}$), given by $\vec s_0 \in \Z_2^{2n}$ which is $0$ on all indices, except uniformly random on $k$ and $n+k$.
This step erases any correlation the $k$th pair currently has with the $(a_j, b_j)_{j \neq k}$ (since this pair is now independently and uniformly random).
The distribution of the error is therefore now $
    \vec{C e}' + \vec s_0 \sim \CD_p^{\otimes k-1} \otimes \CD_{3/4} \otimes \CD_{p}^{\otimes n - k}$.
It is convenient to represent this distribution as the sum of two random variables, one $\vec{e}'' \sim \CD_{p}^{\otimes n}$ and the other $\vec s$ which is the uniform distribution on the $k$th index and 0 elsewhere.
Now, the conditions of \Cref{lemma:noise_symmetrization} require that the error pairs $(e_j, e_{n+j})$ are 0 for all $j$ outside a known set of size $m = \omega(\log n)$.
In our case, add a vector $\vec s'\in\Z_2^{2n}$ which is uniformly random on the pairs in $M\setminus\{k\}$ and zero elsewhere.
Define $\vec t := \vec s + \vec s'$, a vector which has uniform bits on $m$ pairs and is $(0, 0)$ on all other pairs.
We add $\vec s'$ to our error, producing $\vec e'' + \vec t$.
Note that $\vec e''$ and $\vec t$ are independent. The rotation changes no error pair outside $M$, while the added masks make all pairs in $M$ uniform. Thus, conditional on $(\vec A,\vec A_0,\vec C)$, the error has the same distribution $\CD_{3/4}^{\otimes m}\otimes\CD_p^{\otimes(n-m)}$, so this representation can be chosen independently of $(\vec B,\vec x')$.
Applying \Cref{lemma:noise_symmetrization}, we may add an extra certain random vector $\vec t'$ and apply a random permutation $\pi$ on the $n$ index pairs such that $\pi(\vec t + \vec t')$ is within total variation distance $\negl(n)$ of $\CD_q^{\otimes n}$. We use the sampler from its proof, choosing the additional pairs in increasing order.
On the other hand, $\vec e''$ is permutation-invariant, so $\pi(\vec e'') \sim \vec e''$.
Let $p' = p + \log^2 n / n$.
Our choices satisfy $m/n\leq q=\frac{p'-p}{1-\frac43p}\leq\frac34$ and $m=\omega(\log n)$ for sufficiently large $n$, since $\frac34-p=\Omega(1)$.
Hence, by \Cref{lemma:depolarizing_convolution}, 
\begin{align}
    \operatorname{TV}\!\left(\mathcal L\!\left(\pi(\vec e''+\vec t+\vec t')\right),\CD_{p'}^{\otimes n}\right)\leq\negl(n),\qquad p'=p+\left(1-\frac43p\right)q,
\end{align}
Up to negligible total variation distance, $\vec B$ is a uniform isotropic basis independent of $(\vec x',\vec e'',\vec t,\vec t',\pi)$. Since its distribution is invariant under each fixed $\pi$, the pair $(\pi\vec B,\pi(\vec B\vec x'+\vec e''+\vec t+\vec t'))$ is therefore within negligible total variation distance of a structured $\symplpn(n-1,n,p')$ instance, conditional on the good event.
On the other hand, if the input were instead unstructured, then almost all of the operations we performed on $\vec u$ would leave its distribution---uniform over $\Z_2^{2n}$---invariant.
The only step which does not is the first, wherein we alter $\vec u$ to $\vec{u'}$ which yields a uniformly random vector with first bit $0$. 
After multiplication by $\vec C$, $\vec C\vec u'$ is uniform on $\vec r^\perp$, conditional on $\vec A$ and $\vec C$. Since $\vec e_k\odot\vec r=r_{n+k}=1$, adding an independent uniform bit in direction $\vec e_k$ makes the result uniform over $\Z_2^{2n}$. This remains true after the remaining masks and permutation, conditional on the transformed public matrix.
Therefore, conditioned on the good event in the structured case, this transformation produces the appropriate structured or unstructured target distribution up to negligible total variation distance.
We submit our transformed instance to $\CO$ and output its Decision.
If no column of $\vec A$ has first entry $1$, or the rotation has no pivot in $M$, we instead query $\CO$ on a fresh unstructured target instance and return its output. These events have total probability at most $1/(2^n+1)+2^{-m}=\negl(n)$.
To complete the proof, we must address the cases in which the good event does not occur.
That is, we discuss the cases when $(e_1, e_{n+1}) \neq (0, 0)$.
The other cases for the values of $(e_1, e_{n+1})$ are $(0, 1), (1, 0), (1, 1)$.
Let $p_{ij} = \Pr[(e_1, e_{n+1}) = (i, j)]$.
Then $p_{00} = 1 - p$ and $p_{01} = p_{10} = p_{11} = p/3$.
Define $q_{ij}$ to be the probability that $\CO$ outputs \texttt{structured} when we receive a structured input wherein $(e_1, e_{n+1}) = (i, j)$; let $\bar{q}$ be the probability that $\CO$ outputs \texttt{structured} on an unstructured input.
By the above, $|q_{00} - \bar{q}| \geq \frac{1}{\poly(n)}$.
In general, \begin{align}
    Q_0 & := \Pr_{p_{ij} = (1-p, p/3, p/3, p/3)}[\CO = \text{\texttt{structured}} \,|\, \text{\texttt{structured}}] \\
    & = q_{00} (1-p) + \frac{p}{3} (q_{01} + q_{10} + q_{11}) .
\end{align}
Now, consider an alternate reduction wherein we first add uniform noise to bits $1$ and $n+1$ of $\vec u$.
This has no effect on unstructured instances, but on structured instances, we may equivalently express the effect of the extra noise by modifying $p_{ij} = 1/4$ for all $i, j$.
In this case, \begin{align}
    Q_1 & := \Pr_{p_{ij} = (1/4, 1/4, 1/4, 1/4)}[\CO = \text{\texttt{structured}} \,|\, \text{\texttt{structured}}] \\
    & = q_{00} \frac{1}{4} + \frac{1}{4} (q_{01} + q_{10} + q_{11}) .
\end{align}
Suppose that both probabilities were negligibly close to unstructured probability, i.e. $|Q_0 - \bar{q}| = \negl(n)$ and $|Q_1 - \bar{q}| = \negl(n)$.
Define \begin{align}
    a := \frac{1}{1 - \frac{4}{3}p} ,\; b := - \frac{4}{3} a p .
\end{align}
Then $a + b = 1$ and $a Q_0 + b Q_1 = q_{00}$.
By the triangle inequality, \begin{align}
    \frac{1}{\poly(n)} & \leq |q_{00} - \bar{q}| = |a Q_0 + b Q_1 - \bar{q}| \leq a |Q_0 - \bar{q}| + |b| \cdot |Q_1 - \bar{q}| \\
    & = a \cdot \negl(n) + |b| \cdot \negl(n) = \negl(n) ,
\end{align}
where in the last line we used the fact that $\frac{3}{4} - p = \W(1)$ so that $a = O(1)$ and thus $|b| = O(1)$.
This is a contradiction, and therefore either $|Q_0 - \bar{q}| \geq 1/\poly(n)$ or $|Q_1 - \bar{q}| \geq 1/\poly(n)$.
In other words, at least one of the two reductions solves Decision $\symplpn(n, p)$ with non-negligible advantage.
\end{proof}

We may now proceed to construct cryptographic protocols from $\symplpn$, starting with a simple one-way function, and then constructing public-key cryptography and oblivious transfer schemes that rely on the above result.

\section{Public-key encryption from low-noise $\symplpn$} \label{sec:public-key}

While private-key cryptography is readily constructible from $\lpn$ at constant noise rate, public-key encryption is substantially more challenging to construct from $\lpn$ at constant noise rate.
Alekhnovich-style $\lpn$-based $\mathsf{PKE}$ schemes with known attacks running in time $2^{O(\sqrt{n})}$ are correct only when $p = O(1/\sqrt{n})$~\cite{Alekhnovich03,damgaard2012practical}.
% At such scales, a brute-force enumeration attack would take time $2^{\widetilde{\Theta}(\sqrt{n})}$.
% While constructions of $\mathsf{PKE}$ have been proposed since the work of \cite{Alekhnovich03}, their attacks execute in time $2^{\widetilde{O}(\sqrt{n})}$ and therefore have comparable security to low-noise $\lpn$ schemes.
We here construct a $\mathsf{PKE}$ scheme based on $\symplpn$ with the same $\Theta(1/\sqrt{n})$ noise rate.
% Indeed, our construction is inspired by such low-noise $\lpn$ $\mathsf{PKE}$ constructions.

\begin{construction}[$\mathsf{PKE}$ from low-noise $\symplpn$]\label{const-pke}
Let $n \in \mathbb{N}$ be the security parameter and $p \in (0,1)$. The $\mathsf{PKE}$ scheme $\Sigma = (\gen,\mathsf{Enc}, \mathsf{Dec})$ is given by:
\begin{itemize}
    \item $\gen(1^n):$ sample a random full-rank isotropic matrix $\vec A \in \Z_2^{2n \times n}$, $\vec x \sim \Z_2^n$ and $\vec e \sim \CD_p^{\otimes n}$; output public key $\pk = (\vec A,\vec b = \vec A \vec x + \vec e)$ and secret key $\mathsf{sk}=\vec x$.

    \item $\mathsf{Enc}(
    \pk,\mu):$ to encrypt a single bit $\mu \in \bit$ using the public key $(\vec A,\vec b)$,  sample $\vec f \sim \CD_p^{\otimes n}$ and output the ciphertext pair $\mathsf{ct}=(\vec f \odot \vec A, \vec f \odot \vec b + \mu)$.

    \item $\mathsf{Dec}(\sk,\mathsf{ct})$:
    to decrypt $\mathsf{ct}=(\vec u,c)$ using the secret key $\sk = \vec x$, output $c + \vec u \cdot \vec x$. 
\end{itemize}
\end{construction}

We next state the correctness and security of \Cref{const-pke}.
Our security proof demonstrates the CPA security of the scheme.

\begin{theorem}[Correctness]\label{thm:pke_correctness} For any $\delta >0$, there exists $p = \Theta\left(1/\sqrt{n} \right)$ such that
the $\mathsf{PKE}$ scheme $\Sigma$ in \Cref{const-pke} is $(1-\delta)$-correct, i.e., for any message bit $\mu \in \bit$,  
\begin{align}
    \Pr\left[ \mathsf{Dec}(\sk,\mathsf{ct}) = \mu \,\, \vline \,\, \substack{
(\pk,\sk) \leftarrow \gen(1^n)\\
\mathsf{ct} \leftarrow \mathsf{Enc}(\pk,\mu)} \right] \geq 1- \delta.
\end{align}
\end{theorem}
\noindent\textbf{Remark (Correctness amplification).} For the PKE and OT applications, we amplify the scheme at $p=\Theta(1/\sqrt{n})$ to have negligible decryption error by encrypting each bit $2\lceil\log^2 n\rceil+1$ times under the same public key with fresh independent encryption randomness and taking the majority of the decrypted bits. This leaves key generation unchanged and preserves IND-CPA security.\par
The proof is deferred to \Cref{app:additional_proofs}.
Our security proof relies on the following indistinguishability lemma, which one can think of as the \emph{dual mode} of $\symplpn$.
\begin{lemma}\label{lemma:sympLPN_reducedbit_security}
Let $n \in \mathbb{N}$ be the security parameter.
Suppose that there exists an efficient (quantum or classical) algorithm which runs in time $\poly(n)$ and, with non-negligible advantage, distinguishes
\begin{align}
    (\vec H, \vec f \odot \vec H) \quad \text{from} \quad (\vec H, \vec r) ,
\end{align}
where $\vec H =(\vec B \,|\,  \vec v) \in \Z_2^{2n \times (n+1)}$ consists of a random full-rank isotropic matrix $\vec B \in \Z_2^{2n \times n}$ and a random column vector $\vec v \sim \Z_2^{2n}$, and where 
$\vec f \sim \CD_p^{\otimes n}$
and $\vec r \sim \Z_2^{n+1}$. 
Then there exists an efficient (quantum or classical) algorithm which solves Decision $\symplpn[n-1, n, p]$.
\end{lemma}

\begin{proof}
% Assume there exists a distinguisher $\mathcal{D}$ that can distinguish between the two samples with advantage $1/\poly(n)$. We will show that such a distinguisher allows us to solve Decision $\symplpn$.
Suppose we are given as input the $\symplpn(n-1, n, p)$ instance $(\vec A \in \Z_2^{2n \times (n-1)},\vec b \in \Z_2^{2n})$, where $\vec b$ is either structured with $\vec b=\vec A \vec x + \vec f$, for $\vec f \sim \CD_p^{\otimes n}$, or unstructured with $\vec b \sim \Z_2^{2n}$. 
Note that the matrix $\vec A$ has full column rank $n-1$, and thus $S = \mathrm{im}(\vec A)$ spans a $(n-1)$-dimensional subspace within $\Z_2^{2n}$. Consequently, the symplectic dual $S^\perp$---the space of vectors whose symplectic inner product with vectors in $S$ is zero---is a subspace of dimension $n+1$.
Our reduction proceeds as follows:
\begin{enumerate}
    \item Sample $\vec u \sim S^\perp \backslash S$. 
    Let $\vec{B} \in \Z_2^{2n \times n}$ be a random basis for $\operatorname{im}(\vec{A}) \oplus \text{span}(\vec u)$.
    Note that $\vec B$ is isotropic with $\dim \operatorname{im}(\vec B) = n$.
    Since $n$ is the maximum dimension of an isotropic subspace, $\operatorname{im}(\vec B) = \operatorname{im}(\vec B)^\perp$.
    Then, let $\vec v$ be a random vector in $S^\perp \setminus \operatorname{im}(\vec B)$. 
    Set $\vec H =(\vec B \, | \, \vec v)$.
    We argue that $\vec H$ has negligible total variation distance from the random variable mentioned in \Cref{lemma:sympLPN_reducedbit_security}. 

    \item Run the assumed distinguisher on input $(\vec H, \vec b \odot \vec H)$.
\end{enumerate}
Since each column of $\vec H$ is in $S^\perp$, $\vec H^\intercal \odot \vec A = 0$. 
Consequently, 
\begin{itemize}
    \item if $\vec b$ is structured with $\vec b=\vec A \vec x + \vec f$, our reduction computes
    \begin{align}
        \vec b \odot \vec H = (\vec H^\intercal \odot \vec b)^\intercal = (\vec H^\intercal \odot (\vec A \vec x + \vec f))^\intercal = (\vec H^\intercal \odot \vec f)^\intercal = \vec f \odot \vec H;
    \end{align}
    
    \item if $\vec b$ is unstructured with $\vec b \sim \Z_2^{2n}$, then our reduction computes $\vec r=\vec b \odot \vec H$.
\end{itemize}
In both cases, assuming a negligible total variation distance between $\vec H$ and the random variable mentioned above, we claim that the input to the distinguisher is negligibly close to the desired distribution.
In the structured case, this is immediate, whereas in the unstructured case it readily follows from the fact that $\vec r = \vec b \odot \vec H $ is uniformly random over $\Z_2^{n+1}$ when $\vec b \sim \Z_2^{2n}$ itself is uniform.
This is because $\vec H$ is a full-rank matrix, so that every bitstring in $\Z_2^{n+1}$ has the same number of pre-images under $\vec H^\intercal$.

We need only justify the assertion in Step 1 regarding the distribution of $\vec H$. 
First, we claim that $\vec B$ is a uniformly random full-rank isotropic matrix.
Indeed, $\vec A \in \Z_2^{2n \times (n-1)}$ is a uniformly random full-rank isotropic matrix, and there are exactly $2^{n+1} - 2^{n-1}$ choices for $\vec u$ no matter what $\vec A$ is. 
Hence, the probability $p_\vec{B}$ of sampling any given $\vec{B}$ that is full rank and isotropic is $\frac{1}{(2^{n+1} - 2^{n-1})N_\vec{A}}$, where $N_\vec{A}$ is the number of full-rank isotropic matrices $\vec{A} \in \Z_2^{2n \times (n-1)}$. 
This probability is equal for any such $\vec{B}$, so $\vec B$ is indeed uniformly distributed. 

Next, we claim that $\vec{v}$ is negligibly close in distribution to a uniformly random vector outside of $\operatorname{im}(\vec{B})$.
By definition, $\vec v$ is a random vector, not contained in $\operatorname{im}(\vec B)$, that is symplectically orthogonal to a random $(n-1)$-dimensional subspace $S$ of $\operatorname{im}(\vec{B})$, namely $\operatorname{im}(\vec{A})$.
For any nonzero $\vec v \notin \operatorname{im}(\vec B)$, $\vec v$ has an orthogonal complement of dimension $2n - 1$, and its intersection with $\operatorname{im}(\vec B)$ therefore has dimension $n-1$, one less than the maximum possible $n$.
This is because not every element of $\operatorname{im}(\vec B)$ is orthogonal to $\vec v$, as $\vec v \notin \operatorname{im}(\vec B) = \operatorname{im}(\vec B)^\perp$.
For each $\vec{v} \notin \operatorname{im}(\vec B)$, therefore, there is exactly one $(n-1)$-dimensional subspace of $\operatorname{im}(\vec B)$ which is orthogonal to $\vec v$. 
Say that the number of $(n-1)$-dimensional subspaces of $\operatorname{im}(\vec B)$ is $N_\vec{B}$. 
Then the probability of sampling any $\vec{v} \notin \operatorname{im}(\vec B)$ is $\frac{1}{(2^{n+1} - 2^{n})N_\vec{B}}$, a constant not depending on $\vec v$ as desired.
($2^{n+1} - 2^n$, above, is the number of vectors $\vec v \notin \operatorname{im}(\vec B)$ orthogonal to any given $(n-1)$-dimensional $V \subseteq \vec B$.)

Recall that the total variation distance between uniform distributions over sets $\mathcal{X}, \mathcal{Y}$ where $\mathcal{X} \subseteq \mathcal{Y}$, is given by $1 - |\mathcal{X}| / |\mathcal{Y}|$.
Thus, for each $\vec B$, $\vec v$ has total variation distance $2^{-n}$ from uniform over all of $\Z_2^{2n}$, since it is uniform over $\Z_2^{2n}\backslash\operatorname{im}(\vec B)$, which is a $1-2^{-n}$ fraction of the full set. 
Therefore, the distribution of $\vec H = (\vec {B} \,|\, \vec{v})$ has negligible total variation distance from one wherein $\vec B$ is a uniformly random full-rank isotropic matrix and $\vec v \sim \Z_2^{2n}$.
\end{proof}

\begin{theorem}[Security]\label{thm:pke_security}
Let $n \in \mathbb{N}$ be the security parameter and $\Sigma$ be the public-key encryption scheme in~Construction \ref{const-pke} with parameter $p$. 
Suppose there is a (quantum or classical) algorithm $\CA$, running in time $\poly(n)$, which breaks the \textsf{IND-CPA} security of $\Sigma$.
Then there is a (quantum or classical) algorithm $\CB$ which solves either Decision $\symplpn[n, p]$ or Decision $\symplpn[n-1, n, p]$.
\end{theorem}

\begin{proof}
Let $n\in\mathbb N$. For each fixed $\mu\in\bit$, consider the following hybrid distributions of $(\pk,\mathsf{ct})$:
\begin{itemize}
    \item \textsf{H}${_0}:$ The public key $\pk$ and ciphertext $\mathsf{ct}$ are generated by $\Sigma$ in \Cref{const-pke} with \begin{align}
        \mathsf{ct}=(\vec f \odot \vec A, \vec f \odot \vec b + \mu) 
    \end{align}
    where $\pk = (\vec A,\vec b = \vec A \vec x + \vec e)$ and $\mathsf{sk}=\vec x$ are output by $\gen(1^n)$ such that $\vec A \in \Z_2^{2n \times n}$ is a random full-rank isotropic matrix, $\vec x \sim \Z_2^n$, and $\vec e,\vec f \sim \CD_p^{\otimes n}$.

    \item \textsf{H}${_1}:$ Same as $\mathsf{H}_0$, except that $\vec b$ in the public key $\pk = (\vec A,\vec b)$ is instead sampled uniformly at random over $\Z_2^{2n}$.

    \item \textsf{H}${_2}:$ Same as $\mathsf{H}_1$, except that $\mathsf{ct}$ is chosen uniformly at random over $\Z_2^{n+1}$.
\end{itemize}

Recall that the \textsf{IND-CPA} security game for a one-bit \textsf{PKE} scheme is defined as follows.
The challenger samples $(\pk, \sk) \leftarrow \gen(1^n)$ and $\m \leftarrow \set{0, 1}$.
Then the challenger sends $\pk, \enc(\pk, \m)$ to the adversary, who responds with a guess $\hat{\m}$.
The adversary wins if $\hat{\m} = \m$, and we define the adversary's advantage as $|\Pr[\hat{\m} = \m] - 1/2|$.
Importantly, in \textsf{H}${_2}$ the encryption $\mathsf{ct}$ is completely independent of $\mu$, and hence the distinguishing advantage in the $\mathsf{IND}\mbox{-}\mathsf{CPA}$ experiment is zero. 
Thus an advantageous adversary in this game implies a distinguisher with non-negligible advantage for at least one of these hybrids.
Consequently, it suffices to show that the three hybrids are computationally indistinguishable.

First we show that an efficient distinguisher between $\mathsf{H}_0$ and $\mathsf{H}_1$ implies an efficient solver for Decision $\symplpn[n, p]$.
Suppose we are given as input $(\vec A,\vec b)$, where $\vec b$ is either structured (i.e. $\vec b=\vec A \vec x + \vec e$) or unstructured (i.e. $\vec b \sim \Z_2^{2n}$).
Then we can generate the ciphertexts for each $\mu$ ourselves, simulating the two hybrids.
Thus, an efficient algorithm distinguishing \textsf{H}${_0}$ and \textsf{H}${_1}$ with non-negligible advantage implies an efficient solver for Decision $\symplpn[n, p]$. 

Next, we argue that the hybrids \textsf{H}${_1}$ and \textsf{H}${_2}$ are indistinguishable. 
To prove this, we will view the public key $(\vec A,\vec b)$ as a matrix 
$\vec H \in \Z_2^{2n \times (n+1)}$. 
This matrix $\vec H$ satisfies the property that the first $n$ columns form a uniformly random full-rank isotropic matrix, while the final column is uniformly random.
In $\mathsf{H}_1$, $\mathsf{ct} = \vec f \odot \vec H + (0^n, \m)$, whereas in $\mathsf{H}_2$, $\mathsf{ct} = \vec r \sim \Z_2^{n+1}$.
For each fixed $\m\in\bit$, any distinguisher between $\mathsf{H}_1$, which gives $(\vec H, \vec f \odot \vec H + (0^n, \m))$, and $\mathsf{H}_2$, which gives $(\vec H, \vec r)$, implies a distinguisher between $(\vec H, \vec f \odot \vec H)$ and $(\vec H, \vec r)$ by way of adding $(0^n, \m)$ ourselves (in the case of $(\vec H, \vec r)$, adding $(0^n, \m)$ has no effect on the distribution of $\vec r$).
Moreover, a distinguisher for the latter problem implies a solver for Decision $\symplpn[n-1, n, p]$ by \Cref{lemma:sympLPN_reducedbit_security}.
We conclude therefore that $\mathsf{H}_1$ and $\mathsf{H}_2$ are computationally indistinguishable under the hardness of $\symplpn[n-1, n, p]$.
\end{proof}

Previously, in \Cref{sec:reducing_one_logical_sympLPN}, we showed in \Cref{thm:sympLPN_reducing_logicals} that Decision $\symplpn[n, p_0]$ reduces to $\symplpn[n-1, n, p]$ where $p = p_0 + \log^2 n/n$.
Hence, taking $p = p_0 + \log^2 n/n$ in \Cref{const-pke}, its security is based entirely on the hardness of Decision $\symplpn[n, p_0]$, for $p_0 = \Theta(1/\sqrt{n})$: Decision $\symplpn[n, p_0]$ also reduces to Decision $\symplpn[n, p]$ by adding independent depolarizing noise, using \Cref{lemma:depolarizing_convolution}. Since $\log^2 n/n=o(1/\sqrt{n})$, the scheme still uses $p=\Theta(1/\sqrt{n})$, as desired.

% We remark that $\gen$ runs in time $O(n^3)$, $\enc$ in time $O(n^2)$, and $\dec$ in time $O(n)$.

\section{Strongly uniform public-key encryption}
\label{sec:SU-PKE}

In this section, we alter our previous $\mathsf{PKE}$ construction to upgrade it into a so-called Strongly Uniform Type-A $\mathsf{PKE}$ (\textsf{SU-PKE}).
A post-quantum \textsf{SU-PKE} scheme is known to imply oblivious transfer (\textsf{OT}) and secure multiparty computation with constant-round $\varepsilon$-simulation or polynomial-round full simulation~\cite{chatterjee2025round}.

\begin{definition}[Strongly Uniform Type-A Public-Key Encryption] \label{def:SU-PKE}
A $\mathsf{PKE}$ scheme $(\gen, \enc, \dec)$ is called "Strongly Uniform of Type A" if the public key $\pk$ generated by $\gen(1^n)$ is computationally indistinguishable from a uniform distribution over some efficiently sampleable group parameterized by $n$, $G(n)$.
\end{definition}
For our purposes, the public key takes values in $\Z_2^m$ for some $m$ depending on $n$, and this is the corresponding group.
However, the PKE scheme presented in \Cref{sec:public-key} is not strongly uniform. 
The public key is of the form $(\vec A, \vec{Ax} + \vec{e})$---while $\vec{Ax} + \vec{e}$ is computationally indistinguishable from uniform, $\vec A$ is a uniformly random full-rank isotropic, and is thus far from uniform. 
Fortunately, it is possible to change the cryptographic protocol to remedy this problem. 
To do so, we first prove a lemma that relates the random variable $\vec A \in \Z_2^{2n \times n}$ to a uniform distribution over $\Z_2^{4n^2}$. 

\begin{lemma}[A uniformly random public key] \label{lem:uniform_public_key}
Sample $\vec s \sim \Z_2^{4n^2}$ and let $\vec A \in \Z_2^{2n \times n}$ be a uniformly random full-rank isotropic matrix. 
There exists a deterministic, polynomial-time classical algorithm $\mathcal A$ and a polynomial-time, randomized classical algorithm $\mathcal B$ such that $\mathcal{A}(\vec s)$ is statistically indistinguishable from $\vec A$, $\mathcal{B}(\vec A)$ is statistically indistinguishable from $\vec s$, and 
\begin{align}
    \Pr_{\CB, \vec A}[\mathcal{A}(\mathcal{B}(\vec A)) = \vec A] = 1-\negl(n).
\end{align}
\end{lemma}

\begin{proof}
Intuitively, $\CA$ uses $\vec s$ as a source of randomness to deterministically construct $\vec A$ such that if $\vec s$ is uniformly random, then $\vec A$ is a uniformly random full-rank isotropic matrix.
We construct $\mathcal{A}$ as follows.
Given $\vec s \sim \Z_2^{4n^2}$, we build $\vec{A}$ column-by-column. 
Let $\widetilde{\vec{A}}_i \in \Z_2^{2n \times i}$ be a matrix with $i$ columns, and let $\widetilde{\vec{A}}_0$ be an empty matrix. 
For $i = 1, \dots, 2n$, we use a deterministic algorithm to calculate an ordered basis for the symplectic dual space of $\operatorname{im}(\widetilde{\vec A}_{i-1})$. 
(By convention, the dual of the empty matrix is $\Z_2^{2n}$.)
This orthogonal complement has some dimension $d_i \leq 2n$, and we may consume the next $d_i$ unused bits of $\vec{s}$ to sample a random vector $\vec w \in \operatorname{im}(\widetilde{\vec A}_{i-1})^\perp$. 
Then, we can set $\widetilde{\vec A}_i = (\widetilde{\vec A}_{i-1}\,|\, \vec{w})$. 
This process yields a matrix $\widetilde{\vec A}_{2n}$, and there are always enough bits of $\vec s$ since there are $2n$ steps taking $d_i \leq 2n$ bits each. 
$\widetilde{\vec A}_{2n}$ is almost what we want, except that it has twice as many columns as needed and is not necessarily full rank. 

Therefore, we define $\vec A$ by iteratively adding the $i$th column of $\widetilde{\vec A}_{2n}$ if it is linearly independent to all previous columns.
We then pad $\vec A$ with zero columns until all $n$ columns are filled in (this is an edge-case event which we will momentarily show occurs with only negligible probability). 
Note that the number of columns is never larger than $n$, since there can be at most $n$ linearly independent vectors that are symplectically orthogonal. 

Now, we claim that $\mathcal{A}(\vec s)$ is statistically indistinguishable from a uniformly random full-rank isotropic matrix $\vec A \in \Z_2^{2n \times n}$. 
Conditioned on the event $L$---that $\vec A$ was not padded with any zero columns---we are constructing a matrix column by column by sampling symplectically orthogonal, linearly independent vectors, which yields precisely the same distribution as $\vec A$. 
Hence, we need only show $\Pr[L^c] = \negl(n)$. 
For the $i$th column of $\widetilde{\vec A}_{2n}$, the chance that it does not increase the rank from the previous $i-1$ columns is precisely $\frac{2^{2n - d_i}}{2^{d_i}}$. 
For $d_i \geq n + 1$ this probability is at most $1/4$.
At $i = 1$, $d_i$ begins at $2n$, and decreases by 1 each time a linearly independent vector is added.
If $\vec{A}$ does not have $n$ columns prior to padding, then $\widetilde{\vec{A}}_{2n}$ does not have rank $n$, so this event must occur at least $n + 1$ times.
Each time this event occurs, $d_i$ does not decrease, which means $d_i$ decreased at most $n-1$ times.
Before rank $n$ is reached, each step fails to increase the rank with conditional probability at most $1/4$. The number of such failures is therefore stochastically dominated by $X\sim\Bin(2n,1/4)$, and a Chernoff bound gives
$\Pr[L^c]\leq\Pr[X\geq n+1]\leq\exp(-n/6)=\negl(n)$.
We next construct algorithm $\mathcal{B}$. 
Given an input $\vec{A} \in \Z_2^{2n \times n}$ that is full-rank and isotropic, we will describe a procedure for $i = 1, \dots, 2n$ to build matrices $\widetilde{\vec{A}}_{i}$, with $\widetilde{\vec{A}}_0$ empty.
\begin{enumerate}
\item Let $d_i := \dim \operatorname{im}(\widetilde{\vec{A}}_{i-1})^{\perp}$. 
With probability $1 - \frac{2^{2n - d_i}}{2^{d_i}}$, let $\widetilde{\vec{A}}_{i} = (\widetilde{\vec{A}}_{i-1} \, | \, \vec{v}_i)$, where $\vec{v}_i$ is the next column of $\vec A$ which has not yet been sampled.

\item Otherwise, let $\widetilde{\vec{A}}_{i} = (\widetilde{\vec{A}}_{i-1} \, | \, \vec{u}_i)$, where $\vec{u}_i$ is a random vector in the span of the columns in $\vec A$ which have already been sampled.
(If no column has yet been sampled, then $\vec u_i = \vec 0$.)
\end{enumerate}
Notice that the columns are all symplectically orthogonal. 
After all $2n$ iterations, given $\widetilde{\vec{A}}_{2n}$, reconstruct the random bits $\vec s'$ of length $m \leq 4n^2$ that would have produced $\widetilde{\vec{A}}_{2n}$. 
Sample the remaining bits uniformly at random to produce $\vec{s}$ of length $4n^2$. 
Output $\vec{s}$. 

By the same reasoning as before, $\widetilde{\vec{A}}_{2n}$ has rank less than $n$ with only $\negl(n)$ probability. 
Moreover, conditioned on the event $L$ that $\widetilde{\vec{A}}_{2n}$ has rank $n$, we have that $\mathcal{A}(\mathcal{B}(\vec A)) = \vec A$. 
Indeed, applying $\mathcal{A}$ to the seed $\vec{s} = \mathcal{B}(\vec A)$ produces a matrix $\widetilde{\vec{A}}_{2n}$, and the first $n$ linearly independent columns obtained from it give $\vec A$. 
Hence, $
    \Pr[\mathcal{A}(\mathcal{B}(\vec A)) = \vec A] = 1 - \negl(n).$ 
It now suffices to show that $\mathcal{B}(\vec A)$ is statistically indistinguishable from a uniformly random $\vec s$. 
The intermediate matrices need not be uniformly distributed, because their probabilities depend on when linearly dependent columns occur.
However, when $\vec A$ is uniformly random, each next unused column, conditioned on the columns already consumed, is uniform among the admissible extensions, since every such extension has the same number of completions to a full-rank isotropic matrix.
Both samplers use the same rank-transition probabilities and uniform samples from the current span for dependent columns, so they generate the same distribution of column sequences.
Reconstructing the unique seed prefix for each sequence using the same deterministic ordered bases, and appending fresh uniform unused bits, therefore reproduces the uniform seed distribution of the forward sampler.
By the procedure of $\mathcal{B}$, every full-rank $\vec{A}$ has exactly the same number of seeds $\vec{s}$, since the number of seeds only depends on counting which columns increase the rank of $\vec{\widetilde{A}}_{2n}$, and not the actual values of the first $n$ linearly independent columns.
$\mathcal B(\vec A)$ is therefore uniform over all $4n^2$-bit seeds, including those for which $\mathcal A$ outputs a rank-deficient matrix.
Now, noting that the event $L^c$ that $\mathcal{A}(\vec{s})$ is not full-rank has $\negl(n)$ probability, we conclude that there is a negligible fraction of seeds which are not in the preimage of full-rank $\vec{A}$, so that the uniform distribution over seeds $\vec{s'}$ has negligible total variation distance from $\vec s \sim \Z_2^{4n^2}$. 
\end{proof}

\begin{corollary}\label{cor:indistinguishability}
Let $\vec{s} \sim \Z_2^{4n^2}$, $\vec A \in \Z_2^{2n \times n}$, and algorithm $\CA$ be as in \Cref{lem:uniform_public_key}.
Fix any randomized algorithms $\mathcal{A}_1, \mathcal{A}_2$ mapping to $\{0, 1\}^m$ for $m= \poly(n)$.
If $(\vec{A}, \mathcal{A}_1(\vec{A}))$ is computationally indistinguishable from $(\vec{A}, \mathcal{A}_2(\vec{A}))$, then $(\vec{s}, \mathcal{A}_1(\mathcal{A}(\vec{s})))$ and $(\vec{s}, \mathcal{A}_2(\mathcal{A}(\vec{s})))$ are computationally indistinguishable. 
\end{corollary}
\begin{proof}
Suppose for contradiction that there is a distinguisher $\CD$ for the latter two distributions. 
Then we can evaluate $\mathcal{D}$ on the sample $(\mathcal{B}(\vec{A)}, \mathcal{R}(\vec A))$, where $\mathcal{R}$ is either $\mathcal{A}_1$ or $\mathcal{A}_2$. 
By writing $\vec{s}' = \mathcal{B}(\vec A)$, this is precisely $\mathcal{D}(\vec s', \mathcal{R}(\mathcal{A}(\vec s')))$ with probability $1 - \negl(n)$, since $\mathcal{A}(\mathcal{B}(\vec A)) = \vec A$ with that probability. 
However, $\vec{s}'$ has negligible total variation distance from $\vec{s}$, and hence $\mathcal{D}$ can distinguish between  $\mathcal{R}$ being $\mathcal{A}_1$ versus $\mathcal{A}_2$. 
This is a contradiction. 
\end{proof}

\begin{corollary}\label{cor:su_pke}
Assuming $\symplpn[n, p_0]$ with $p_0 = \Theta(\frac{1}{\sqrt{n}})$, there exists a \textsf{IND-CPA}-secure \textsf{SU-PKE} scheme.
\end{corollary}
\begin{proof}
Take $p=p_0+\log^2 n/n$. Using \Cref{lem:uniform_public_key}, the construction, correctness, and security of the scheme in \Cref{sec:public-key} can be readily repeated, which we now outline. 
In the construction keep the same secret key, but take the public key $\pk$ to be $(\vec s, \vec b)$ for $\vec b = \mathcal{A}(\vec s) \cdot \vec x + \vec e$, where $\vec s \sim \Z_2^{4n^2}$ is uniformly random.
Assuming the hardness of $\symplpn$, this $\pk$ is manifestly computationally indistinguishable from uniformly random, satisfying the requirements of \textsf{SU-PKE} given in \Cref{def:SU-PKE}. 
Then, the encryption algorithm $\enc(\pk, \mu)$ outputs $\mathsf{ct} = (\vec f \odot \mathcal{A}(\vec s), \vec f \odot \vec b + \mu)$, while the decryption algorithm $\dec(\sk, \mathsf{ct})$ for $\mathsf{ct} = (\vec u, c)$ outputs $c + \vec u \cdot \vec x$. 
For the basic scheme, reconstructing $\vec A = \mathcal{A}(\vec s)$ incurs a polynomial-time preprocessing cost for each public key. Once $\vec A$ is cached, encryption takes $O(n^2)$ time per bit and decryption takes $O(n)$ time per bit. Amplification multiplies these per-bit costs and the ciphertext length by $O(\log^2 n)$.
Correctness with negligible error follows from \Cref{thm:pke_correctness} and the subsequent amplification remark, since the scheme is exactly the same as the one in \Cref{sec:public-key}, up to negligible total variation distance. 
The only remaining claim is \textsf{IND-CPA} security. 
In the security proof in \Cref{thm:pke_security}, one can replicate the exact same hybrid argument, replacing $\vec{A}$ in the public key with $\vec s$. 
Our arguments that $\mathsf{H}_0$ and $\mathsf{H}_1$ are indistinguishable, and that $\mathsf{H}_1$ and $\mathsf{H}_2$ are indistinguishable, both follow in the \textsf{SU-PKE} construction from an application of \Cref{cor:indistinguishability}.
Indeed, in both cases the new hybrid argument is the same as the original, except that $\vec{A}$ is replaced in the public key with $\vec{s}$.
Meanwhile, the rest of the public key and ciphertext is one of two different randomized functions of $\vec{A}$, as formulated in \Cref{cor:indistinguishability}.
\end{proof}

Our post-quantum \textsf{SU-PKE} scheme is a dense cryptosystem: its public keys are computationally indistinguishable from uniformly random bitstrings.
Consequently, the results of Chatterjee, Liang, Pandey, and Yamakawa~\cite[Theorems~12, 13, and~20]{chatterjee2025round} yield both constant-round $\varepsilon$-simulatable \textsf{OT} and polynomial-round \textsf{OT} with full simulation against malicious quantum adversaries.
Both guarantees follow from the same low-noise $\symplpn$ hardness assumption.
Here, $\varepsilon$-simulation allows distinguishing advantage at most $\varepsilon(n)+\negl(n)$ between the real and ideal executions, for every noticeable $\varepsilon(n)$, with a quantum simulator running in time $\poly(n,1/\varepsilon)$~\cite[Definition~7]{chatterjee2025round}. Full simulation requires negligible distinguishing advantage with a quantum simulator running in time $\poly(n)$.
\begin{corollary}\label{cor:malicious_ot}
Assuming the computational hardness of $\symplpn[n,p_0]$ with $p_0=\Theta(1/\sqrt n)$, there exist (i) a constant-round $\varepsilon$-simulatable oblivious transfer protocol for every noticeable $\varepsilon=\varepsilon(n)$, and (ii) a polynomial-round oblivious transfer protocol with full simulation. Both protocols are secure against malicious quantum adversaries.
\end{corollary}
The corresponding MPC guarantees follow from \cite[Theorems~3 and~21]{chatterjee2025round}.
\begin{corollary}\label{cor:secure_mpc}
Assuming the computational hardness of $\symplpn[n,p_0]$ with $p_0=\Theta(1/\sqrt n)$, there exist (i) a constant-round $\varepsilon$-simulatable secure multiparty computation protocol for every noticeable $\varepsilon=\varepsilon(n)$, and (ii) a polynomial-round secure multiparty computation protocol with full simulation. Both protocols are secure against malicious quantum adversaries.
\end{corollary}

These results demonstrate that low-noise $\symplpn$ can underlie the same central cryptographic primitives that low-noise $\lpn$ can, including public-key encryption, oblivious transfer, and secure multiparty computation.
It is open, however as to the comparative practical security of these schemes.

\section*{Acknowledgments}

We thank Alexandru Gheorghiu, Gregory Kahanamoku-Meyer, Peter Shor and Vinod Vaikuntanathan for many insightful discussions, particularly on the state-of-the-art attacks on $\lpn$ and its many cryptographic applications. 
We would also like to credit Kabir Tomer with an independent discovery of how to construct one-way functions from the $\lsn$ assumption.

JZL is funded in part by a National Defense Science and Engineering Graduate (NDSEG) Fellowship.
YQ is supported by a collaboration between the US DOE and other Agencies. 
This material is based upon work supported by the U.S. Department of Energy, Office of Science, National Quantum Information Science Research Centers, Quantum Systems Accelerator. 

\bibliographystyle{alpha}
\bibliography{main}

\appendix

\newpage

\crefalias{section}{appendix}

% hyperref anchor uniqueness
\renewcommand{\theHsection}{app.\Alph{section}}

\section{Further discussion of $\lpn$}
\label{app:more_on_lpn}

For completeness, we here give a formal definition of $\lpn$ and give a simple reduction from $\lpn[k, n, p]$ to $\lpn[k - k', n, p]$, where $k' = O(\log n)$.

\begin{definition}[Decision Learning Parity with Noise, $\lpn$] \label{def:lpn}
For positive integers $n,k$ and $p \in (0,1)$, let $\vec A \sim \Z_2^{n \times k}$, $\vec x \sim \Z_2^k$, and $\vec e \sim \Ber_p^{\otimes n}$, i.e. each $e_i$ is i.i.d. $1$ with probability $p$ and $0$ with probability $1-p$.
Let $\vec u \sim \Z_2^n$. Here $\vec A$, $\vec x$, and $\vec u$ are uniformly random.
$\lpn(k, n, p)$ is the task of distinguishing between the two distributional samples \begin{align}
    (\vec A, \vec{Ax} + \vec e) \quad \text{or} \quad (\vec A, \vec u \sim \Z_2^n)
\end{align}
with advantage at least $1/\poly(n)$.
The former case is called \emph{structured} and the latter \emph{unstructured}.
\end{definition}

\begin{lemma}[Reducing a logical bit in $\lpn$] \label{lem:LPN_reduce_logical}
Let $p \in (0,1)$ and positive integers $k,k',n$ satisfy $k > k'$ and $k \leq \poly(n)$.
Let $\CO$ be an oracle which solves Decision $\lpn(k-k', n, p)$ with advantage $\delta$.
Then there is an algorithm $\CA$, using $O(\poly(n))$ time apart from the oracle call, which solves Decision $\lpn(k, n, p)$ with advantage $\delta / 2^{k'}$ using a single call to $\CO$.
\end{lemma}

\begin{proof}
Let the Decision $\lpn(k, n, p)$ instance be $(\vec A \sim \Z_2^{n \times k}, \vec u)$.
%the code given as part of the input to the $\lpn(k, n, p)$ instance be $\vec A \sim \Z_2^{n \times k}$, and $\vec u$ be the other part of the input whose distribution depends on whether the instance is structured.
Define $\vec A' \in \Z_2^{n \times (k-k')}$ by discarding the last $k'$ columns of $\vec A$, and define $\vec{\bar{A}} \in \Z_2^{n \times k'}$ to be the last $k'$ columns of $\vec A$.
The algorithm will run $\CO$ on $(\vec A', \vec u)$ and output its answer.
Suppose that the $\lpn$ instance is structured, so $\vec u = \vec{Ax} + \vec e$ for $\vec x \sim \Z_2^k$ and $\vec e \sim \Ber_p^{\otimes n}$.
Let $\vec x = (\vec x', \vec{\bar{x}})$ with dimensions $k-k'$ and $k'$.
With probability $1/2^{k'}$, $\vec{\bar{x}} = \vec 0$. 
Conditioned on this event, $(\vec A', \vec u)$ is precisely a structured instance of $\lpn(k-k', n, p)$.
Otherwise, $\vec{\bar{x}} \neq 0$ and $\vec u = \vec A' \vec x' + \vec{\bar{A} \bar{x}} + \vec e$.
But $\vec{\bar{A} \bar{x}}$ is a uniformly random bitstring since the columns of $\vec{\bar{A}}$ are uniformly random and $\vec{\bar{x}} \neq \vec 0$.
Hence, $\vec u$ is independent from $\vec{A}'$ and marginally $\vec u \sim \Z_2^{n}$, i.e. $(\vec A', \vec u)$ is an unstructured instance.
Let $p_0 := \Pr[\CA = \text{\texttt{structured}} \,|\, \text{\texttt{structured}}]$ and $p_1 := \Pr[\CA = \text{\texttt{structured}} \,|\, \text{\texttt{unstructured}}]$.
Define $q_0, q_1$ similarly for $\CO$ as an oracle solving $\lpn(k-k', n, p)$.
Then \begin{align}
    p_0 = \frac{1}{2^{k'}} q_0 + \left( 1 - \frac{1}{2^{k'}} \right) q_1 , \quad p_1 = q_1 .
\end{align}
By definition, the advantage of $\CA$ is $\frac{|p_0 - p_1|}{2}$, and by assumption, $\frac{|q_0 - q_1|}{2} = \delta$.
Consequently,
\begin{align}
    \frac{|p_0 - p_1|}{2} = \frac{1}{2^{k'}} \frac{|q_0 - q_1|}{2} = \frac{\delta}{2^{k'}} 
\end{align}
as claimed.
\end{proof}

This proof implies a reduction from $\lpn[k, n, p]$ to $\lpn[k-k', n, p]$ for any $k' = O(\log n)$ logical bits.

\section{One-way function from high-noise $\slsn$}
\label{app:owff_tight_equivalence}

We here construct a one-way function family ($\owff$) such that any algorithm inverting the $\owff$ implies an algorithm solving $\slsn$.

\begin{definition}[$\owff$] \label{def:owff}
    Let $\CF_n := \set{f_I : \CD_I \to \CR_I}_{I \in \CI_n}$ be a collection of functions $f_I$ from a domain set $\CD_I$ to a range set $\CR_I$, where $I$ is an index drawn from an index set $\CI_n$ which depends on the security parameter $n$.
    Then $\CF_n$ is a \emph{one-way function family} ($\owff$) if the following conditions hold. \begin{enumerate}
        \item[(1) ] There exists a classical algorithm $\gen(1^n)$, running in time $O(\poly(n))$, which produces an index $I \in \CI_n$.
        
        \item[(2) ] There exists a classical algorithm $\sample(I)$ such that for any $I \in \CI_n$, $\sample(I)$ samples a uniformly random element of $\CD_I$ in time $O(\poly(n))$.
        
        \item[(3) ] There exists a classical algorithm $\eval(I, x)$ such that for any $I \in \CI_n$ and any $x \in \CD_I$, $\eval(I, x)$ outputs $f_I(x) \in \CR_I$ in time $O(\poly(n))$.
        
        \item[(4) ] For any (quantum or classical) algorithm $\CA$ running in time $\poly(n)$, \begin{align}
            \Pr_{\CA, I, x}[f_I(\CA(I, f_I(x))) = f_I(x)] = \negl(n) 
        \end{align}
        where in the randomness, $I \sim \gen(1^n)$ and $x \sim \sample(I)$.
    \end{enumerate}
\end{definition}
The existence of $\owff$ is equivalent to the existence of a standard one-way function, which in turn implies pseudorandom generators, pseudorandom functions, and all other primitives of private-key cryptography~\cite{katz2007introduction}.
Our $\owff$ construction is indexed by $\lsn$ instances, but we will show that its security is bounded below by that of $\slsn$.

\begin{construction}[$\owff$ candidate from Search $\slsn$] \label{const:owff_from_lsn}
We instantiate a $\owff$ candidate based on Search $\slsn(k, n, p)$ as follows.
\begin{itemize}
    \item $\gen(1^n)$ samples the $\lsn$ matrices $I = (\vec A, \vec B)$.
    
    \item Define $\CD_I = \Z_2^{n} \times \Z_2^{k} \times \Z_2^{m(n)}$, where $m(n) = \poly(n)$ is a sufficiently large number of bits so as to enable a sample from $\CD_p^{\otimes n}$. 
    $\sample(I)$ ignores $I$ and outputs $(\vec r, \vec y, \vec s)$ where $\vec r \sim \Z_2^n$, $\vec y \sim \Z_2^k$, and $\vec s \sim \Z_2^{m(n)}$.

    \item $f_I(x) = f_{\vec A, \vec B}(\vec y, \vec r, \vec s) = (\vec A, \vec B, \vec{Ar} + \vec{By} + \vec e)$.
    Here $\vec e$ is defined as follows.
    Let $\CW_d$ be the set of length-$2n$ bitstrings representing $n$-qubit Paulis with weight at most $d$. 
    We use $\vec s$ as a random seed to draw a sample from $\CD_p^{\otimes n}$.
    If $|\vec e| > 2.01 np$ then $\sample$ instead outputs $\vec 0$ for the error so as to remain in $\CW_{2.01np}$; however, the Chernoff bound implies that this event occurs with probability $\negl(n)$ so long as $np = \omega(\log n)$.
    Since this function is given by matrix multiplication, $f_I$ can be efficiently computed as required. 
\end{itemize}
\end{construction}

We now show that inverting $f_I$ with non-negligible probability is at least as hard as solving Search $\lsn(k, n, p)$, for any reasonable choice of $k$ and $p$ as functions of $n$.

\begin{theorem}[$\owff$ secured by $\slsn$] \label{thm:owff_security_LSN}
Let $\CF$ be the $\owff$ candidate in \Cref{const:owff_from_lsn}, with parameters $k, p$ such that $k = \omega(\log n)$ and $p = \omega(\frac{\log n}{n})$.
Further, with $\delta := 4.03p$ and $R := k / n$, assume that $(\delta, R)$ satisfies the quantum Gilbert-Varshamov bound \begin{align}
    H_2(\delta) + \delta \log 3 < 1 - R ,
\end{align}
where $H_2(x) = -x \log x - (1-x)\log(1-x)$ is the binary entropy.
Suppose that there exists an algorithm $\CA$, running in time $T$, such that \begin{align}
    \Pr_{\CA, I, x}[f_I(\CA(I, f_I(x))) = f_I(x)] = \frac{1}{\poly(n)} .
\end{align}
Then there exists an algorithm $\CB$, running in time $T + \poly(n)$, which solves Search $\slsn$ with probability $\frac{1}{2^k} + \frac{1}{\poly(n)} = \frac{1}{\poly(n)}$.
\end{theorem}

\begin{proof}
Let $(\vec C, \vec{EC} \ket{0^{n-k}, \psi})$ be the $\slsn[k, n, p]$ instance given to $\CB$.
$\CB$ first measures the syndrome, and chooses a uniformly random Pauli $\vec P_0$ with the same syndrome.
Hence, $\vec P_0 \vec E$ commutes with all stabilizers.
Then, choose a uniformly random logical Pauli $\vec L$ and multiply onto $\vec P_0$.
Let $\vec A_0 \in \Z_2^{2n \times n}$ be the matrix whose columns are symplectic representations of the $n-k$ stabilizers and $k$ logical $\vec Z_i$ Paulis, and let $\vec B_0 \in \Z_2^{k}$ likewise represent the logical $\vec X_i$ Paulis.
The symplectic representation of $\vec{L P_0}$ is given by $\vec A_0 \vec r_0 + \vec B_0 \vec y_0 + \vec e$, where $\vec e = \Symp(\vec E)$, and $\vec r_0, \vec y_0$ are uniformly random.
Choose random invertible matrices $\vec Q_1 \in \Z_2^{n \times n}$ and $\vec Q_2^{k \times k}$, and define $\vec A := \vec A_0 \vec Q_1$, $\vec B := \vec B_0 \vec Q_2$, $\vec r := \vec Q_1^{-1} \vec r_0$, and $\vec y := \vec Q_2^{-1} \vec y_0$.
One can now check that these objects follow the distribution of a Search $\lsn$ sample $([\vec A \,|\, \vec B], \vec{Ar} + \vec{By} + \vec e)$.

Next, $\CB$ will set $I = (\vec A, \vec B)$ and $f_I(x) = \vec{Ar} + \vec{By} + \vec e$ before running $\CA$ on input $(I, f_I(x))$, yielding a guess \begin{align}
    (\hat{\vec r}, \hat{\vec y}, \hat{\vec e}) \in \Z_2^{n} \times \Z_2^{k} \times \CW_{n \d} ,
\end{align}
where $\hat{\vec e}$ comes from a guess of the seed $\hat{\vec s}$ which is then deterministically mapped into a corresponding depolarizing error.
By assumption, with probability at least $1/\poly(n)$, \begin{align}
    \vec{A} \hat{\vec r} + \vec{B} \hat{\vec y}  + \hat{\vec e}= \vec{Ar} + \vec{By} + \vec e .
\end{align}
Thus, $\vec e + \hat{\vec e} = \vec A(\hat{\vec r} + \vec r) + \vec B(\hat{\vec y} + \vec y)$ commutes with all stabilizers.
For $(\delta, R)$ satisfying the quantum Gilbert-Varshamov bound, a random $n$-qubit stabilizer code with rate $R = k/n$ has distance at least $\delta n$ with probability $1 - \negl(n)$~\cite{khesin2025universal,poremba2024learning,nielsen2010quantum}.
Note that $|\hat{\vec e}|$ and $|\vec e|$ are both at most $2.01np$.
Since the distance of the corresponding quantum stabilizer code is at least $4.03np$ (more than twice the error weight) with probability $1 - \negl(n)$, $\vec e + \hat{\vec e}$ cannot be a logical operator.
Since $\vec e + \hat{\vec e}$ commutes with all stabilizers, it must itself be a stabilizer.
Hence, by a union bound, with probability at least $\frac{1}{\poly(n)} - \negl(n) = \frac{1}{\poly(n)}$, $\hat{\vec e}$ differs from $\vec e$ by a stabilizer.
Now, $\CB$ can apply $\Symp^{-1}(\hat{\vec e})$ to the original noisy code state, yielding $\vec C \ket{0^{n-k}, \psi}$.
Inverting $\vec C$, we recover $\ket{\psi}$ exactly. 
Thus, the total fidelity between our reconstructed state and $\ket{\psi}$ is at least $1/\poly(n)$ as desired.
\end{proof}

\section{Additional proofs}
\label{app:additional_proofs}

We here provide proofs which were deferred in the main text.

\begin{proof}[Proof of \Cref{lemma:randomizing_symp_hyperplanes}]
We first check directly that the above construction preserves the symplectic inner product, and then show that this Clifford matrix produces a uniformly random basis of a uniformly random isotropic subspace of dimension $n-1$.
In the edge case $\vec{C} = \vec I$, it is clearly symplectic.
Otherwise, the sampling procedure constructed $\vec{C}$ as $\vec{\Pi} \vec{C}_0$, so it suffices to show separately that $\vec C_0$ and $\vec \Pi$ are symplectic matrices.
First, for $\vec \Pi$, we need to show that $(\vec{\Pi} \vec e_i) \odot (\vec {\Pi} \vec e_j) = 0$, $(\vec{\Pi} \vec f_i) \odot (\vec {\Pi} \vec f_j) = 0$, and $(\vec{\Pi} \vec e_i) \odot (\vec {\Pi} \vec f_j) = \delta_{ij}$. 
Once these relations are shown, by bilinearity of the inner product $\vec{\Pi}$ must be symplectic.
However, since $\vec{\Pi}$ applies the \emph{same} permutation on each pair $(\vec e_i, \vec f_i)$, these relations are immediate. 

Meanwhile, we observe that $\vec{C}_0$ translates $\vec e_j$ and $\vec f_j$ by a multiple of $\vec e_1$ for any $j = 2, \dots, n$. 
Their orthogonality relations with each other are therefore preserved, because $\vec e_1$ is orthogonal to each $\vec e_j$, $\vec f_j$, and to itself. Moreover, for any $j=2, \dots, n$, $\vec{C}_0\vec{e}_j$ and $\vec{C}_0 \vec{f}_j$ do not include any $\vec{f}_1$, so they are also orthogonal to $\vec{C}_0\vec{e}_1 = \vec{e}_1$. 
Since $\vec C_0 \vec{e}_1 = \vec{e}_1$, $(\vec C_0 \vec{e}_1) \odot (\vec C_0 \vec{e}_1) = 0$ as desired. 
It only remains to check the orthogonality relations for $\vec C_0\vec f_1$:
\begin{align}
    \vec{C}_0\vec e_1 \odot \vec{C}_0\vec f_1 & = \vec e_1 \odot \vec{r}' =1\\
    \vec{C}_0\vec f_i \odot \vec{C}_0\vec f_1 & = (\vec f_i + (\vec r' \odot \vec f_i)\vec e_1) \odot \vec{r}'  = 0 \quad \text{for $2 \leq i\leq n$} ,\\
    \vec{C}_0\vec e_i \odot \vec{C}_0\vec f_1 & = (\vec e_i + (\vec r' \odot \vec e_i)\vec e_1) \odot \vec{r}' =\delta_{i1} \quad \text{for $2 \leq i\leq n$}
\end{align}
where we use in each equality that $\vec{e}_1 \odot \vec{r}' = 1$, since $r'_{n+1} = 1$.  
We conclude that $\vec C$ is symplectic, since both $\vec C_0$ and $\vec \Pi$ are. 

Next, we argue that $\CR_n$ maps a uniformly random basis of a random $(n-1)$-dimensional isotropic subspace orthogonal to $\vec f_{1}$ and not containing $\vec f_{1}$ to (up to negligible total variation distance) a uniformly random basis of a random $(n-1)$-dimensional isotropic subspace.
For any $\vec{w} \neq 0$, let 
\begin{equation}
S_\vec{w} = \{(\vec{u}_1, \dots, \vec{u}_{n-1}) \;|\; (\vec{u}_i) \text{ basis for isotropic } V \subseteq \Z_2^{2n}, \vec{w} \in V^\perp\setminus V \}. 
\end{equation}
Condition on the existence of a pivot $k$, an event of probability $1-2^{-n}$. Then, $\vec{C}$ maps $S_{\vec{f}_1}$ bijectively to $S_{\vec{r}}$, with inverse $\vec{C}^{-1}$ (in particular, all $S_\vec{w}$ have the same cardinality). 
Hence, since $(\vec{v}_1, \dots, \vec{v}_{n-1})$ is a uniformly random element of $S_{\vec{f}_1}$, $(\vec{w}_1, \dots, \vec{w}_{n-1})$ is a uniformly random element of $S_{\vec{r}}$. 
Under this conditioning, $\vec r$ is uniformly random over vectors with $r_{n+i}=1$ for some $1\leq i\leq n$, so its distribution is within total variation distance $2^{-n}$ of a uniformly random nonzero vector. 
Hence, $(\vec{w}_1, \dots, \vec{w}_{n-1})$ is statistically indistinguishable from a random element $(\vec{u}_1, \dots, \vec{u}_{n-1})$ of a uniformly randomly $S_{\vec{w}}$. 

But notice that any basis $(\vec{u}_1, \dots, \vec{u}_{n-1})$ of any isotropic $V$ is included in exactly $2^{n+1}-2^{n-1}$ sets $S_{\vec{w}}$---one for each $\vec{w} \in V^\perp\setminus V$.
Each $S_{\vec{w}}$ has exactly the same cardinality, and every $(\vec{u}_1, \dots, \vec{u}_{n-1})$ is present in exactly $2^{n+1}-2^{n-1}$ of these sets.
It follows that a random element of a randomly selected $S_{\vec{w}}$ is identical in distribution to a uniformly random element of 
\begin{equation}
S = \{(\vec{u}_1, \dots, \vec{u}_{n-1}) \;|\; (\vec{u}_i) \text{ basis for isotropic } V \subseteq \Z_2^{2n}\}. 
\end{equation}
Removing the conditioning changes the distribution by at most $2^{-n}$ in total variation distance. Every $(n-1)$-dimensional isotropic subspace has exactly the same number of bases, so we conclude that $(\vec{w}_1, \dots, \vec{w}_{n-1})$ has negligible total variation distance from a uniformly random basis of a uniformly random $(n-1)$-dimensional isotropic subspace $V$. 
\end{proof}

\begin{proof}[Proof of \Cref{lemma:noise_symmetrization}]
Sample $T \sim \Bin(n, \frac{4}{3} q)$, so that $\mathbb E[T] = \frac{4}{3}nq \geq \frac{4}{3}m$.
We use the Chernoff bound on the binomial distribution, i.e. for $X \sim \Bin(n, p)$, $\Pr[X \leq (1-\d)\mathbb{E}[X]] \leq \exp{-\d^2 \mathbb{E}[X] / 2}$.
Thus, \begin{align}
    \Pr[T\leq m]\leq\Pr[T\leq\tfrac34\mathbb E[T]]\leq\exp(-\mathbb E[T]/32)\leq\exp(-m/24)=\negl(n) 
\end{align}
because $m = \w(\log n)$.
If $T<m$, output $\vec e'=\vec 0$. Otherwise, proceed as follows.
Sample $\vec e'$ by choosing $T$ indices, including the $m$ indices in $M$ and $T-m$ arbitrary others; on each chosen index $j$, sample $(e'_j,e'_{n+j})$ independently from $\CD_{3/4}$, and set all other pairs to $(0,0)$.
Now, $\vec e + \vec e' \sim \vec e'$ since the noise in $\vec e'$ manifestly subsumes the noise in $\vec e$.
Next, the random permutation $\pi$ scrambles any index asymmetry.
Conditioned on $T\geq m$, the random permutation makes the chosen index set uniform among all $T$-element subsets. Without this conditioning, sampling $T\sim\Bin(n,\frac43q)$ and then a uniform $T$-element subset is equivalent to selecting each index independently with probability $\frac43q$; placing independent $\CD_{3/4}$ noise at the selected indices therefore gives $\CD_q^{\otimes n}$.
The actual output and this comparison distribution can be coupled to agree whenever $T\geq m$. Thus, the total variation distance between the law of $\pi(\vec e+\vec e')$ and $\CD_q^{\otimes n}$ is at most $\Pr[T<m]=\negl(n)$.
\end{proof}

\begin{proof}[Proof of \Cref{thm:pke_correctness}]
Let $\mu \in \bit$ be any message bit. 
Suppose that $(\pk,\sk) \leftarrow \gen(1^n)$ and $\mathsf{ct} \leftarrow \mathsf{Enc}(\pk,\mu)$. 
We can parse the two keys as $\pk = (\vec A,\vec b = \vec A \vec x + \vec e)$ and $\mathsf{sk}=\vec x$, and we can parse the ciphertext as 
 $\mathsf{ct}=(\vec f \odot \vec A, \vec f \odot \vec b + \mu)$, for some $\vec e,\vec f \sim \CD_p^{\otimes n}$.
Then, the output of $\mathsf{Dec}(\sk,\mathsf{ct})$ on input $\mathsf{ct}=(\vec u,c)$ is
\begin{align}
    c + \vec u \cdot \vec x &=    \vec f \odot \vec b + \mu + (\vec f \odot \vec A)\cdot\vec x = \vec f \odot \vec A \vec x + \vec f \odot \vec e + \mu+ \vec f \odot \vec A \vec x \\
    & = \mu + \vec f \odot \vec e \pmod{2} \, ,
\end{align}
where we used that $(\vec f \odot \vec A) \cdot \vec x = \vec f \odot (\vec A \cdot \vec x)$.
Therefore, it suffices to show that for appropriate choices of $p \in (0,1)$, $\vec f \odot \vec e =0$ with high probability.

Recall that $\vec e, \vec f$ are two independent depolarizing errors, so we may express $\vec e = (\vec a, \vec b)$ where $(a_i, b_i)$ are i.i.d. (over $i$) distributed such that $(0, 0)$ occurs with probability $1 - p$ and the remaining 3 possibilities $(1, 0), (0, 1), (1, 1)$ occur each with probability $p/3$. 
Similarly, express $\vec f = (\vec a', \vec b')$.
Thus, \begin{align}
    \vec e \odot \vec f = \sum_{i=1}^n a_i b_i' + a_i' b_i .
\end{align}
Note that $a_i b_i' + a_i' b_i = 1$ occurs with probability $r := \frac{2}{3} p^2$, since there are six possible ways to satisfy $a_i b_i' + a_i' b_i = 1$, each occurring with probability $(p/3)^2 = p^2/9$.
Hence, $a_i b_i' + a_i' b_i \sim \Ber(r)$, and $\Pr[\vec e \odot \vec f = 0]$ is precisely the probability that a $\Bin(n, r)$ random variable is even.
The latter probability is well-known to be $\frac{1}{2} + \frac{1}{2} (1 - 2r)^n$, so \begin{align}
    \label{eq:sympLPN_correctness_probability}
    \Pr[\vec e \odot \vec f = 0] = \frac{1}{2} + \frac{1}{2} \left(1 - \frac{4}{3} p^2 \right)^n \geq \frac{1}{2} + \frac{1}{2} \exp{- \frac{n \frac{4}{3} p^2}{1 - \frac{4}{3} p^2}}
\end{align}
by the standard inequality $\ln(1-x) \geq - \frac{x}{1-x}$ for $x \in (0, 1)$. 
Thus, for any $\delta >0$, we can choose $p = \Theta\left(1/\sqrt{n} \right)$ such that
the scheme is $1-\delta$ correct.
\end{proof}
\begin{proof}[Justification of the amplification remark] Except with negligible probability over key generation, $w_P(\vec e)\leq2np$. Conditional on any such $\vec e$, the repetitions decrypt independently with success probability $\frac12+\frac12(1-\frac43p)^{w_P(\vec e)}=\frac12+\Omega(1)$. A Chernoff bound therefore gives negligible majority error for $2\lceil\log^2 n\rceil+1$ repetitions. IND-CPA security follows by a standard hybrid argument, and key generation is unchanged. \end{proof}

\section{Barriers to a reduction from $\symplpn$ to $\lpn$}
\label{app:reduction_barrier}

It is known that there exists a reduction from $\lpn[k, 2n, p/6]$ to $\symplpn[n, n, p]$, for $k = \lfloor pn/6 \rfloor$~\cite{khesin2025average}. 
This reduction is only meaningful for certain regimes of $p$.
$\lpn(k, n, p)$ is easy for any $p = O(k^{-1})$, as a $k \times k$ sub-block of the encoding matrix has substantial probability of experiencing no error at all, and thus being directly invertible. 
Meanwhile, when $p = \Omega(k^{-c})$ for $c < 1$, there is no known algorithm that runs in polynomial time in $k$. 
It follows that the reduction from $\lpn$ to $\symplpn$ gives a strong lower bound on the hardness of $\symplpn$ when $p = \Omega\left(n^{-c}\right)$ for $c < 1/2$. 
However, in the regime for our proposed \textsf{PKE} and \textsf{OT} schemes, $p = \Theta\left(n^{-1/2}\right)$, so the reduction becomes vacuous. 
It is therefore open as to whether $\symplpn$-based schemes inherit security from the hardness of $\lpn$. 

A more pressing concern one might raise with low-noise $\symplpn$-based proposals is, on the other hand, the possibility of a \emph{converse} reduction from $\symplpn$ to $\lpn$ with $p = \Theta(1/\sqrt{n})$.
If such a reduction exists, then any low-noise $\symplpn$-based scheme is no more secure than one based off of low-noise $\lpn$.
We here show that there are significant barriers to natural linear reductions from $\symplpn[n, p]$ to $\lpn$.
Thus, despite their similarities, low-noise $\lpn$ and $\symplpn$ are not currently known to be reducible to one another.
It is unknown as to how the security of the two cryptographic hardness assumptions rigorously compare in theory and in practice.

Given a sample of $\symplpn$, $(\vec A, \vec r)$, where $\vec r$ is either of the form $\vec A\vec x+\vec e$ or $\vec u$, a natural class of reductions would proceed by preparing
\begin{equation}
(\vec B\vec A, \vec B\vec r),
\end{equation}
where $\vec B\in\mathbb Z_2^{m\times2n}$ is sampled using $\vec A$ and fresh randomness.
If $\vec r = \vec u$, then $\vec B\vec r$ is uniformly random as long as $\vec B$ has full row rank; if $\vec r = \vec A \vec x + \vec e$ then $\vec B \vec r = \vec B \vec A \vec x + \vec B \vec e$. 
If one could show that $\vec B \vec A$ were statistically close to a uniform distribution, while $\vec B \vec e$ were independent Bernoulli noise, independent of $\vec{BA}$, then this approach would constitute a reduction from $\symplpn$ to $\lpn$.
We focus on the regime $m\geq cn$ for a fixed $c>1$.
(If we allowed $m \leq n$, then one could simply set $\vec B$ to remove the last $(1+\delta)n$ rows of $\vec A$ for any constant $\delta \in (0,1)$ with $(1-\delta)n \in \mathbb{N}$.
Then $\vec{BA} \in \Z_2^{(1-\delta)n \times n}$ is indeed close to uniformly random, but produces an information-theoretically unsolvable instance of $\lpn$.)

It is not clear at all, a priori, that such a reduction would not exist. A random isotropic code $\vec A \in \mathbb{Z}_2^{2n \times n}$ has roughly $\frac{3}{2}n^2$ bits of entropy, since there are ${n \choose 2}$ symplectic orthogonality conditions between the pairs of columns. 
Hence, choosing $m = cn$ for $1<c < \frac{3}{2}$, it is a priori possible that even for some fixed $\vec B$, $\vec B \vec A$ could be statistically indistinguishable from a uniform distribution. 
However, we demonstrate that the random variable $\vec A$ cannot be randomized in this manner---in fact, $\vec{BA}$ for any fixed $\vec{B}$ is severely deficient in entropy. 
In what follows, let $H(X)$ denote the entropy of a random variable $X$.

\begin{theorem}\label{thm:no_randomized_symplpn}
Let $m \geq cn$ for $c > 1$. Then there exists a constant $d > 0$, such that for sufficiently large $n$ and any fixed choice of $\vec B \in \mathbb{Z}_2^{m \times 2n}$, $H(\vec B \vec A) \leq (1-d)mn$.
\end{theorem}

This is a strong bound because we require $\vec{BA}$ to be close in total variation distance to a uniformly random $m \times n$ matrix in the reduction to $\lpn[n, m, p']$.
To satisfy this closeness, $\vec{BA}$ must have entropy negligibly close to $mn$.
But by \Cref{thm:no_randomized_symplpn}, any fixed $\vec B$ is off in entropy by at least a constant factor, and thus $\vec B$ itself must be chosen from a distribution random enough to supplement $\W(n^2)$ bits of entropy.
Intuitively, when $\vec B$ has that much entropy, it cannot possibly be very sparse, and therefore it should be forced to blow up the error $\vec{Be}$ in the transformed codeword $(\vec{BA})\vec x + \vec{Be}$, beyond even the information-theoretic decoding limit perhaps.
However, placing this intuition on rigorous grounds to prove the implication requires additional technical subtlety.
We show that randomizing $\vec{BA}$ forces $\vec{Be}$ to have large expected weight.
At constant rate, making the errors independent by adding fresh noise then leaves too much noise for reliable decoding.

This result holds for $\symplpn$ with error rate $p = \w(1/n)$.
If $p = O(1/n)$, then there is a polynomial-time brute-force error enumeration algorithm to solve $\symplpn[n, n, p]$, so this ultra-low-noise regime is of no cryptographic interest.

\begin{theorem}\label{thm:symplpn_error_blowup}
Let $m=\poly(n)$ with $m > cn$ for some $c > 1$, and let $p = \w(1/n)$ with $0 < p \leq 3/4$.
Suppose that $\vec B \in \Z_2^{m \times 2n}$ is a random variable such that $\vec{BA} \in \Z_2^{m \times n}$ is statistically indistinguishable from uniformly random.
Fix any $\delta>0$ and let $r := n/m$ be the rate of the transformed code $\operatorname{im} (\vec{BA})$. 
Then $|\vec{Be}|$, the weight of the distorted error, satisfies
\begin{equation}
    \Pr_{\vec B} \left[ \underset{\vec{e} \sim \CD_p^{\otimes n}}{\mathbb{E}} [|\vec{Be}|] \geq \left(\frac{1 - r - \delta}{2}\right)m \right] \geq 1 - \negl(n) .
\end{equation}
\end{theorem}
In other words, this theorem states that if $\vec{B}$ can randomize the code distribution well, $\vec{Be}$ has expected weight lower bounded by $(\frac{1-r-\delta}{2})m$ for any $\delta > 0$. 
For $m = O(n)$, we can interpret this as a strong barrier on decoding. 
Indeed, $\frac{1-r}{2} > H_2^{-1}(1-r)$ for any $r \in (0, 1)$, where $H_2(x) := -x \log x - (1-x)\log(1-x)$ is the binary entropy. 
For independent Bernoulli noise of rate $p'$, Shannon's noisy coding converse rules out reliable decoding when $H_2(p')>1-r$.
By \Cref{lem:depolarizing_weight}, each bit of $\vec{Be}$ has probability at most $1/2$ of being $1$, so adding noise independent of $\vec e$ conditional on $(\vec A,\vec B)$ cannot decrease its expected weight. Thus, if the resulting error is, jointly with $\vec{BA}$, negligibly close to independent $\Ber(p')^{\otimes m}$ noise with $p'\leq1/2$, then $p'\geq(1-r-\delta)/2-\negl(n)>H_2^{-1}(1-r)$ for sufficiently small fixed $\delta>0$ in the constant-rate regime.

When $m = \omega(n)$, the result shows that the output error has expected relative weight at least $\frac{1}{2} - \delta$ for any constant $\delta$, since $r = o(1)$. This bound is not sufficient to fully rule out decodability. 
We believe that it should be possible to do so for any $m = \poly(n)$---however, the primary case of interest for the result was to rule out a reduction where $m = (1+\epsilon)n$, i.e. when the number of rows decreased rather than increasing. 
It seems unlikely that reductions that resort to adding many more rows will do a better job of randomizing the code instance without amplifying the error. 

In order to prove these theorems, we will begin by introducing new notions in symplectic linear algebra and establishing basic lemmas that will be necessary for the proofs. 
First, we define the \emph{radical} of a vector space $V \subseteq \Z_2^{2n}$. 
\begin{definition}
    For a vector space $V \subseteq \Z_2^{2n}$, define the \emph{radical} $\operatorname{rad}(V) := V \cap V^\perp$, i.e. it is the set of all vectors $\vec v \in V$ such that $\vec v \odot \vec w = 0$ for any $\vec w \in V$.  
\end{definition}
For example, $\operatorname{rad}(\Z_2^{2n}) = \{\vec 0\}$, and $\operatorname{rad}(\operatorname{span}(\set{\vec v})) = \operatorname{span}(\set{\vec v})$, since $\vec v$ is always symplectically orthogonal to itself.
% As mentioned, in $\Z_2^{2n}$, $\text{dim } V + \text{dim } V^{\perp} = 2n$. 
% It follows that for any $\vec w \neq \vec 0$, $\text{dim } \text{span}(\vec w)^\perp  = 2n - 1$, and in particular $\vec w$ is \emph{not} orthogonal to some nonzero $\vec v \in \Z_2^{2n}$. 

We will now prove an upper bound on the dimension of $\text{rad}(V)$ for a uniformly random $V \subseteq \Z_2^{2n}$ of a given dimension.
In order to do so, we first establish that a random vector space $V \subseteq \Z_2^{2n}$ is very unlikely to be isotropic. 
\begin{lemma}[Random subspaces are not isotropic]
\label{lem:isotropic_count}
Let $\epsilon > 0$ be a constant. 
If $V \subseteq \Z_2^{2n}$ is a random subspace of dimension $m \geq \epsilon n$, then there exists a constant $\delta>0$ so that, for all sufficiently large $n$, \begin{equation}
\Pr[V \subseteq V^\perp] \leq 2^{-\delta n^2}.
\end{equation}
\end{lemma}
\begin{proof}
If $m>n$, the probability is zero, so assume $m\leq n$. Instead of $V$, we consider a random matrix $\vec{M}$ formed by the following random process: new uniformly random (linearly independent) columns are added to $\vec{M}$ until it has dimension $m$. 
The span of such a random matrix is a uniformly random subspace, since it is obtained by repeatedly adding new random vectors to a given subspace. 
However, the chance that the result is isotropic is at most $2^{-\delta_0 m^2}$ for some $\delta_0>0$. 
Indeed, say that $k$ columns have been added to $\vec{M}$ and their span is isotropic. Then, the chance that the subsequent column lies in the symplectic dual of the previous ones is $\frac{2^{2n - k} - 2^k}{2^{2n} - 2^k}$, because the current column span has dimension $k$ (none of the vectors in this span can be added) and the symplectic dual has dimension $2n - k$. 
Since $m \leq n$ for any isotropic vector space, the value of $k$ must satisfy $k < n$, and therefore \begin{align}
    \frac{2^{2n - k} - 2^k}{2^{2n} - 2^k} \leq \frac{2^{2n - k}}{2^{2n} - 2^k} \leq 2 \frac{2^{2n - k}}{2^{2n}} = \frac{1}{2^{k-1}} .
\end{align}
Taking the product of these values from $k = 1$ to $k = m-1$ (the probability is exactly $1$ for $k = 0$) yields $2^{-\frac{(m-1)(m-2)}{2}}$. 
In particular, since $m \geq  \epsilon n$, there exists some $\delta$ for which $\delta n^2 \leq \frac{(m-1)(m-2)}{2}$ for all sufficiently large $n$, and therefore $2^{-\delta n^2} \geq 2^{-\frac{(m-1)(m-2)}{2}}$. 
\end{proof}
We also record a standard result on the number of subspaces $V$ satisfying $U \subseteq V \subseteq W$ of fixed dimension.

\begin{lemma}[Number of sandwiched subspaces]
\label{lem:subspace_count}
There exist constants $C_1, C_2>0$ such that for any vector spaces $U \subseteq W$ over $\Z_2$ of respective dimensions $k$ and $m$, the number $N$ of vector spaces $V$ of dimension $l$ with $k\leq l\leq m$ for which $U \subseteq V \subseteq W$ satisfies 
\begin{equation}
C_12^{(l - k)(m-l)} \leq N \leq C_2 2^{(l-k)(m-l)}. 
\end{equation}
\end{lemma}
Using the previous two lemmas, we may bound the size of $\operatorname{rad}(V)$ for a random subspace $V$ of fixed dimension.
Recall that if $W \subseteq \Z_2^{2n}$ has dimension $t$, then $\dim(W^\perp) = 2n -t$.

\begin{lemma}[Dimension of a random radical]
\label{lem:radical_count}
Let $\epsilon > 0$ and $\a > 0$ be constants. If $V \subseteq \Z_2^{2n}$ is a random subspace of dimension $m \geq \a n$, then there exists a constant $\delta > 0$ such that, for all sufficiently large $n$, 
\begin{equation}
\Pr[\operatorname{dim}(\operatorname{rad}(V)) \geq \epsilon n] \leq 2^{-\delta n^2}.
\end{equation}
\end{lemma}

\begin{proof}
Let $k = \lceil \epsilon n \rceil$. If $k>\min\{m,2n-m\}$, the probability is zero; otherwise, assume $k\leq m\leq 2n-k$.
A vector space $V$ has radical of dimension at least $\epsilon n$ if and only if there is some isotropic subspace $W$ with dimension $k$ such that $W \subseteq V \subseteq W^\perp$. 
Indeed, a subspace $W \subseteq V$ is contained in the radical of $V$ if and only if $V \subseteq W^\perp$, i.e. every vector in $W$ is orthogonal to every vector in $V$. 

The number $N_k^{\text{iso}}$ of isotropic vector spaces of dimension $k$, by \cref{lem:isotropic_count} and \cref{lem:subspace_count}, satisfies $N_k^{\text{iso}} \leq C_22^{k(2n - k) - \delta' n^2} $ for some $\delta' > 0$. 
Meanwhile, for some fixed isotropic $W$ of dimension $k$, let $M$ be the number of vector spaces $V$ for which $W \subseteq V \subseteq W^\perp$. 
By \cref{lem:subspace_count}, since $\dim(W^\perp) = 2n - k$, $ M \leq C_22^{(m - k)(2n - m-k) }$.
It follows that the number of vector spaces $V$ of dimension $m$ with radical of dimension at least $k$ has upper bound 
\begin{equation}
M N_k^{\text{iso}} \leq C_2^2 2^{k(2n - k) + (m - k)(2n - m-k) - \delta' n^2}.
\end{equation}
Meanwhile, the total number of vector spaces $N_m$ of dimension $m$ satisfies the lower bound $C_1 2^{m(2n-m)}\leq N_m$. 
Hence, the fraction of vector spaces $V$ of dimension $m$ with radical of dimension at least $k$ is at most 
\begin{align}
\frac{M N_k^{\text{iso}}}{N_m} &\leq \frac{C_2^2}{C_1} 2^{k(2n - k) + (m - k)(2n - m-k) - m(2n-m) - \delta 'n^2}\\
&= D2^{ - \delta' n^2},
\end{align}
where $D := \frac{C_2^2}{C_1}$. 
This shows that a randomly sampled $V$ of dimension $m$ has a radical of dimension at least $\epsilon n$ with probability at most $D2^{-\delta' n^2}$. Taking $\delta=\delta'/2$, we have $D2^{-\delta'n^2}\leq 2^{-\delta n^2}$ for all sufficiently large $n$, giving the desired result. 
\end{proof}

With these lemmas, we can now prove a technical result that will be instrumental for \Cref{thm:no_randomized_symplpn}. 
\begin{lemma}[Symmetrized product of random matrices is nearly full rank]
\label{lemma:random_product_rank}
Let $\epsilon > 0$ be a constant. 
Say that $\vec C, \vec B$ are independent, uniformly random $m \times n$ matrices, such that $\a n \leq m \leq n$ for some constant $\a \in (1/2, 1)$. 
Then for sufficiently large $n$, there exists a constant $\d > 0$ such that
\begin{align}
    \Pr[\operatorname{rank}(\vec C^\intercal\vec B + \vec B^\intercal\vec C) \leq (1-\epsilon)n] \leq 2^{-\d n^2} .
\end{align}
\end{lemma}
\begin{proof}
It suffices to prove the claim for $0<\epsilon<1$. Express
\begin{equation} \label{eq:description_of_anticommutant}
\vec C^\intercal\vec B + \vec B^\intercal\vec C = 
\begin{bmatrix} \vec C^\intercal & \vec B^\intercal \end{bmatrix}\begin{bmatrix} 0 & \vec I \\ \vec I & 0 \end{bmatrix}\begin{bmatrix} \vec C \\ \vec B \end{bmatrix}
\end{equation}
We first may note that $\vec T := \begin{bmatrix} \vec C \\ \vec B \end{bmatrix} \in \mathbb{Z}_{2}^{2m \times n}$ is uniformly random, and therefore since $2m \geq n$ it has rank at least $\left(1 - \frac{\epsilon}{2}\right) n$ with probability at least $1 - 2^{-\delta_0 n^2}$ for some $\delta_0 > 0$.
To see this, note that in order for $\vec{T}$ to have rank less than $(1-\frac{\epsilon}{2})n$, then there must be at least $\frac{\epsilon n}{2}$ columns that are contained in the span of the previous columns (whose rank is bounded by $(1-\frac{\epsilon}{2})n$). 
For any fixed set of $\lceil\epsilon n/2\rceil$ column positions, the probability that each of these columns lies in the span of its predecessors while that span has dimension less than $(1-\epsilon/2)n$ is at most $2^{-\epsilon^2n^2/4}$. Taking a union bound over at most $2^n$ choices of these positions gives a rank-failure probability at most $2^{n-\epsilon^2n^2/4}$. 
Hence, setting $\delta_0=\epsilon^2/8$ yields the desired bound for sufficiently large $n$. 

Therefore, with this probability, $W := \operatorname{im}(\vec T) \subseteq \mathbb{Z}_2^{2m}$ has dimension at least $\left(1 - \frac{\epsilon}{2}\right)n$. 
Let $t = \dim(W)$, and note that $W$ is a uniformly random subspace with dimension $t$.
We must now bound the dimension of $\operatorname{rad} (W)$, because this dimension constrains the rank of $\vec C^\intercal \vec B + \vec B^\intercal \vec C$. 
Indeed, say that some vector $\vec u$ satisfies $(\vec C^\intercal \vec B + \vec B^\intercal \vec C) \vec u = 0$. 
Then, by Eqn.~\eqref{eq:description_of_anticommutant} this is equivalent to $(\vec T \vec u') \odot (\vec T \vec u) = 0$ for any $\vec u'$, i.e. that $\vec T \vec u \in \operatorname{rad}(W)$. 
It follows that with probability at least $1 - 2^{-\delta_0n^2}$,
\begin{align}
\text{dim}(\text{ker}(\vec{C}^\intercal \vec B + \vec B^\intercal \vec C)) &= \text{dim}(\text{rad}(W)) + \text{dim}(\text{ker}(\vec T))\\ & \leq \text{dim}(\text{rad}(W)) + \frac{\epsilon}{2} n ,
\end{align}
by the rank-nullity theorem.
By \Cref{lem:radical_count}, there exists some $\delta_1$ for which the probability that $\dim(\operatorname{rad}(W)) \geq \frac{\epsilon}{2}n$ is at most $2^{-\delta_1 n^2}$. 
Choosing $\delta < \min(\delta_0, \delta_1)$, we have that $2^{-\delta_0 n^2} + 2^{-\delta_1n^2} \leq 2^{-\delta n^2}$ for sufficiently large $n$ as claimed.
\end{proof}

With the tools in place, we may now prove \Cref{thm:no_randomized_symplpn}. 
Before we give the proof, we recall the existence of a standard, useful basis for any subspace $V \subseteq \Z_2^{2n}$. 
The standard symplectic basis $\vec e_1, \dots, \vec e_n, \vec f_1, \dots, \vec f_n$ of $\Z_2^{2n}$ has convenient orthogonality relations: $\vec e_i \odot \vec e_j = 0$, $\vec f_i \odot \vec f_j = 0$, and $\vec e_i \odot \vec f_j = \delta_{ij}$. 
For a general $V$, there is a similar convenient basis. 
\begin{lemma}[Symplectic basis of a subspace] \label{lem:symplectic_subspace_basis}
Let $V \subseteq \Z_2^{2n}$ be a subspace. 
There exists a basis $\vec u_1, \dots, \vec u_k, \vec v_1, \dots, \vec v_l, \vec w_1, \dots, \vec w_l$ such that
\begin{enumerate}
\item[(1) ] $\vec{u}_i \in V^\perp$ for $i = 1, \dots, k$.
\item [(2) ]$\vec{v}_i \odot \vec{v}_j = 0$ for $i, j = 1, \dots, l$.
\item[(3) ] $\vec{w}_i \odot \vec{w}_j = 0$ for $i, j = 1, \dots, l$.
\item[(4) ] $\vec{v}_i \odot \vec{w}_j = \delta_{ij}$ for $i, j = 1, \dots, l$. 
\end{enumerate}
\end{lemma}
This basis is almost precisely analogous to $(\vec e_i, \vec f_i)$ in the case of $\Z_2^{2n}$---the $\vec v_i$ correspond to the $\vec e_i$, while the $\vec w_i$ correspond to the $\vec f_i$. 
The additional vectors $\vec u_i$ give a basis for $\text{rad}(V)$.

\begin{proof}[Proof of \Cref{thm:no_randomized_symplpn}]
Let $\mathcal{S}$ denote the set of $2n \times n$ isotropic matrices, and let $F_\vec{B} = \{\vec{BA} \,|\, \vec{A} \in \mathcal{S} \}$. 
We will start by showing that for some constant $d' > 0$, for any $\vec{B}$, $|F_{\vec{B}}| \leq 2^{mn - d'n^2}$.
This implies that $H(\vec{BA}) \leq mn - d'n^2$, and with a short additional argument we will conclude that $H(\vec{BA}) \leq (1-d)mn$ for some $d > 0$. 
It suffices to consider $1<c<2$. We may also assume $m\leq 2n$ and that $\vec B$ has full row rank. Indeed, for $m\leq 2n$, extend a basis of the row space of $\vec B$ to $m$ linearly independent rows; the original product $\vec{BA}$ is a fixed linear image of the resulting product, so its support cannot be larger. For $m>2n$, the case $\vec B=\vec I_{2n}$ gives $|F_{\vec B}|\leq|\mathcal S|\leq 2^{2n^2-d'n^2}\leq 2^{mn-d'n^2}$ for every such $\vec B$.
First, note that $|F_{\vec{B}}| = |F_{\vec{TBS}}|$, where $\vec{T}$ is any invertible $m \times m$ matrix, and $\vec{S} \in \mathbb{Z}_2^{2n \times 2n}$ is any matrix that preserves the symplectic inner product. 
Indeed, consider the map that takes $\vec{M} \in F_{\vec{B}}$ to $\vec{TM}$. 
Writing $\vec{M} = \vec{BA}$, $\vec{TM} = (\vec{TBS})(\vec{S}^{-1}\vec{A})$, and hence $\vec{TM} \in F_{\vec{TBS}}$ because $\vec S^{-1} \vec A$ is an isotropic matrix. 
Likewise, the inverse map $\vec M \mapsto \vec T^{-1} \vec M$ for $\vec M = (\vec{TBS}) \vec A$ also places $\vec T^{-1} \vec M \in F_{\vec B}$.
Thus, $|F_{\vec{B}}| = |F_{\vec{TBS}}|$.

Let $W := \ker(\vec B) \subseteq \mathbb{Z}_2^{2n}$. 
By \Cref{lem:symplectic_subspace_basis}, $W$ has a basis of the form $\vec u_1, \dots, \vec u_k, \vec v_1, \dots, \vec v_{l}, \vec w_1, \dots, \vec w_l$, where $\vec v_i \odot \vec w_j = \delta_{ij}$, $\vec v_i \odot \vec v_j = 0$, $\vec w_i \odot \vec w_j = 0$, and $\vec u_i \in W^\perp$. 
Let $\vec{S}$ be a symplectic matrix which sends the standard symplectic basis vectors $\vec e_1, \dots, \vec e_{l+k}, \vec f_1, \dots, \vec f_l$ to $\vec v_1, \dots, \vec v_l, \vec u_1, \dots, \vec u_k, \vec w_1, \dots, \vec w_l$. 
Then, $\ker({\vec{BS}})$ is spanned by $\vec e_1, \dots, \vec e_{l+k}, \vec f_1, \dots, \vec f_l$, while $\operatorname{im}(\vec{BS})$ is spanned by the images of $\vec e_{l+k+1}, \dots, \vec e_n, \vec f_{l+1}, \dots, \vec f_n$. 
By choosing an appropriate invertible $\vec{T} \in \Z_2^{m \times m}$, we can construct a $\vec P = \vec{TBS}$ which is a projector onto this latter set of coordinates. As a consequence, $|F_\vec{B}| = |F_{\vec{P}}|$. 

It therefore suffices to prove that for any projection $\vec{P}$, $|F_{\vec{P}}| \leq 2^{mn - d'n^2}$. 
Say that $\vec{P}$ is a projection onto the coordinates $\vec e_{l+k+1}, \dots, \vec e_n, \vec f_{l+1}, \dots, \vec f_n$. 
In particular, $m = (n-l) + (n-l-k) = 2n - 2l - k$. 
By our bound on $m$, it follows that $2n - 2l - k \geq c n$, i.e. \begin{align}
\label{eq:m_n_k_l_bound}
    2l + k \leq (2-c)n .
\end{align} 
Let $\vec{M}$ be a uniformly random element of $\mathbb{Z}_2^{m \times n}$.
We will consider the conditions under which there exists $\vec{N} \in \CS$ for which $\vec{M} = \vec{PN}$. 
We may decompose $\vec{N}$ as
\begin{equation}
\vec{N} = \begin{bmatrix}\vec{B}_1 \\ \vec{B}_2 \\ \vec{B}_3 \\ \vec{C}_1 \\ \vec{C}_2 \\ \vec{C}_3,\end{bmatrix}
\end{equation}
where $\vec{B}_1, \vec{C}_1 \in \mathbb{Z}_2^{l \times n}$, $\vec{B}_2, \vec{C}_2 \in \mathbb{Z}_2^{k \times n}$, and $\vec{B}_3, \vec{C}_3 \in \mathbb{Z}_2^{(n-l-k) \times n}$.
The only condition on $\vec{N}$ is that its columns are symplectically orthogonal, which translates to the requirement that $\vec D := \vec{C}_1^\intercal \vec{B}_1 + \vec{C}_2^\intercal \vec{B}_2 + \vec{C}_3^\intercal \vec{B}_3$ is symmetric. 
Then, 
\begin{equation}
\vec{M} = \vec{PN} = \begin{bmatrix}\vec{B}_3 \\ \vec{C}_2 \\ \vec{C}_3\end{bmatrix}. 
\end{equation}
Hence, $\vec{M} \in F_{\vec{P}}$ if and only if there is some choice of $\vec{B}_1, \vec{B}_2, \vec{C}_1$ for which $\vec{D}$ is symmetric. 
Let $\vec{R}^\intercal \in \mathbb{Z}_2^{n \times n}$ be an invertible matrix chosen so that $\operatorname{im}(\vec R^\intercal\vec C_2^\intercal)$ is contained in the last $k$ coordinate vectors. 
This is possible, since $\vec C_2^\intercal \in \mathbb{Z}_2^{n \times k}$ has rank at most $k$. 
Then, define $\vec D' := \vec{R^\intercal D R} = \vec C_1'^\intercal\vec B_1' + \vec C_2'^\intercal \vec B_2' + \vec C_3'^\intercal \vec B_3'$, where $\vec C_i' =\vec C_i \vec R$ and $ \vec B_i' = \vec B_i \vec R$. 
$\vec D'$ is symmetric if and only if $\vec D$ is symmetric, and so we need to select $\vec B_1', \vec B_2'$, and $\vec C_1'$ for which $\vec D'$ is symmetric. 

By construction of $\vec R^\intercal$, $\vec{C}_2'^\intercal \vec{B}_2' \in \Z_2^{n \times n}$ is only nonzero in the last $k$ rows, since $\operatorname{im}(\vec{C}_2'^\intercal \vec{B}_2') \subseteq \operatorname{im}(\vec R^{\intercal}\vec C_2^{\intercal})$. 
It immediately follows that the top $(n-k) \times (n-k)$ block of $\vec C_1'^\intercal\vec B_1' + \vec C_3'^\intercal \vec B_3'$ must be symmetric.
Equivalently, defining $\widetilde{\vec C}_i$ to be the first $n-k$ columns of $\vec C_i'$ and $\widetilde{\vec B}_i$ as the first $n-k$ columns of $\vec B_i'$, 
\begin{equation}
(\widetilde{\vec C}_1^\intercal\widetilde{\vec B}_1 + \widetilde{\vec B}_1^\intercal\widetilde{\vec C}_1 )+(\widetilde{\vec C}_3^\intercal\widetilde{\vec B}_3 + \widetilde{\vec B}_3^\intercal\widetilde{\vec C}_3 ) = \vec 0.
\end{equation}
Notice that for arbitrary $\widetilde{\vec{C}}_1, \widetilde{\vec{B}}_1$, $\widetilde{\vec C}_1^\intercal\widetilde{\vec B}_1 + \widetilde{\vec B}_1^\intercal\widetilde{\vec C}_1$ is the sum of two matrices with rank at most $l$, and hence must have rank at most $2l$. 

Meanwhile, since $\vec{M}$ is uniformly random, $\widetilde{\vec B}_3$ and $\widetilde{\vec{C}}_3$ are uniformly random as well.
Then we claim that $\widetilde{\vec B}_3, \widetilde{\vec{C}}_3 \in \Z_2^{(n-l-k) \times (n-k)}$ satisfy the conditions of \Cref{lemma:random_product_rank}.
The relevant condition to establish is that 
\begin{equation}
\a(n-k) \leq n-l-k \leq (n-k),
\end{equation}
for some $\a \in (1/2, 1)$. 
The upper bound on $n-l-k$ is immediate. 
Meanwhile, by Eqn.~\eqref{eq:m_n_k_l_bound}, we write that $l \leq \frac{(2-c)n - k}{2} \leq \frac{(2-c)(n-k)}{2}$, since $2-c \leq 1$.
Thus,
\begin{align}
n - l - k &\geq  (n-k) - (n-k)\frac{(2-c)}{2} = \left(\frac{c}{2}\right) (n-k).
\end{align}
By assumption, $c > 1$, which gives the desired lower bound with $\a := \frac{c}{2}$.

Now, applying the lemma with $\epsilon := c-1$, $\widetilde{\vec C}_3^\intercal\widetilde{\vec B}_3 + \widetilde{\vec B}_3^\intercal\widetilde{\vec C}_3$ has rank strictly larger than $(1 - \epsilon)(n-k) = (2-c)(n-k)$ with probability at least $1 - 2^{-\delta (n-k)^2}$, for some $\delta>0$.
Using Eqn.~\eqref{eq:m_n_k_l_bound}, we see that $k \leq (2-c)n$ or $n-k \geq (c-1)n$, implying that $2^{-\delta (n-k)^2} \leq 2^{-\delta (c-1)^2n^2}$. 
Hence this rank bound holds with probability at least $1 - 2^{-\delta (c-1)^2n^2}$. 
Furthermore, 
\begin{align}
(2-c)(n-k)& \geq (2-c)n - k \geq 2l.
\end{align}
Consequently, the rank of
\begin{equation}
\widetilde{\vec C}_1^\intercal\widetilde{\vec B}_1 + \widetilde{\vec B}_1^\intercal\widetilde{\vec C}_1,
\end{equation}
for \emph{any} choice of $\widetilde{\vec{C}}_1, \widetilde{\vec{B}}_1$ (which is $\leq 2l$) is strictly less than that of 
\begin{equation}
\widetilde{\vec C}_3^\intercal\widetilde{\vec B}_3 + \widetilde{\vec B}_3^\intercal\widetilde{\vec C}_3
\end{equation}
with probability at least $1 - 2^{-\delta(c-1)^2 n^2}$. 
Conditioned on this event, \begin{align}
    (\widetilde{\vec C}_1^\intercal\widetilde{\vec B}_1 + \widetilde{\vec B}_1^\intercal\widetilde{\vec C}_1) + (\widetilde{\vec C}_3^\intercal\widetilde{\vec B}_3 + \widetilde{\vec B}_3^\intercal\widetilde{\vec C}_3) \neq \vec 0 ,
\end{align}
since the first term must have the same rank of the second term to cancel it out.
By the previous argument, therefore, $\vec M \notin F_{\vec P}$.
Now, set $d' = \delta(c-1)^2$.
Then at most a $2^{-d' n^2}$ fraction of matrices $\vec{M} \in \mathbb{Z}_2^{m \times n}$ lie in $F_\vec{P}$. 
Thus, $|F_\vec{P}| \leq 2^{mn - d'n^2}$.
It follows that $H(\vec{BA}) \leq mn - d'n^2$.

To complete the proof, we observe that $H(\vec A) \leq 2n^2$, so by the data processing inequality, $H(\vec{BA}) \leq 2n^2$.
Then, the bound \begin{align}
    H(\vec{BA}) \leq \left(1 - \frac{\min(d', 1)}{3} \right) m n
\end{align} 
always holds, since when $m \leq 3n$, $mn - d'n^2 \leq mn - \frac{d'}{3}mn$ and when $m \geq 3n$, $2n^2 \leq mn - \frac{1}{3}mn$. 
Setting $d = \frac{\min(d', 1)}{3}$ completes the proof. 
\end{proof}

Before proving \Cref{thm:symplpn_error_blowup}, we record a simple result about the distribution of $\vec{b} \cdot \vec{e}$, where $\vec{b}$ is any fixed vector and $\vec{e} \sim \mathcal{D}^{\otimes n}_p$. 
\begin{lemma}\label{lem:depolarizing_weight}
Let $\vec{b} \in \Z_2^{2n}$ be a fixed vector and $\vec{e} \sim \mathcal{D}^{\otimes n}_p$ for $p \leq \frac{3}{4}$. 
Then $\vec{b} \cdot \vec{e} \sim \Ber(q)$, where $\eta(|\vec{b}|, p) \leq q \leq \frac{1}{2}$ for
\begin{equation}
\eta(w, p) = \frac{1 - \left(1 - \frac{4}{3}p\right)^{\frac{w}{2}}}{2}.
\end{equation}
Note that $\eta(w, p)$ is non-negative, and is increasing in both $w$ and $p$.
\end{lemma}
\begin{proof}
Since $\vec e \sim \mathcal{D}_p^{\otimes n}$,
\begin{align}
    \vec{b} \cdot \vec e \;=\; \sum_{k\in T} e_k,
\end{align}
where $T$ is the set of indices where $\vec{b}$ is nonzero. For any $k\in[n]$, let
\begin{align}
    e'_k \;=\; b_ke_k \;+\; b_{n+k}e_{n+k}.
\end{align}
By direct calculation,
\begin{align}
    e_k' \sim \begin{cases}
        \Ber(0) & k, n+k \notin T , \\
        \Ber(\frac{2}{3} p) & \text{else} .
    \end{cases}
\end{align}
For $S=\{\,k \leq n \mid k\in T \text{ or } k+n\in T\,\}$, we may rewrite
\begin{align}
    \sum_{k\in T} e_k =\sum_{k\in S} e'_k.
\end{align}
These $e'_k$ are independent Bernoulli random variables with probability $q=\tfrac{2}{3}p$, and note that
\begin{align}
    |S| \;\ge\; \frac{|\vec b|}{2}.
\end{align}
Using the fact that the probability that a $\textsf{Binomial}(m, q)$ random variable is odd with probability $\frac{1}{2} - \frac{1}{2} (1 - 2q)^m$,
\begin{equation}
 \Pr\left [\sum_{k \in S} e_k' = 1 \right] = \frac{1 - \left(1 - \frac{4}{3}p\right)^{|S|}}{2},
\end{equation}
which satisfies
\begin{equation}
\frac{1 - \left(1 - \frac{4}{3}p\right)^{|\vec b|/2}}{2} \leq  \frac{1 - \left(1 - \frac{4}{3}p\right)^{|S|}}{2} \leq \frac{1}{2}
\end{equation}
as desired. 
\end{proof}

\begin{proof}[Proof of \Cref{thm:symplpn_error_blowup}]
Let \begin{align}
    \mu(\vec{B}_0) := \underset{{\vec{e}\sim \mathcal{D}_p^{\otimes n}}}{\mathbb{E}} [|\vec{B}_0\vec{e}|],
\end{align}
for any $\vec{B}_0 \in \Z_2^{m \times 2n}$. 
Let $\vec{B}$ be any random variable such that $\vec{BA}$ has negligible total variation distance from a uniformly random matrix. 
Recalling that $r := \frac{n}{m}$, let $A$ be the event that $\mu(\vec{B}) < (\frac{1-r-\delta}{2})m$. 
We will show that $\Pr[A] = \negl(n)$. 
When $\Pr[A]=0$ there is nothing to prove; otherwise $\delta<1-r$.
We claim that it suffices to show that \begin{align}
    H(\vec{BA} | A) = mn - \Omega(mn) .
    \label{eq:fractionally_bounded_entropy}
\end{align}
Suppose for contradiction that $\Pr[A] \geq \frac{1}{\poly(n)}$ for infinitely many $n$, but Eqn.~\eqref{eq:fractionally_bounded_entropy} holds. We restrict to this subsequence.
Then we may describe the probability distribution $h:\Z_2^{m \times n} \to [0, 1]$ of $\vec{BA} \in \Z_2^{m \times n}$, as 
\begin{equation}
h = p f + (1-p)g,
\end{equation}
where $p := \Pr[A] \geq \frac{1}{\poly(n)}$ and $H(f) = mn - \Omega(mn)$. 
From this decomposition, there is a resulting bound on entropy of 
\begin{align}
H(h) &\leq pH(f) + (1-p)H(g) + H_2(p)\\
&\leq p(mn - \Omega(mn)) + (1-p)mn + O(p\log(2/p))\\
&= mn - \Omega(pmn) + O(p\log(n))\\
&\leq mn - \frac{1}{\poly(n)},
\end{align}
where in the second inequality we use the fact that $H_2(p) = O(p\log(2/p))$ for $0<p\leq 1$. 

However, $h$ is the probability distribution of $\vec{BA}$, and since $mn=\poly(n)$, negligible total variation distance from uniform implies a negligible difference in entropy.
Since $h$ has negligible distance from uniform by assumption, and the entropy of the uniform distribution is $mn$, this is a contradiction. 

We will now bound the entropy $H(\vec{BA} | A)$. Define a random variable $\vec \Pi \in \Z_2^{m \times m}$, which is the permutation matrix such that the rows of $\vec{B'} := \vec{\Pi B}$ are sorted in descending order by weight. 
Note that $H(\vec{\Pi B A|A, \Pi}) = H(\vec{B A|A, \Pi})$ since $\vec \Pi$ is a permutation, so that
\begin{align}
H(\vec{\Pi B A}|A) &\geq H(\vec{\Pi B A}|A, \vec{\Pi}) = H(\vec{B A}|A, \vec{\Pi})\\
&\geq H(\vec{B A}|A)-H(\vec{\Pi}|A).
\end{align}
Thus, 
\begin{align}
    H(\vec{BA} | A)&\leq H(\vec{\Pi BA} | A)+ H(\vec{\Pi}|A)\\
    &\leq H(\vec{B'A} | A)+ m\log(m).
\end{align}
Since $m = \poly(n)$, $m\log(m) = o(mn)$, it therefore suffices to show that $H(\vec{B'A} | A) = mn - \Omega(mn)$. 

Let $\vec{b}'_i$ denote the $i$th row of $\vec{B'}$. 
Then, by linearity of expectation and \Cref{lem:depolarizing_weight}, 
\begin{align}
\mu(\vec{B}') &= \sum_i \underset{\vec e \sim \CD_p^{\otimes n}}{\mathbb{E}} [\vec{b}'_i \cdot \vec{e}] \geq \sum_i \eta(|\vec{b}'_i|, p)
\end{align}
where we recall that $\eta(w, p) = \frac{1}{2}\left[1-\left(1-\frac{4}{3}p\right)^{\frac{w}{2}}\right]$. 
Because $\vec{B'e}$ and $\vec{Be}$ have exactly the same weight, 
\begin{align}
\mu(\vec{B}) &= \mu(\vec{B'}) . 
\end{align}
Hence, the event $A$ is also the event that $\mu(\vec{B'}) < (\frac{1-r- \delta}{2})m$. 
% Next, $\Pr[B^c] \leq \Pr[A^c]$ (equivalently, $\Pr[A] \leq \Pr[B]$), where $B^c$ is the event that 
% \begin{equation}
% \sum_i \eta(|\vec{b}'_i|, p) \geq \left(\frac{1-r- \delta}{2}\right)m. 
% \end{equation}
The use of $\vec \Pi$ to sort $\vec B$ into $\vec B'$ is helpful because it will allow us to split the entropy analysis into two cases: for the first few rows of $\vec B'$, the entropy is large and we will not be able to give a non-trivial bound, but for all remaining rows, we will show that their weight is relatively small, and thus have a much lower entropy.
Define this ``cutoff'' row to be $k = \lceil m(1-r - \frac{\delta}{2}) \rceil$.
We claim that for any $d_0 > 0$, for sufficiently large $n$, $\Pr_{\vec B}[|\vec{b}'_k| \geq d_0 n \,|\, A] = 0$. 
This is because if $|\vec{b}'_k| \geq d_0 n$, then 
\begin{align}
\sum_{i} \eta(|\vec{b}'_i|, p) &\geq \sum_{i = 1}^k \eta(|\vec{b}'_k|, p) \geq k\eta(d_0n, p)\\
&= k \left(\frac{1}{2}-\frac{\left(1-\frac43p\right)^{d_0n/2}}{2}\right) = \frac{k}{2} - o(k) .
\end{align}
Note that $\frac{k}{2} \geq \left(\frac{1-r }{2}- \frac{\delta}{4}\right)m$, so this immediately implies that for sufficiently large $n$, $\mu(\vec B) \geq \frac{k}{2} - o(k)> \left(\frac{1-r - \delta}{2}\right)m$, contradicting the conditioned event $A$. 
Hence, for sufficiently large $n$, $|\vec{b}'_k| < d_0n$, and by construction $|\vec{b}'_j| < d_0n$ for all $j \geq k$.
Now, we consider the $m-k$ bottom rows, which all satisfy this weight bound. 
\begin{align}
m-k &\geq m\left(r+\frac{\delta}{2}\right) - 1 \gtrsim m\left(r+\frac{\delta}{4}\right) \geq n\left(1 + \frac{c\delta}{4}\right),
\end{align}
where $\gtrsim$ indicates that inequality holds for sufficiently large $n$.
If $\vec{B}'_1$ denotes the first $k$ rows of $\vec{B}'$ and $\vec{B}'_2$ denotes the remaining $m-k$ rows, then by the support-size bound in the proof of \cref{thm:no_randomized_symplpn}, 
\begin{align}
H(\vec{B}'_2\vec A|A) & \leq H((\vec B'_2, \vec B'_2 \vec A) |A) = H(\vec B_2' \vec A | \vec B_2', A) + H(\vec B_2' | A) \\
& \leq (m-k)n - d(m-k)n + H(\vec{B}_2'|A)\\
&\leq (m-k)n - d\left(r + \frac{\delta}{4}\right)mn + H(\vec{B}_2'|A)
\end{align}
for some $d > 0$, that depends on the constant $\frac{c\delta}{4}$. 

Each row of $\vec{B}_2'$ has weight at most $d_0 n$.
Thus, the entropy of each row is at most the entropy of a uniformly random length-$2n$ bitstring of weight at most $\lfloor d_0n\rfloor$, which is $\log\sum_{j=0}^{\lfloor d_0n\rfloor}\binom{2n}{j}\leq 2nH_2(d_0/2)$; the entropy overall of $\vec{B}_2'$ is at most the sum of the entropies of its rows, so \begin{align}
    H(\vec B_2'\mid A)\leq(m-k)\log\!\left(\sum_{j=0}^{\lfloor d_0n\rfloor}\binom{2n}{j}\right)\leq2mnH_2(d_0/2).
\end{align}
Choose $0<d_0<1$ sufficiently small so that $H_2(d_0/2)\leq d'/4$, where $d':=d\delta/4$.
Then
\begin{equation}
H(\vec{B}_2'|A) \leq \frac{d'}{2}mn.
\end{equation}
As a consequence, we obtain that $H(\vec{B}'_2\vec A|A) \leq (m-k)n - \frac{d'}{2}mn$.
Meanwhile, by a trivial dimension bound, $H(\vec{B}'_1\vec{A} | A) \leq kn$.
Together, these bounds imply that
\begin{align}
H(\vec{B'A} | A) &\leq H(\vec{B}'_1\vec{A} | A) + H(\vec{B}'_2\vec{A} | A)\\
&\leq mn - \frac{d'}{2}mn. 
\end{align}
Hence, as desired, we have that $H(\vec{B'A} | A) = mn - \Omega(mn)$, and as a result $\Pr[A] = \negl(n)$. 
\end{proof}

\end{document}